\pdfoutput=1
\documentclass[a4paper,11pt]{book}

\usepackage[T1]{fontenc}
\usepackage[utf8]{inputenc}
\usepackage[french,german,english]{babel}

\usepackage{setspace} 

\usepackage{graphicx}
\usepackage[table]{xcolor}
\graphicspath{{images/}}

\usepackage{subfigure}
\usepackage{booktabs}
\usepackage{microtype}
\usepackage{url}
\usepackage[final]{pdfpages}

\usepackage{fancyhdr}

\fancypagestyle{plain}{
	\fancyhf{}

	\fancyfoot[EL,OR]{\thepage}}
\fancypagestyle{addpagenumbersforpdfimports}{
	\fancyhead{}
	
	\fancyfoot{}
	\fancyfoot[RO,LE]{\thepage}
}

\usepackage{listings}
\usepackage{hyperref}
\hypersetup{pdfborder={0 0 0},
	colorlinks=true,
	linkcolor=black,
	citecolor=black,
	urlcolor=black}
\makeatletter
\def\cleardoublepage{\clearpage\if@twoside \ifodd\c@page\else
    \hbox{}
    \thispagestyle{empty}
    \newpage
    \if@twocolumn\hbox{}\newpage\fi\fi\fi}
\makeatother \clearpage{\pagestyle{plain}\cleardoublepage}

\usepackage{color}
\usepackage{tikz}
\usepackage[explicit]{titlesec}

\usepackage[explicit]{titlesec}
\titleformat{\chapter}[hang]
        {\normalfont\bfseries\Huge}
        {\makebox[.5ex][r]{\colorbox{black}{%
            \hspace*{5cm}\rule[-1.5mm]{0pt}{13mm}\color{white}\thechapter\,%
          }}\,}
        {0pt}{#1}

\titlespacing*{\chapter}{0pt}{50pt}{30pt}
\titlespacing*{\section}{0pt}{13.2pt}{*0}
\titlespacing*{\subsection}{0pt}{13.2pt}{*0}
\titlespacing*{\subsubsection}{0pt}{13.2pt}{*0}

\newcounter{myparts}
\newcommand*\partlabel{}
\titleformat{\part}[display]
	{\normalfont\bfseries\Huge}
	{\gdef\partlabel{\thepart\ }}
 	{0pt}
 	  {\setlength{\unitlength}{20mm}
	  \addtocounter{myparts}{1}
	  \begin{tikzpicture}[remember picture,overlay]
    \node[anchor=north west,xshift=-65mm,yshift=-6.9cm-\value{myparts}*20mm] at (current page.north east)
      {\begin{tikzpicture}[remember picture, overlay]
        \draw[fill=black] (0,0) rectangle(62mm,20mm);
        \node[anchor=north west,yshift=-6.1cm-\value{myparts}*20mm,xshift=-60.5mm,minimum height=30mm,inner sep=0mm] at (current page.north east)
        {\parbox[top][30mm][t]{55mm}{\raggedright \color{white}Part \partlabel $\phantom{\textrm{l}}$}};
        \node[anchor=north east,yshift=-6.1cm-\value{myparts}*20mm,xshift=-63.5mm,text width=\textwidth,minimum height=30mm,inner sep=0mm] at (current page.north east)
              {\parbox[top][30mm][t]{\textwidth}{\raggedleft \color{black}#1}};
       \end{tikzpicture}
      };
   \end{tikzpicture}
   \gdef\partlabel{}
  }

\usepackage{amsmath,amsfonts,amsthm}

\newtheorem{theorem}{Theorem}[chapter]

\makeatletter
\def\resetMathstrut@{%
  \setbox\z@\hbox{%
    \mathchardef\@tempa\mathcode`\(\relax
      \def\@tempb##1"##2##3{\the\textfont"##3\char"}%
      \expandafter\@tempb\meaning\@tempa \relax
  }%
  \ht\Mathstrutbox@1.2\ht\z@ \dp\Mathstrutbox@1.2\dp\z@
}
\makeatother

\usepackage{multirow}
\usepackage{footnote}
\usepackage{amsmath}
\usepackage{amsfonts}
\usepackage[ruled,vlined]{algorithm2e}
\usepackage{amsthm}
\newtheorem{remark}{Remark}[chapter]
\usepackage[outdir=./]{epstopdf}
\newcommand{\norm}[1]{\lVert#1\rVert}

\begin{document}
\frontmatter
\begin{titlepage}

\thispagestyle{plain}
\begin{center}
    \Large
    \textbf{Decentralised Resource Allocation and Coordination for\\ 5G Cellular Communication Networks}
       
    \vspace{0.4cm}
    Gabriel Otero P\'erez, Manuel Fern\'{a}ndez Veiga\let\thefootnote\relax\footnote{G.\,O.\,P\'erez is affiliated with the Atlantic Research Center for Information and Communication Technologies (AtlantTIC), University of Vigo, Spain(e-mail:\,gabrieloteroperez@gmail.com). M.\,F.\,Veiga is affiliated with the Department of Telematics Engineering, University of Vigo, Spain (e-mail:\,mveiga@det.uvigo.es)}

    \vspace{0.9cm}
    \textbf{Abstract}
\end{center}
In order to cope with the ever increasing traffic load that networks will need to support, a new approach for planning cellular networks deployments should be followed. Traditionally, cell association and resource allocation has been based on the received signal power but this approach seems to be inadequate regarding the brewing of heterogeneous networks. In this work, we first implement a network simulator in order to test new cell associacion and resource allocation techniques. Then, we pose the network utility maximisation problem, reformulating the Downlink and Uplink Decoupling (DUDe) scheme under the framework and tools of mathematical optimisation. We derive the explicit solution of the problem under fixed and non-fixed association policy so as to propose and develope both centralised and decentralised algorithms capable of solving cell association and resource allocation problems. We observe that the decentralised approach requires low computational effort and represents a significant gain in the overall performance of the network.

\vspace{0.9cm}

\end{titlepage}

\let\cleardoublepage\clearpage

\setlength{\parskip}{1em}

\mainmatter
\chapter{Introduction}

\section{Motivation and related work}

 According to the \emph{Visual Networking Index} (\textsc{vni}) Global Mobile Data Forecast \cite{cisco} released by Cisco in February 2016, there will be 5.5 billion global mobile users by 2020. In addition, networks will face 11.6 billion mobile-ready devices and connections, nearly 4 billion more than in 2015. The average mobile connection speed will increase 3.2-fold, from 2.0 Mbps in 2015 to 6.5 Mbps by 2020 and global mobile IP traffic will reach an annual run rate of 367 Exabytes, up from 44 Exabytes in 2015. Furthermore, this raising numbers are not confined to the downlink plane. There has been a growth of the uplink importance due to the sensor networks and \emph{Machine2Machine} communications. This context favours the brewing of new cutting-edge technologies such as 5\textsc{g}.

5\textsc{g} is not meant to be an incremental advance on 4\textsc{g} but a complete paradigm shift. Indeed, it will represent higher carrier frequencies, moving towards and into millimeter wave spectrum. Just a few years ago this scenario seemed unthinkable but nowadays, this is possible because new semiconductor technologies and short-range standards are maturing~\cite{mature1}-\cite{mature3}. Also, the costs and power consumption of mobile devices are rapidly falling. Massive bandwidths and higher aggregate data rates will turn out to be essential. The amount of data that networks must be able to handle will need to increase, roughly, by 1000 times from 4\textsc{g} to 5\textsc{g}~\cite{whatwillbe}. To that end, many changes will appear on the physical layer such as massive \textsc{mimo} technologies which blossomed in the late $1990$s~\cite{mimo1},\cite{mimo2}. \textsc{mimo} leverages the spatial dimention of communications taking advantage of multipath propagation and achieving enormous enhancement in spectral efficiency. Wireless cellular networks are evolving towars heterogeneity. Heterogeneous cellular networks comprise traditional cellular networks overlaid with smaller base stations which work using a lower transmission power~\cite{hetnets}.

Focusing on the network layer, single macrocells will need to support both high-rate and low-rate devices. This will require large-scale changes to the control plane. Centralised solutions do not seem to be appropriate anymore for such a potentially large subscriber base. Additionally, granting the same treatment to each type of device might not be the best option. A simple, albeit effective way to increase the network capacity is to increase the density of small cells which live together with the old macrocells and therefore, making the coverage areas smaller. This adds significant additional complexity to the design and deployment of new networks, since now the decision of which base station should serve a given user in downlink and which one in uplink is not trivial. Consequently, there exists much room for improvement and optimising the associations between users and base stations in 5\textsc{g}. To fully support all the above-mentioned features, the network will have to meet higher levels of intelligence.

Among the related literature, it is worth highlighting several approaches which have been explored recently. Smiljkovikj \emph{et al.} \cite{capacity},\cite{capacity2} study the downlink-uplink decoupled access using the framework of stochastic geometry and probability theory in order to derive the association probabilities. Architectural changes needed to facilitate the decopling are also outlined in their work. In \cite{utilitymax}, Wildman and Weber study the network utility maximisation problem under single station association policies assuming only one link per user (\emph{downlink}). Their results include solutions based on greedy rounding of multi-station associations and association heuristics. Last but not least, Palomar \emph{et al.} \cite{Palomar06} address the understanding of decomposition methods applied to network utility maximisation. They review the basics of convexity, Lagrange duality, distributed subgradient method and Jacobi and Gauss-Seidel iterations.

\section{Contributions}
The main goal of this work is to explore some alternatives that might help solve the actual challenges in the network control plane. Namely, we pay attention to the user association and the resource allocation process.

First, we present the network model which we will rely on to conduct our study. We describe the heterogeneous network paradigm, explaining the system model along with the association scheme that this approach suggests, that is, the \emph{Downlink and Uplink decoupling} (DUDe). We make use of the \emph{Poisson Point Processes} (\textsc{ppp}) in order to model the locations of the users and the base stations. These represent the trending alternative contrary to the traditional hexagonal grid deployments. Next, we implement a network simulator based on a 2-tier heterogeneous network and derive some performance measurements such as: distribution of the distance to the serving station, average signal-to-noise ratio, association probabilities, aggregate throughput, etc. After assessing the validity of the simulator, we pose the network utility maximisation problem, reformulating the DUDe scheme under the framework and tools of mathematical optimisation. We derive the explicit solution of the problem under fixed association policy and we validate the resulting centralised algorithm after, using the implemented network simulator. Again, we include different performance indicators so as to compare the benefits of this alternative in contrast to the original DUDe scheme. We pay special attention to aggregate spectral efficiency and uplink-downlink rate asymmetry.

Finally, we propose and develope a decentralised algorithm capable of solving, jointly, the cell association and the resource allocation problems. To that end, a full dual decomposition of the problem is performed, followed by the implementation and testing. Two kinds of tests have been carried out. In the first one,  with the aim of presenting the characteristics and strengths of the algorithm in a more suitable and friendly way, we test this approach in a custom deployment with a few base stations and users, chosen manually. Afterwards, once the validity and convergence of the decenstralised algorithm has been tested, we integrate it into the simulation tool. To conclude, we conduct a comparison between the three main options studied throughout this work.
\chapter{Downlink and Uplink Decoupling (DUDe)}
\label{chap:dude}

\section{Network architecture}
In order to cope with the ever increasing traffic load that networks will need
to support, a new approach for planning cellular networks deployments should
be followed. One way to expand mobile network capacity is to enlarge the
number of base stations but this can only be performed to a certain extent due
to the fact that it necessitates a huge capital expenditure and because
finding new spots for stations is increasingly hard, particularly in big
cities. Exactly for that reason, nowadays \textbf{heterogeneous networks} are
becomimg more and more important.

\subsection{Heterogeneous cellular networks}

A heterogeneous cellular network often implies the use of multiple radio
access technologies, each one supported by different base stations. Namely, a
wide area network might use \emph{macrocells}, \emph{small cells} and/or
\emph{femtocells} to offer seamless coverage. The last two of them are just
low-power and low-cost access nodes (base stations). We define a
\textbf{multi-tier} network as a network in which both traditional cellular
network (macrocells) and small cell network coexist. Each one of them
constitutes a network tier.

\begin{figure}[!htb]
\begin{center}
  \includegraphics[scale=0.75]{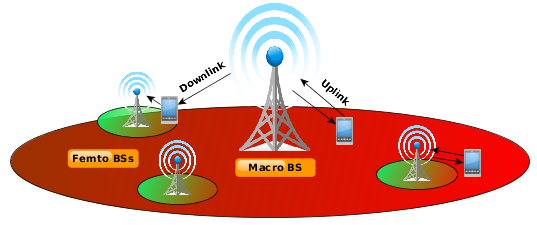}
  \caption{Two tier heterogeneous network.}
  \label{fig:5g}
\end{center}
\end{figure}

The proposed network architecture is a two-tier heterogeneous cellular
network. We take into account two levels which correspond to each one of the
two types of base stations we are going to model, \emph{macrocells} and
\emph{femtocells}. Therefore, we have the following components:

\begin{itemize}
\item Two kinds of Base Stations (\textsc{bs}), each one belonging to a
  specific tier.
\item Macro Base Stations (\textsc{mbs}s) for tier 1.
\item Femto Base Stations (\textsc{fbs}s) for tier 2.
\item User devices, which are inherently mobile.
\end{itemize}

\subsection{System model}
\label{sec:system_model}

In this section, we provide some preliminaries on the implemented model
required to support the discussion presented later.

\subsubsection*{Network components distribution}
The location of each one of the aforementioned network components is computed
using independent \textsc{ppp} (\emph{Poisson Point Process}) for every tier.
 
Def (\textsc{ppp}): Let A $\subseteq \mathbb{R}^2$. $\Phi$ is a \textsc{ppp}
on A if:
\begin{enumerate}
\item The number of points in B $\subseteq $ A is Poisson distributed with
  rate $\lambda$, per unit area.
\item $\Phi(B_1), \,\, \Phi(B_2)$ are statistically independent if
  $B_1,\,\, B_2$ are disjoint.
\end{enumerate} 
That is, the base stations in the $i$-th tier are spatially distributed as a
\textsc{ppp} $\Phi_i$ of density $\lambda_i$. Ramdom models based
on stochastic geometry have shown their accuracy to model real-world network
deployments \cite{ktier}-\cite{stochastic}, and usually they work better than grid-based models.

Similarly, user devices location is also modeled by an independent
\textsc{ppp} $\Phi_d$ of density $\lambda_d$.

\subsubsection*{Channel model}

In wireless environments, the simplest method of relating the transmitted and
the received signal power is to state that the received signal power is
proportional to the distance between the transmitter and the receiver raised
to a certain exponent, $\alpha$, which models the path-loss. Typical values
for the path-loss exponent are $\alpha = 2$ for free-space and $\alpha = 4$
for a path model of an urban radio channel.

In addition, we use $P_i$ to denote the \textbf{transmission power} of a node
at tier $ i \in \lbrace \textsc{m},\textsc{f}\rbrace$, where $i = \textsc{m}$
for \textsc{mbs}s and $i = \textsc{f}$ for \textsc{fbs}s. Every base station
in the same tier uses the same transmission power. The transmission power for
each user device is $P_d$.\textbf{ Noise} is additive, Gaussian and has
constant power $\sigma^2$. \textbf{Rayleigh fading} is used to model the
channel quality fluctuations between the \textsc{bs} and the mobile
device. $h_x \sim exp(1)$ describes the Rayleigh fading and it is an
exponentially distributed random variable with unit mean.
 
Hence, the received power (downlink) at a typical user device located at $y$
from a \textsc{bs} located at $x_0$ is
\begin{equation}
  \label{eq:1}
  P_r^{DL} = P_i \,h_{x_0} \, || x_0 - y ||^{-\alpha}
\end{equation}
where $||\cdot||$ is the Euclidean norm. Similarly, the signal power received
at the \textsc{bs} in the uplink is given by:
\begin{equation}
  \label{eq:2}
  P_r^{UL} = P_d \,h_{x_0} \, || x_0 - y ||^{-\alpha}.
\end{equation}
 
\subsubsection*{Interference model}

The interference power depends on medium access control protocol and the
network characteristics (e.g., network topology, association criterion, etc.)

We should note that no intra-cell interference is considered in this
implementation. To that end, both orthogonal (e.g., \textsc{tdma},
\textsc{ofdma}, etc) or non-orthogonal (e.g., \textsc{cdma}) multiple access
methods should be employed. In other words, we assume that each base station
avoids the interference among the devices associated to it by orthogonal
resource allocation. Conversely, inter-cell interference is considered. There
exist several techniques to model the interference between base stations. Most
of them only take into account region bounds or $k$ nearest
interferers. Despite the fact that it is a very popular technique because of
its simplicity and accuracy, it seems obvious that we cannot neglect distant
interferers when the path-loss exponent is low, $\alpha < 4$. Thus, with the
aim of building the most versatile model (and deployment simulator) as
possible, we have chosen to consider every other base station in the interest
area (which is not being part of the uplink/downlink) as an interferer.

The resulting donwlink \textsc{sinr} expression assuming a user device at $y$
connects to a \textsc{bs} at $x_0$ is
\begin{equation}
\displaystyle
  \mathsf{SINR}_{DL}(y) = \frac{P_i \,h_{x_0} \, \norm{x_0 -
      y}^{-\alpha}}{\sum^k_{j=1}\,\, \sum_{x \in \Phi_i,\, x\neq
      x_0} P_j \,h_{x_0} \, \norm{x - y}^{-\alpha}\,\, +\,\, \sigma^2}
\end{equation}
where $j$ is the $j$-th tier, $\Phi_i$ is the Poisson Point Process which
models the location of tier $i$'s \textsc{BSs} and
$ i \in \lbrace \textsc{m},\textsc{f}\rbrace$.

The same reasoning applies to the uplink. This model assumes that all the
\textsc{bs}s are transmitting without interruption (all the time) and always
with the same power. If that is not the case, the analysis can be extended by
considering a fraction of the density of interfering \textsc{bs}s.

\subsection{Association scheme}
\label{sec:assoc_rules}

In the proposed model, we dismiss the traditional association scheme based on
the downlink received power (\textsc{rp}). While the downlink association will
still be based on downlink \textsc{rp}, this is no longer true for the
uplink. The latter is going to be based on pathloss, that is, in the uplink,
the device is associated to the \textsc{bs} to which it transmits with the
highest average power.

These two different assumptions for uplink and downlink association is what we
call \emph{\textbf{Downlink and Uplink Decoupling}}~\cite{refdude}. \textsc{dud}e leads to different coverage boundaries for
uplink and downlink, which will be discussed later on this document.

The association decision is now based on the average received signal in
\textsc{dl}/\textsc{ul} separately. The expectation is taken over the pdf of
the fading. We can obtain the signal powers for uplink and downlink by
averaging (\ref{eq:1}) and (\ref{eq:2}) as follows

\begin{equation}
  \mathbb{E}_h[P_r^{DL}] = \mathbb{E}_h[P_i \,h_{x_0} \, \norm{x_0 -
    y}^{-\alpha}] = \mathbb{E}_h[P_i \, \norm{x_0 - y}^{-\alpha}]\,
  \mathbb{E}_h[h_{x_0}]= P_i\,\norm{x_0 - y}^{-\alpha} 
\end{equation}
\begin{equation}
  \mathbb{E}_h[P_r^{UL}] = \mathbb{E}_h[P_d \,h_{x_0} \, \norm{x_0 -
    y}^{-\alpha}] = \mathbb{E}_h[P_d \, \norm{x_0 - y}^{-\alpha}]
  \,\mathbb{E}_h[h_{x_0}]= P_d \, \norm{x_0 - y}^{-\alpha}
\end{equation}
where $i \in \lbrace \textsc{m},\textsc{f}\rbrace$. Let $D_i$ be the distance
between the device and the serving base station, that is,
$D_i = \norm{x_0 - y}$ where $i \in \lbrace \textsc{m},\textsc{f}\rbrace
$.
Thus, following the aforementioned policy we can derive some association
rules:
\begin{itemize}
\item Connect to a \textbf{macrocell in downlink} if
  $P_\textsc{m}\, D_\textsc{m}^{-\alpha} > P_\textsc{f}\,
  D_\textsc{f}^{-\alpha}$. Otherwise, connect to a femtocell.
\item Associate to a \textbf{macrocell in uplink} if
  $P_\textsc{d}\, D_\textsc{m}^{-\alpha} > P_\textsc{d}\,
  D_\textsc{f}^{-\alpha} \longrightarrow D_\textsc{m}^{-\alpha} >
  D_\textsc{f}^{-\alpha} \longrightarrow D_\textsc{f}^{\alpha} >
  D_\textsc{m}^{\alpha}$ and connect to a femtocell otherwise.
\end{itemize}

As we can see, the distance to the serving base station is the parameter we
will measure in order to decide which is the appropriate base station to
associate to in the uplink.

In a two-tier heterogeneous network, these simple association rules lead us
four possible cases while choosing the base station for uplink and downlink.

\begin{enumerate}
\item \emph{Case 1}: associate to a macrocell in both donwlink and uplink.
  This happens whenever:
  \begin{equation}
    P_\textsc{m}\, D_\textsc{m}^{-\alpha} > P_\textsc{f}\,
    D_\textsc{f}^{-\alpha} \text{ and } D_\textsc{f}^{\alpha} >
    D_\textsc{m}^{\alpha} 
  \end{equation}

\item \emph{Case 2}: choose a macrocell in downlink and a femtocell in
  uplink. Two conditions may hold for this case:
  \begin{equation}
    P_\textsc{m}\, D_\textsc{m}^{-\alpha} > P_\textsc{f}\,
    D_\textsc{f}^{-\alpha} \text{ and } D_\textsc{f}^{\alpha} \leq
    D_\textsc{m}^{\alpha} 
  \end{equation}

\item \emph{Case 3}: associate to a femtocell in downlink and choose a
  macrocell in uplink. The intersection of the conditions that should hold for
  this to be true is an empty set:
  \begin{equation}
    P_\textsc{m}\, D_\textsc{m}^{-\alpha} \leq P_\textsc{f}\,
    D_\textsc{f}^{-\alpha} \text{ and } D_\textsc{f}^{\alpha} >
    D_\textsc{m}^{\alpha}
  \end{equation}

  Both conditions cannot hold at the same time since
  $\frac{P_\textsc{f}}{P_\textsc{m}} < 1$.

\item \emph{Case 4}: connect to a femtocell in both donwlink and uplink. The
  conditions are:
  \begin{equation}
    P_\textsc{m}\, D_\textsc{m}^{-\alpha} < P_\textsc{f}\,
    D_\textsc{f}^{-\alpha} \text{ and } D_\textsc{f}^{\alpha} \leq
    D_\textsc{m}^{\alpha} 
\end{equation}
\end{enumerate}

\section{Initial proof of concept}\label{sec:simulator_reference}

A 2-tier network simulator has been implemented to evaluate some performance
parameters such as: distribution of the distance to serving station, average signal-to-noise
ratio, user's throughput, etc.  In order to accomplish the empirical test of
the aforementioned network architecture, some issues need to be addressed.

\subsection{Simulator basics}
The first aspect we need to tackle is the construction of the coverage
maps. As we already stated above, our goal is to assess the impact of uplink
and downling decoupling on the overall system performace. To that end, we have
split the problem in two separate parts.

Firstly, we are going to describe how to build the \textbf{uplink} coverage
regions according to the random antenna deployment.

\subsubsection*{Uplink coverage. Voronoi diagrams}

In the uplink case, reacall that the device is associated to the base station
to which it transmits with the highest average power. Since only the distance
to the serving base station is going to be taken into account for the uplink,
the method for dividing up the interest area between the different base
stations (Femtocells, Macrocells) is the \textbf{Voronoi
  tesselation}~\cite{voronoi}.

Therefore, the uplink coverage map is basically a Voronoi diagram, that is, a
partitioning of the interest area into regions based on distance to the
available base stations. In simple words, it is a diagram created by taking
pairs of points that are close together and drawing a line that is equidistant
between them and perpendicular to the line connecting them.

As shown in figure~\ref{fig:voronoi}, all points on the blue lines in the
diagram are equidistant to the nearest two (or more) base stations.

\begin{figure}[!htb]
  \begin{center}
    \includegraphics[scale=0.75]{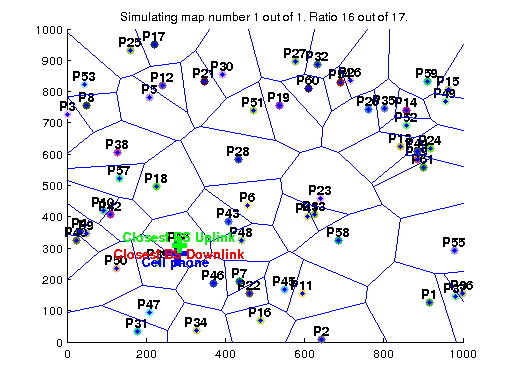}
    \caption{Voronoi diagram computation.}
    \label{fig:voronoi}
  \end{center}
\end{figure}

\subsubsection*{Downlink coverage. Multiplicatively weighted Voronoi diagrams.}

Regarding to the downlink, we need to be aware of the received power in
addition to the distance. In other words, the device is associated to the base
station from which it receives the highest average power. Unfortunaley, the
standard Voronoi diagram is no longer usefull if we expect to grasp both
received power and distance at the same time.

We need Voronoi cells to be defined in terms of a distance which is modified
by weights assigned to each generator point (base station), that is, we have
to change the Euclidean distance $d(x, p)$ for the weighted distance
$d_w(x, p)$. Now, the straight lines separating each region become circles. It
can be seen as if the straight lines shown before were just circles with an
infinite radius.

We can set out the main problem to solve, which is finding the points that are
equidistant to two given points, given the new distance definition. In order
to do that, solve:
\begin{itemize}
\item Let $\mathbb{X} = (x,y)$ be a point on a two dimensional grid which is
  equidistant to two given points ($P, Q$).
\item Let $P,\, Q$ be the two points under study. We want to compute the
  dominance area for each one of them.
\item Let $\displaystyle d_w(a,\, b) = \frac{|a-b|}{W_b}$ be the definition
  of the weighted distance between two given points.
\item Let $W_p$ and $W_q$ be the weigth factors for each one of the points.
\end{itemize}

Therefore, the problem can be stated as follows,
\begin{equation}
    d\,(\mathbb{X},\,P) = d\,(\mathbb{X},\,Q)
\end{equation}

\begin{center}
$\displaystyle \frac{|\mathbb{X} - P|}{W_p} = \frac{|\mathbb{X} - Q|}{W_q} \longrightarrow \frac{|\mathbb{X} - P|}{|\mathbb{X} - Q|}  = \frac{W_p}{W_q} = \lambda$
\end{center}

\begin{center}
$\displaystyle \frac{|(x, y) - (P_x, P_y)|}{|(x, y) - (Q_x, Q_y)|} = \lambda$
\end{center}

Solving the above, leads to the following circumference equation which is called \emph{Circle of Apollonious}.  For a more detailed proof, we refer the reader to
  Appendix~\ref{app:before}.

\begin{center}
$\displaystyle (x - \frac{P_x-Q_x\lambda^2}{1-\lambda^2})^2 + (y - \frac{P_y-Q_y\lambda^2}{1-\lambda^2})^2 = \frac{\lambda^2[(P_x-Q_x)^2+(P_y-Q_y)^2]}{(1-\lambda^2)^2}$
\end{center}

As we can see, \emph{Apollonious' circles} are closely related to \emph{Weighted Voronoi} diagrams. The points on the \emph{Apollonious} circumference are equidistant to $P$ and $Q$, considering weighted distances $W_i > W_j$ (see figure~\ref{fig:apollonious}). Thus, $P$ dominates $Q$ and $P$'s dominance area is the outer area (blue) and $Q$'s is the inner area of the circle.


\begin{figure}[!htb]
\begin{center}
    \includegraphics[scale=0.8]{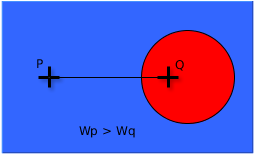}
   \caption{Apollonious circle for two points.}
   \label{fig:apollonious}
   \end{center}
\end{figure}


The complete coverage map is built in an iterative way. We compute each point's dominance area by intersecting every \emph{Apollonious circle} involving that point (base station) and any of the others (base stations). After carrying out the aforementioned procedure, we reach our final coverage map (see figure~\ref{fig:voronoi_weighted}). Figure~\ref{fig:superimposed} shows superimposed uplink and downlink maps.

It is worth saying that the realizations of the Poisson Point Processes are generated in the following manner:
\begin{enumerate}
\item Draw the number of points using a Poisson distribution with mean parameter $\lambda$.
\item Uniformly distribute the obtained number of points across across the 2-dimensional area independently for each dimension.
\end{enumerate}


\begin{figure}[!htb]
\begin{center}
    \includegraphics[scale=0.82]{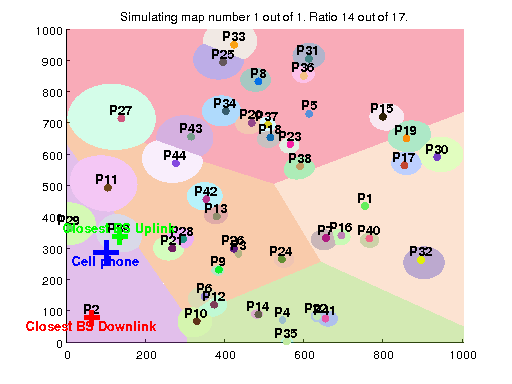}
   \caption{Downlink coverage map.}
   \label{fig:voronoi_weighted}
   \end{center}
\end{figure}



\begin{figure}[!htb]
\begin{center}
    \includegraphics[scale=0.78]{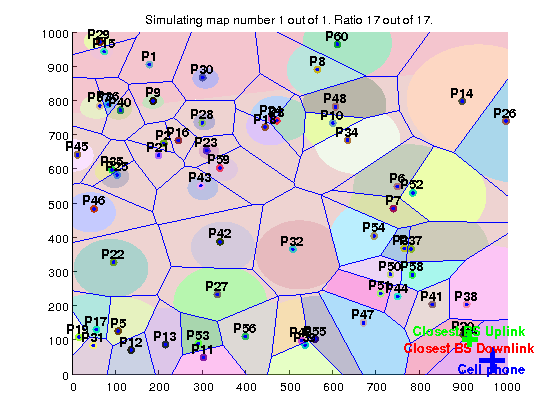}
   \caption{Coverage map: uplink and downlink.}
   \label{fig:superimposed}
   \end{center}
\end{figure}


\newpage
\subsection{Network deployment}
Relying on the 2-tier heterogeneous network architecture, we now proceed with the description of the test-bed network. Table~\ref{parameters} shows the parameters we have used to conduct the assessment of the simulator.

\begin{savenotes}
  \begin{table}[!htb]
    \caption{Deployment parameters}
    \centering
    \label{parameters}
    \begin{tabular}{cl} \\ \hline
      \multicolumn{1}{|c|}{\cellcolor[HTML]{EFEFEF}{\color[HTML]{343434} Area of interest}} & \multicolumn{1}{l|}{$1000$ m $\times$ $1000$ m} \\ \hline
      \multicolumn{1}{|c|}{\cellcolor[HTML]{EFEFEF}{\color[HTML]{343434} }}
                                                                                            &
                                                                                              \multicolumn{1}{l|}{$\lambda_\text{Macrocells}
                                                                                              =
                                                                                              3$}
      \\ 
\multicolumn{1}{|c|}{\cellcolor[HTML]{EFEFEF}{\color[HTML]{343434} }}                                                                                                 & \multicolumn{1}{l|}{$\lambda_\text{Femtocells}$ = $\lambda_\text{Macrocells}\cdot$  ratio\footnote{$ratio = \frac{\lambda_F}{\lambda_M}$. Ratio of the number of femtocells to the number of macrocells. In our case, this ratio ranges from 1 to 17.}} \\
\multicolumn{1}{|c|}{\multirow{-3}{*}{\cellcolor[HTML]{EFEFEF}{\color[HTML]{343434} \begin{tabular}[c]{@{}c@{}}Network deployment\\ (PPP intensities)\end{tabular}}}} & \multicolumn{1}{l|}{$\lambda_ {Users} = 5500$}                         \\ \hline
\multicolumn{1}{|c|}{\cellcolor[HTML]{EFEFEF}{\color[HTML]{343434} }}                                                                                                 & \multicolumn{1}{l|}{Macrocell DL/UL = 20 MHz}                       \\
\multicolumn{1}{|c|}{\multirow{-2}{*}{\cellcolor[HTML]{EFEFEF}{\color[HTML]{343434} Channel bandwidth}}}                                                              & \multicolumn{1}{l|}{Femtocell DL/UL = 1 GHz}                        \\ \hline
\multicolumn{1}{|c|}{\cellcolor[HTML]{EFEFEF}{\color[HTML]{343434} }}                                                                                                 & \multicolumn{1}{l|}{MBS = 46 dBm}                                   \\
\multicolumn{1}{|c|}{\cellcolor[HTML]{EFEFEF}{\color[HTML]{343434} }}                                                                                                 & \multicolumn{1}{l|}{FBS = 20 dBm}                                   \\
\multicolumn{1}{|c|}{\multirow{-3}{*}{\cellcolor[HTML]{EFEFEF}{\color[HTML]{343434} Transmit power}}}                                                                 & \multicolumn{1}{l|}{Device = 20 dBm}                                \\ \hline
\multicolumn{1}{|c|}{\cellcolor[HTML]{EFEFEF}{\color[HTML]{343434} Path-loss exponent}}                                                                               & \multicolumn{1}{l|}{$\alpha = 4$}                                      \\ \hline
\multicolumn{1}{|c|}{\cellcolor[HTML]{EFEFEF}{\color[HTML]{343434} Propagation constant}}                                                                             & \multicolumn{1}{l|}{1}                                              \\ \hline
\multicolumn{1}{|c|}{\cellcolor[HTML]{EFEFEF}{\color[HTML]{343434} Noise level}}                                                                                      & \multicolumn{1}{l|}{$- 106$ dBm}                                      \\ \hline
\multicolumn{1}{|c|}{\cellcolor[HTML]{EFEFEF}Number of iterations}                                                                                                    & \multicolumn{1}{l|}{$450$ maps}                                       \\ \hline
\end{tabular}
\end{table}
\end{savenotes}

\subsubsection*{Scenario description}

The experimental evaluation of decoupled access scheme is going to be focused
on a finite area of interest, which is a square of side $S = 1000$ m. As we
have already stated a randomised deployment, based on stochastic geometry and
Poison Point Processes (\textsc{ppp}s) has been used to settle all the base
stations (taking into account a specific ratio). After that, we use the same
association rules explained before (see Section~\ref{sec:assoc_rules}) in
order to build the coverage maps for both uplink and downlink.

As soon as we have got all these elements, in what follows we simulate the
arrival of users to the area of interest. To that end, we again employ a
\textsc{ppp} $\Phi_u$ of density $\lambda_{Users}$ to simulate it. Once we
have made all the desired measurements, we store them so that we can average
them later and repeat the process as many times as iterations (number of maps)
we have set but leaving the ratio unchanged. The aim of averaging is to smooth
the measurements. We do this in order to avoid giving excessive weight to
unusual deployments/cases that may emerge whiel using a randomised
deployment. Thus, we get a more reliable snapshot of the network
parameters. Finally, we repeat this process for every possible ratio.

\subsubsection{Distribution of the distance to the serving base station}

The first parameter we wanted to check was the distribution of the distance to a base
station which is serving a given user device. To that end, we now provide the
numerical results. Figure~\ref{fig:distances} shows the probability density
functions of the distance to the serving \textsc{bs} for two different ratios
($\frac{\lambda_F}{\lambda_M} = 5$ and $\frac{\lambda_F}{\lambda_M} = 17$) and
compares them with no decoupling case.

\begin{figure}[!htb]
  \begin{center}
    \includegraphics[scale=0.85]{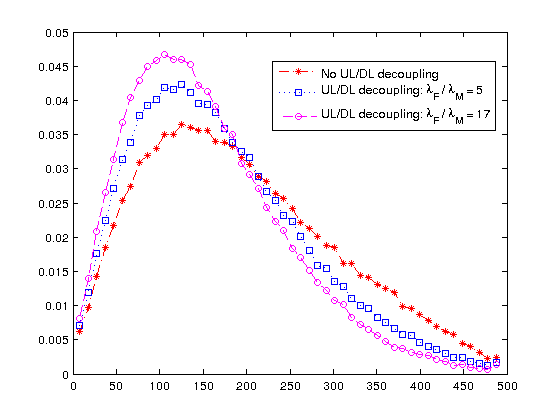}
    \caption{Probability density functions for the distances to the serving
      station.} 
    \label{fig:distances}
  \end{center}
\end{figure}

It seems clear that using uplink/downlink decoupling, the distribution becomes
narrower and is shifted to the left, that is, the average distance to the
serving station is smaller.  As we can see, the distance distribution shifts
to the left as we increase the ratio $\frac{\lambda_F}{\lambda_M}$. This makes
sense since femtocells have a smaller transmit power and coverage. Therefore,
user devices tend to remain in \emph{Case 2} and \emph{Case 4} (see
\ref{sec:assoc_rules}). Finally, we should note that the tails of the
distributions vanish sooner while using decoupling scheme than when
traditional approach is used.

\subsubsection{Association probabilities}

One of the main parameters that most probably will affect the network
performance is the optimal association decision. Willing to shed some light on
this issue, we carried out a simulation to assess the effect of increasing the
femtocell density ($\lambda_\textsc{f}$) in the association
cases. Figure~\ref{fig:prob_cases} shows association probabilities for the
deployed network.

\begin{figure}[!htb]
  \begin{center}
    \includegraphics[scale=0.85]{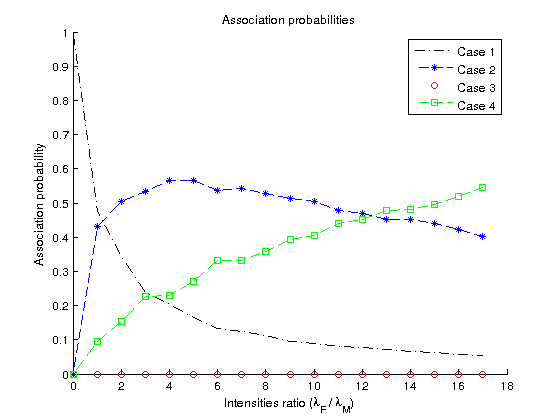}
    \caption{Association probabilities of each association case.}
    \label{fig:prob_cases}
  \end{center}
\end{figure}

It seems important to notice that as we increase the density of femtocells,
decoupled access scheme is enhanced (\emph{Case 2}), but this only happens to
a certain point. By further increasing $\lambda_\textsc{f}$, \emph{Case 2}'s
probability gradually decreases in favour of choosing a small base station for
both uplink and downlink, i.e. \emph{Case 4}.

\subsubsection{Average throughput}
In order to test the throughput measurement tool, we need to set some
additional parameters. First of all, while measuring throughput values for a
given user, some other users are assumed to be active at the same time. If
they are being served by the same base station that is serving the target
user, they are going to share the available bandwidth at that time. The
interference model taken into account to compute the \textsc{sinr} value for
the user is the one explained in section~\ref{sec:system_model}.

Finally, we define the uplink and downlink rate using \textbf{Shannon-Hartley
  theorem} as follows
\begin{equation}
  R = \frac{1}{N_{\text{Users}}} \,\,\textsc{bw} \,\,\log_2(1 + \mathsf{SINR}).
\end{equation}

It is woth mentioning that in this case we assume $500$ active users in the
downlink and $400$ in the uplink, in addition to the target user device.

Figure~\ref{fig:avg_sinrs} shows average downlink and uplink \textsc{sinr}
values as we increase femtocell densities.

\begin{figure}[!htb]
  \centering 
  \subfigure[Average downlink
  \textsc{sinr}.]{\label{fig:Avg_DL_SINR}\includegraphics[scale=0.50]{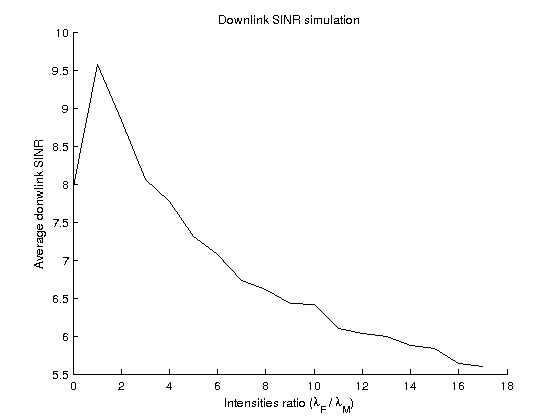}}
  \subfigure[Average uplink
  \textsc{sinr}.]{\label{fig:Avg_UL_SINR}\includegraphics[scale=0.50]{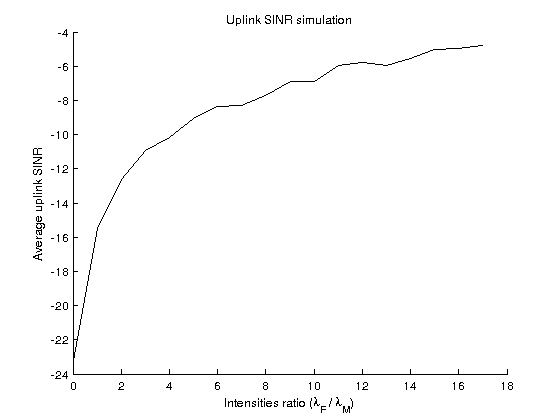}}
  \caption{Average \textsc{sinr} values.}
  \label{fig:avg_sinrs}
\end{figure}

As shown in figure~\ref{fig:Avg_DL_SINR}, despite the fact that at first sight
increasing the density of small cells leads to higher values of \textsc{sinr}
at the downlink, as $\frac{\lambda_\textsc{f}}{\lambda_\textsc{m}}$ becomes
large the number of interfering nodes also increase. As a consequence, a
trade-off between the number of femtocells and decoupled access seems the best
practise.

On the other hand, throughput measurements are shown in
figure~\ref{fig:avg_rates}. It seems that the enhancement obtained by using decoupled access is enough to compensate for the previously mentioned \textsc{sinr} degradation, at least with the assumed user density ($500$ active
users in the downlink and $400$ in the uplink at the same time). It might be
due to the fact that increasing small cell density increases interference but
also decreases the average number of users per cell and thus, leaves more
available bandwidth for each one of the remaining users in the cell.

\begin{figure}[!htb]
  \centering     
  \subfigure[Downlink rate (bits per second).]{\label{fig:Downlink_Throughput}\includegraphics[scale=0.50]{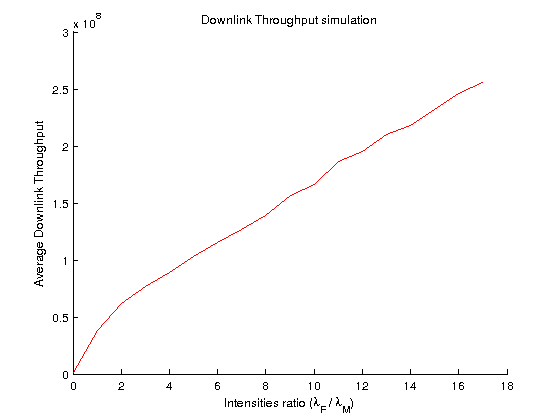}}
  \subfigure[Uplink rate (bits per second).]{\label{fig:Uplink_Throughput}\includegraphics[scale=0.50]{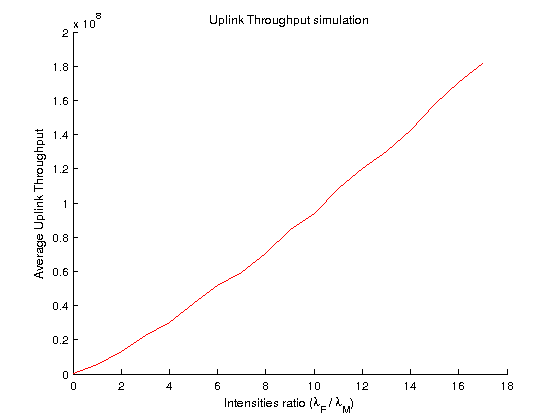}}
  \caption{Average throughput values.}
  \label{fig:avg_rates}
\end{figure}

\subsubsection*{Comparison with 'received power' network deployment}
Finally, we try to justify the benefits of using \textsc{ul}/\textsc{dl}
decoupling (\textsc{dud}e) instead of the traditional approach.

In figure~\ref{fig:comparison}, we show a comparison between \textsc{dud}e and
received power approaches for the uplink. No comparison for the downlink is
provided since it still uses received power approch in order to decide to
which base station a devices connects to.

\begin{figure}[!htb]
  \centering     
  \subfigure[Uplink \textsc{sinr} comparison.]{\label{fig:Uplink_SINR_comp}\includegraphics[scale=0.55]{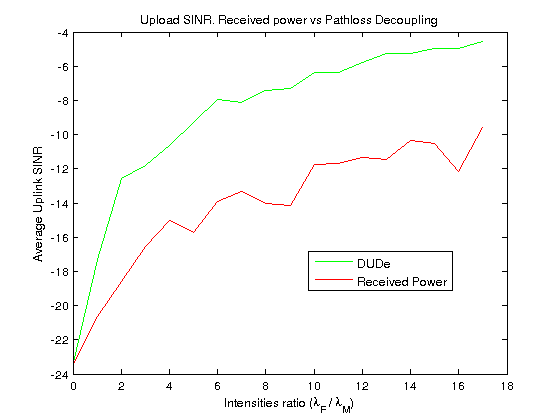}}
  \subfigure[Uplink rate (bits per second).]{\label{fig:Uplink_Throughput_comp}\includegraphics[scale=0.55]{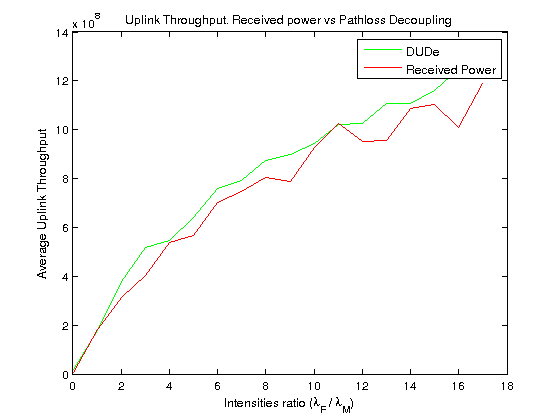}}
  \caption{Average throughput comparison.}
  \label{fig:comparison}
\end{figure}

As we can see, there exists a significant gain in both average \textsc{sinr}
($\simeq 4$dB) and average uplink throughput when we use decoupling.

\section{Conclusions}
We have assessed the performance gain obtained when using uplink/downlink
decoupling on a heterogeneous network deployment. The provided model can be
easily extended to implement new physical layer technologies such as
\textsc{mimo}, cell biasing, power control, etc. At the network level, device
to device communications, scheduling and complex cooperation techniques
between base stations can be included with minimal effort.

A random spatial distribution approach has been used to compute both the
position of base stations and user devices (\textsc{ppp}). These random models
based on stochastic geometry have shown their accuracy to model real-world
network deployments \cite{ktier}. Nevertheless, future work may include point
processes which model a minimum separation between points, i.e, \emph{Hard
  core point processes} (\textsc{hcpp}s). In that case, no two points of the
process coexist with a separating distance less than a predefined hard core
parameter. \emph{Poisson cluster processes} (\textsc{pcp}s), built from a
parent \textsc{ppp} can also be useful to model the clustering behaviour
observed on real cellular networks. The same discussion applies to user
devices.

It seems obvious that \textsc{ul}/\textsc{dl} decoupling allows us to
dynamically switch on or off some base stations which are not being used.

\chapter{Network utility maximisation}
\label{chap:num}

\section{DUDe as an optimization problem}

In emerging heterogeneous networks as the one presented in the latter chapter, using
association metrics like \textsc{sinr} or received power can lead
to load imbalance due to the disparate transmit powers of base stations. This
problem is illustrated in figure~\ref{fig:load_imbalance}. Despite the fact
that DUDe still yields substantial performance gains over co-located
association, the plot shows that a few Macro Base Stations are serving most of
the users whereas some other \textsc{bs}s are idle. Though this problem can be
partially solved by increasing the smaller base stations density, another
approach for granting the resources of \textsc{bs}s to the users could be
adopted.

\begin{figure}[!htb]
  \begin{center}
    \includegraphics[scale=0.7]{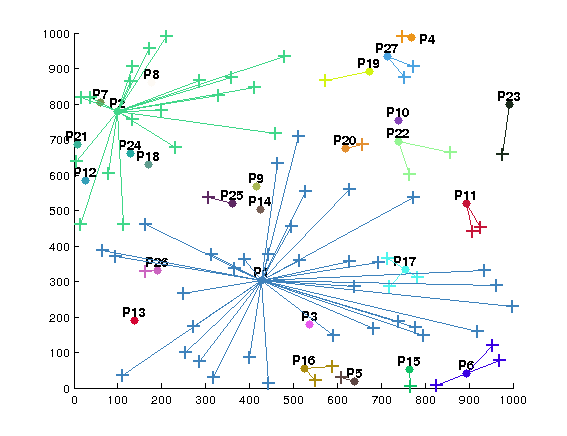}
    \caption{Load imbalance example. Downlink snapshot.}
    \label{fig:load_imbalance}
  \end{center}
\end{figure}

In this chapter, we shall reformulate the DUDe scheme under the framework and
tools of mathematical optimisation. Specifically, DUDe will be modeled as a
general \emph{network utility maximisation} (\textsc{num}) problem adapted to a
heterogeneous network architecture. Under \emph{single station association}
  (\textsc{ssa}) policies, a mobile user will be attached to a single base
station in each link (downlink and uplink). Consequently, the problem involves
computing both the optimal \emph{association} from mobile users to base
stations and optimal \emph{allocation} of the \textsc{bs} resources to every
user associated to it. In general, finding the optimal solution is a
combinatorial problem whose complexity grows exponentially with the number of
base stations and users. Since this approach is well-known to be
\textsc{np}-hard, we shall address this issue by introducing problem relaxations
and thereby reducing the complexity at the expense of obtaining tight upper
bounds. For a thorough review of NUM and its possible decomposition methods,
refer to~\cite{Palomar06}.

\section{System model and notation}

In this section, we are going to consider a network utility maximisation
problem (\textsc{num}) with an $\alpha$-proportional fair utility
function. As we mentioned before, each mobile user (\textsc{mu}) is allowed to
associate to only one \textsc{bs} per link (downlink and uplink). The total
utility a \textsc{mu} obtains from the network is computed by summing the
allocations granted by the serving \textsc{bs}s in downlink and uplink to that
\textsc{mu}. In addition, each allocation is then weighed by the instantaneous
rate on each link.

Consider a coverage area $A \subseteq \mathbb{R}^2$ for the wireless
network. Let $\mathcal{B_{\textsc{dl}}}$ be a set of \textsc{bs}s capable of
providing a downlink service, $\mathcal{B_{\textsc{ul}}}$ the set of
\textsc{bs}s capable of providing an uplink service, and $\mathcal{U}$ the
finite set of users. The \emph{instantaneous rate} for user
$u \in \mathcal{U}$ in the uplink or downlink direction with respect to the
serving \textsc{bs} $b$ is the ergodic (Shannon) channel capacity given by
\begin{align}
  r_{ub} &= \log_2 (1 + \underline{\mathsf{sinr}}_{ub}) \\
  r^\prime_{ub} &= \log_2 (1 + \overline{\mathsf{sinr}}_{ub})
\end{align}
where the signal-to-interference-plus-noise ratio in the downlink for a user
at location $y$ and a serving \textsc{bs} at $x_0$ is
\begin{equation}
  \label{eq:sinr}
  \underline{\mathsf{sinr}}_{ub} = \frac{P_i \,h_{x_0} \, \norm{x_0 -
      y}^{-\alpha}}{\sum^k_{j=1}\,\, \sum_{x \in \Phi_i,\, x\neq x_0} P_j
    \,h_{x_0} \, \norm{x - y}^{-\alpha}\,\, +\,\, \sigma^2}, 
\end{equation}
and $\overline{\mathsf{sinr}}_{ub}$ is defined similarly for the
uplink. In~\eqref{eq:sinr}, $j$ is the $j$-th tier, $\Phi_i$ denotes the
Poisson Point Process which models the location of tier $i$'s \textsc{bs}s and
$i \in \lbrace \textsc{m},\textsc{f} \rbrace$. We use $P_i$ to denote the
\textbf{transmission power} of a node at tier
$i \in \lbrace \textsc{m},\textsc{f}\rbrace$, where $i = \textsc{m}$ for macro
\textsc{bs}s and $i = \textsc{f}$ for femto \textsc{bs}s. The noise is
additive, Gaussian and has constant power $\sigma^2$. \textbf{Rayleigh fading}
is used to model the channel quality fluctuations between the \textsc{bs} and
the mobile device. $h_x \sim exp(1)$ describes the Rayleigh fading and it is
an exponentially distributed random variable with unit mean. In addition, we
assume path-loss exponent $\alpha \geq 2$. The same reasoning applies for the
uplink. In summary, we borrowed a very simmilar system model to that explained
in section~\ref{sec:system_model}.

Let $y_{ub}$ be the fraction of resources that \textsc{bs} $u$ grants to
\textsc{mu} $u$. This resource allocation fraction may represent a certain
amount of time to transmit depending on the multiplexing scheme the
\textsc{bs} is using.  Then, the sum-rates from the downlink and the uplink
for user $u$ are, respectively,
\begin{align}
  R_{u} (\mathbf{y}_1) &= \sum_{b \in \mathcal{B}_\text{DL}} r_{ub} y_{ub} \\
  R^\prime_{u} (\mathbf{y}_2) &= \sum_{b \in \mathcal{B}_\text{UL}}
                                     r^\prime_{ub} y^\prime_{ub}
\end{align}
where
$\mathbf{y}_1 = (y_{ub}, u \in \mathcal{U}, b \in \mathcal{B}_{\textsc{dl}})$
and
$\mathbf{y}_2 = (y_{ub}, u \in \mathcal{U}, b \in
\mathcal{B}_{\textsc{ul}})$
are the $|\mathcal{U}| \times |\mathcal{B}_{\textsc{dl}}|$ and
$|\mathcal{U}| \times |\mathcal{B}_{\textsc{ul}}|$ resource allocation
matrices for the downlink and for the uplink, respectively. In the following,
scalar or vector symbols with prime superindices will denote quantities for
the uplink channels, and vectors are denoted with boldface symbols. The first
constraint we need to define concerns the set of resources a \textsc{bs} can
offer. The maximum amount of resources that a \textsc{bs} can allocate is
normalised and set to unity. Therefore, the set of feasible allocations for
each link are\footnote{Note that both sets $\mathcal{Y_{\textsc{dl}}}$ and $\mathcal{Y_{\textsc{ul}}}$ are closed and convex.}
\begin{align*}
  \mathcal{Y_{\textsc{dl}}} &= \lbrace \mathbf{y}_1 : \mathbf{y}_1 \in
  \mathbb{R}_+^{|\mathcal{U}| \times  |\mathcal{B}_{\textsc{dl}}|},
  \sum_{u}y_{ub} = 1, \forall \, b \in \mathcal{B}_{\textsc{dl}} \rbrace \\
  \mathcal{Y_{\textsc{ul}}} &= \lbrace \mathbf{y}_2 : \mathbf{y}_2
  \in \mathbb{R}_+^{|\mathcal{U}| \times |\mathcal{B}_{\textsc{ul}}|}, \sum_u
  y^\prime_{ub} = 1, \forall \, b \in \mathcal{B}_{\textsc{ul}} \rbrace 
\end{align*}
Our goal is to maximise the network utility function, which is defined as the
sum of the individual users' utility function plus an additional term. For the
individual utility functions, we will use the class of $\alpha$-proportional
fair utility functions~\cite{Mo00,Uchida11}, defined as follows
\begin{equation}
  \label{eq:utility_function}
  U_{\alpha}(R) = \begin{cases}
    \frac{R^{1-\alpha}}{1- \alpha}, & \quad \alpha \geq 0,\, \alpha \neq 1 \\
    \log(R), & \quad \alpha = 1.
  \end{cases}
\end{equation}
Here, $R$ denotes the rate the user is perceiving from the network either in
uplink or downlink.

\section{\textsc{ssa} with optimal resource allocation and symmetrical link balance}
In this section, we assume that the association of users to their respective serving base stations is fixed. In this situation, the only way of changing the performance of the network is by tweaking the resource allocation parameters for each user, that is, the resource allocation vectors. 

\subsection{Problem formulation}

We formulate the complete optimisation problem using all the elements we
discussed in the last section. First, we need to capture the concept of each
\textsc{mu} associating with at most one base station on each link.  To that
end, define two association restriction vectors
$\mathbf{z}_1 = (z_{ub}, u \in \mathcal{U}, b \in \mathcal{B}_{\textsc{dl}})$
and
$\mathbf{z}_2 = (z^\prime_{ub}, u \in \mathcal{U}, b \in
\mathcal{B}_{\textsc{ul}})$,
one for each link. The elements of vectors $\mathbf{z}_1$ and
$\mathbf{z}_2$ are binary, so $z_{ub} = 1$ means that user $u$ is
associated to base station $b$ in the downlink. Thus, under a \textsc{ssa}
association rule, the sets $\mathcal{Z_{\textsc{dl}}}$ and
$\mathcal{Z_{\textsc{ul}}}$ of feasible associations are
\begin{align}
  \label{eq:feasible-dl}
  \mathcal{Z_{\textsc{dl}}} &= \lbrace \mathbf{z}_1: \mathbf{z}_1 \in
  \mathbb{Z}_+^{|\mathcal{U}| \times |\mathcal{B}_{\textsc{dl}}|}, \sum_{b \in
    \mathcal{B}_\text{DL}} z_{ub} = 1, \forall \, u \in \mathcal{U} \rbrace \\
  \label{eq:feasible-ul}
  \mathcal{Z_{\textsc{ul}}} &= \lbrace \mathbf{z}_2: \mathbf{z}_2 \in
  \mathbb{Z}_+^{|\mathcal{U}| \times |\mathcal{B}_{\textsc{ul}}|}, \sum_{b \in
  \mathcal{B}_\text{UL}} z_{ub} = 1, \forall \, u \in \mathcal{U} \rbrace.
\end{align}

The problem of optimal resource allocation under a \textsc{ssa} policy is that
of finding an optimal allocation of resources $\mathbf{y}_1$ and
$\mathbf{y}_2$ that maximises the sum utility of the sum rates at each
one of the \textsc{mu}s in the network. In addition, we want to positively
reward symmetrical downlink and uplink rates, that is, we want to minimise the
difference between them.

After introducing the \textsc{ssa}
constraints~\eqref{eq:feasible-dl}-\eqref{eq:feasible-ul}, the sum rates of
user $u$ for both downlink and uplink may now be expressed as
\begin{align}
  \label{eq:rates_dl}
  R_{u} (\mathbf{y}_1, \mathbf{z}_1) = \sum_{b \in \mathcal{B}_\textsc{dl}}
  r_{ub} y_{ub} z_{ub}, \\
  \label{eq:rates_ul}
  R^\prime_{u} (\mathbf{y}_2, \mathbf{z}_2) = \sum_{b \in
  \mathcal{B}_\textsc{ul}} r^\prime_{ub} y^\prime_{ub} z^\prime_{ub}.
\end{align}
Accordingly, the network utility maximisation (\textsc{num}) problem
under single station association (\textsc{ssa}) policy is
\begin{equation}\label{eq:main_v1}
\displaystyle
  f_\alpha^{\textsc{ssa}} \equiv \underset{\substack{\mathbf{y}_1,
      \mathbf{y}_2\\ \mathbf{z}_1, \mathbf{z}_2}}{\max} \sum_u
    U_{\alpha}(R_u (\mathbf{y}_1, \mathbf{z}_1)) + U_{\alpha}(R^\prime_u
    (\mathbf{y}_2, \mathbf{z}_2)) - A |R_u (\mathbf{y}_1, \mathbf{z}_1)
    - R^\prime_{u} (\mathbf{y}_2, \mathbf{z}_2)|
\end{equation}
such that
\begin{subequations}
  \begin{equation}
    \label{eq:constraint-ra-dl}
    \sum_{u \in \mathcal{U}} y_{ub} = 1, \forall\, b \in
    \mathcal{B}_{\textsc{dl}}, \mathbf{y}_1 \in \mathcal{Y}_\textsc{dl}
  \end{equation}
  \begin{equation}
    \label{eq:constraint-ra-ul}
    \sum_{u \in \mathcal{U}} y^\prime_{ub} = 1, \forall\, b \in
    \mathcal{B}_{\textsc{ul}}, \mathbf{y}_2 \in \mathcal{Y}_\textsc{ul}
  \end{equation}
  \begin{equation}
    \label{eq:constraint-ssa-dl}
    \sum_{b \in \mathcal{B}_\textsc{dl}} z_{ub} = 1, \forall\, u \in
    \mathcal{U}, \mathbf{z}_1 \in \mathcal{Z}_\textsc{dl}
  \end{equation}
  \begin{equation}
    \label{eq:constraint-ssa-ul}
    \sum_{b \in \mathcal{B}_\textsc{ul}} z^\prime_{ub} = 1, \forall\, u \in
    \mathcal{U}, \mathbf{z}_2 \in \mathcal{Z}_\textsc{ul}
  \end{equation}
\end{subequations}
where $A$ is a positive constant. Note that under feasible \textsc{ssa}
allocations for each link, the summation over $b$ in~\eqref{eq:rates_dl}
and~\eqref{eq:rates_ul} contains only one positive term each. That is, each
mobile user is associated with exactly one base station per link. Therefore,
in general we can rewrite the sum utility function as
\begin{equation}
  \sum_u U_{\alpha} \bigl( R(\mathbf{y}_1, \mathbf{z}_1) \bigr) = \sum_u
  U_{\alpha} \bigl( \sum_b r_{ub} y_{ub} z_{ub} \bigr) = \sum_{u,b}
  U_{\alpha}(r_{ub} y_{ub}) z_{u}.
\end{equation}
Thus, problem~\eqref{eq:main_v1} may be equivalently written as
\begin{equation}
  \label{eq:main_v2}
  f_\alpha^{\textsc{ssa}} \equiv \max \sum_{u,b} \bigl( U_{\alpha}(r_{ub} y_{ub})
  z_{ub} + U_{\alpha}(r^\prime_{ub} y^\prime_{ub}) z^\prime_{ub} \bigr) - A
  |\sum_{u,b} r_{ub} y_{ub} z_{ub} - r^\prime_{ub} y^\prime_{ub} z^\prime_{ub}| 
\end{equation}
with
constraints~\eqref{eq:constraint-ra-dl}-\eqref{eq:constraint-ssa-ul}. Note
that~\eqref{eq:main_v1} or~\eqref{eq:main_v2} is a mixed optimization problem
with both integer and continuous variables. Due to the combinatorial nature of
the objective function, the optimal solution is, except for degenerate cases,
hard to find. Note also that the last term in~\eqref{eq:main_v2} couples the
variables $\mathbf{y}_1$ and $\mathbf{y}_2$, so $f_\alpha^\textsc{ssa}$ is
not separable. The term
$|R_u(\mathbf{y}_1, \mathbf{z}_1) - R_u(\mathbf{y}_2, \mathbf{z}_2)|$
in~\eqref{eq:main_v1} quantifies the \emph{user rate imbalance}. i.e., the
asymmetry between the downlink and uplink rates; hence, the last term in the
objective function introduces a penalty on the \emph{network imbalance},
defined as the aggregate sum of users' imbalances\footnote{Another possible interpretation for the last term is that of regularisation, i.e., the introduction of a penalty term so that the preferred solutions are almost symmetrical in the user rates.}. An alternative statement of
the problem can be obtained by defining the new utility function
\begin{equation}
  V_\alpha(\mathbf{y}_1, \mathbf{y}_2, \mathbf{z}_1, \mathbf{z}_2)
  = \sum_b r_{ub} y_{ub} z_{ub} + \sum_b r^\prime_{ub} y^\prime_{ub}
  z^\prime_{ub} - A | \sum_b r_{ub} y_{ub} z_{ub} - r^\prime_{ub}
  y^\prime_{ub} z^\prime_{ub} |
\end{equation}
so that
\begin{equation}
  f_\alpha^\textsc{ssa} \equiv \max \sum_{u \in \mathcal{U}}
  V_\alpha(\mathbf{y}_1, \mathbf{y}_2, \mathbf{z}_1, \mathbf{z}_2)
\end{equation}
with the same constraints as before.

In spite of the complexity of the optimization problem, notice that if a
\emph{fixed} association between users and base stations is known, the problem
simplifies remarkably.
\begin{theorem}
  \label{thm:ssa-convexity}
  Choose feasible association schemes $\mathbf{z}_1$ and $\mathbf{z}_2$ for
  the downlink and uplink transmissions. Then, problem
  $f_\alpha^{\textsc{ssa}}$ is convex.
\end{theorem}
\begin{proof}
  If $\mathbf{z}_1$ and $\mathbf{z}_2$ are known, the first term in the
  objective function is the composition of a concave function with an affine
  function of $\mathbf{y}_1$, the second term is similar and the third summand
  is linear in $(\mathbf{y}_1, \mathbf{y}_2)$. Therefore, the objective
  function is concave. The equality constraints~\eqref{eq:constraint-ra-dl}
  and~\eqref{eq:constraint-ra-ul} are linear and the feasible sets
  $\mathcal{Y}_\textsc{dl}$ and $\mathcal{Y}_\textsc{ul}$ are easily seen to
  be convex. For a more detailed proof, we refer the reader to
  Appendix~\ref{app:a}.
\end{proof}

\subsection{Problem solution}

Under fixed association, the solution to $f_\alpha^\textsc{ssa}$ can be
found by means of standard convex optimization theory. Such solution gives
insight into structural properties of the optimal solution to the \textsc{ssa}
problem. Actually, we restate the problem in a slightly more general form
\begin{equation}
  \label{eq:main_v3}
  f_\alpha^\textsc{ssa}(\epsilon) \equiv \max_{\mathbf{y}_1, \mathbf{y}_2}
  \sum_u U_\alpha\bigl( R_u(\mathbf{y}_1, \mathbf{z}_1) \bigr) + U_\alpha\bigl(
  R_u(\mathbf{y}_2, \mathbf{z}_2) \bigr) - A \sum_u |
  R_u(\mathbf{y}_1, \mathbf{z}_1) - R_u ( \mathbf{y}_2,
  \mathbf{z}_2) |^{1 + \epsilon}
\end{equation}
with constraints $\mathbf{y}_1 \succeq 0$, $\mathbf{y}_2 \succeq 0$ and
$\sum_u y_{ub} = \sum_{u} y^\prime_{ub} = 1$ for all \textsc{BS} $b$, where
$\epsilon > 0$ is a sufficiently small number.\footnote{This is because
  $|x|^q$ is differentiable at $x = 0$ if $q > 1$.} Clearly,
$f_\alpha^\textsc{ssa}(\epsilon)$ is not decomposable in the variables
$\mathbf{y}_1$, $\mathbf{y}_2$, but its optimal solution may be
characterised by computing explicitly the Karush-Kuhn-Tucker (KKT) conditions
for the Lagrangian which, by Theorem~\ref{thm:ssa-convexity}, are sufficient and necessary for optimality. The
Lagrangian is\footnote{We omit the association vectors $\mathbf{z}_1$ and
  $\mathbf{z}_2$ from the notation for simplicity.}
\begin{equation}
  \label{eq:lagrangian}
  \begin{aligned}
  L(\mathbf{y}_1, \mathbf{y}_2, \boldsymbol\lambda,
  \boldsymbol\lambda^\prime, \boldsymbol\mu, \boldsymbol\mu^\prime) =
  &-f_\alpha^\textsc{ssa}(\epsilon) + \sum_{b_{\textsc{dl}}} \lambda_b (\sum_u y_{ub} - 1) + \sum_{b_{\textsc{ul}}}\lambda^\prime_b (\sum_u y^\prime_{ub} - 1) \\
  &+ \sum_{ub_{\textsc{dl}}} \mu_{u}
  y_{ub} + \sum_{ub_{\textsc{ul}}} \mu^\prime_{u} y^\prime_{ub}.
  \end{aligned}
\end{equation}
From this, there must exist Lagrange multipliers vectors $\boldsymbol\lambda$,
$\boldsymbol\lambda^\prime$, $\boldsymbol\mu$ and $\boldsymbol\mu^\prime$ such
that
\begin{subequations}\label{eq:kkts}
  \begin{equation}
    \label{eq:kkt-stationarity}
    \frac{\partial L}{\partial y_{ub}} = \frac{\partial L}{\partial
      y^\prime_{ub}} = 0, \quad \forall u, b \quad(\text{stationarity})
  \end{equation}
  \begin{equation}
        \label{eq:kkt-primal-feas}
    \begin{cases}
      & \sum_u y_{ub} = \sum_u y^\prime_{ub} = 1 \\
      & y_{ub} \geq 0, y^\prime_{ub} \geq 0
    \end{cases} \quad (\text{primal feasibility})
  \end{equation}
  \begin{equation}
    \label{eq:kkt-dual-feas}
    \mu_{ub} \geq 0, \mu^\prime_{u,b} \geq 0 \quad (\text{dual feasibility})
  \end{equation}
  \begin{equation}
    \label{eq:kkt-slackness}
    \mu_{u} y_{ub} = \mu^\prime_{u} y^\prime_{ub} = 0 \quad
    (\text{complementary slackness}).
  \end{equation}
\end{subequations}

  When $\alpha > 0$, $\alpha \neq 1$, the utility function is
  
  \begin{equation*}
  U_{\alpha}(R) = \frac{R^{1 - \alpha}}{1 - \alpha}.
  \end{equation*}   

Substituting the utility function expression into the Lagrangian leads to

\begin{equation*}
\begin{aligned}
L(\mathbf{y}_1, \mathbf{y}_2, \boldsymbol\lambda,
  \boldsymbol\lambda^\prime, \boldsymbol\mu, \boldsymbol\mu^\prime) &=
  -\sum_{ub} \bigl[ \frac{(r_{ub} y_{ub})^{1-\alpha}}{1-\alpha} z_{ub}\quad + \quad \frac{(r^\prime_{ub} y^\prime_{ub})^{1-\alpha}}{1-\alpha} z^\prime_{ub}  \bigr] \\
  &+ A\sum_{ub} \bigl[|r_{ub} y_{ub} z_{ub} - r^\prime_{ub} y^\prime_{ub} z^\prime_{ub} |^{1+\epsilon}\bigr] + \sum_{b_{\textsc{dl}}} \lambda_b (\sum_u y_{ub} - 1) \\
  &+ \sum_{b_{\textsc{ul}}} \lambda^\prime_b (\sum_u y^\prime_{ub} - 1) + \sum_{ub_{\textsc{dl}}} \mu_{u} y_{ub} + \sum_{ub_{\textsc{ul}}} \mu^\prime_{u} y^\prime_{ub}.
  \end{aligned}
\end{equation*}

Its partial derivatives, focusing only on one user are

\begin{subequations}
  \begin{equation}
  \begin{aligned}
   \frac{\partial L}{\partial y_{ub}} = &A (1+\epsilon)\,|r_{ub}y_{ub}z_{ub} - r^\prime_{ub} y^\prime_{ub} z^\prime_{ub}|^{\epsilon}\,r_{ub}z_{ub} \operatorname{sign}(r_{ub}y_{ub}z_{ub} - r^\prime_{ub} y^\prime_{ub} z^\prime_{ub}) \\ 
    &-r_{ub}^{1-\alpha}y_{ub}^{-\alpha} z_{ub} + \lambda_b + \mu_{u}  
  \end{aligned}
  \end{equation}
  \begin{equation}
  \begin{aligned}
\frac{\partial L}{\partial y^\prime_{ub}} = &- A (1+\epsilon)\,|r_{ub}y_{ub}z_{ub} - r^\prime_{ub} y^\prime_{ub} z^\prime_{ub}|^{\epsilon}\,r^\prime_{ub}z^\prime_{ub}\operatorname{sign}(r_{ub}y_{ub}z_{ub} - r^\prime_{ub} y^\prime_{ub} z^\prime_{ub}) \\
 &-r_{ub}^\prime{}^{1-\alpha} y_{ub}^\prime{}^{-\alpha} z^\prime_{ub}  + \lambda^\prime_b + \mu^\prime_{u}  
  \end{aligned}
  \end{equation}
\end{subequations}

For a given pair user $u$, base station $b$ assume that $y^\ast_{ub} > 0$, which partially satisfies primal feasibility~\eqref{eq:kkt-primal-feas}, and set
$\mu^\star_{u,b} = 0$, satisfying dual feasibility~\eqref{eq:kkt-dual-feas}
and complementary slackness~\eqref{eq:kkt-slackness}. In addition, assuming $\epsilon \to 0$ along with the stationary conditions~\eqref{eq:kkt-stationarity}, the above equations become 

\begin{subequations}
  \begin{equation}
    \frac{r_{ub}^{1-\alpha}}{y_{ub}^{\alpha}} z_{ub} = A r_{ub}z_{ub}\operatorname{sign}(r_{ub}y_{ub}z_{ub} - r^\prime_{ub} y^\prime_{ub} z^\prime_{ub}) + \lambda_b^\star
  \end{equation}
  \begin{equation}
\frac{r_{ub}^\prime{}^{1-\alpha}}{y_{ub}^\prime{}^{\alpha}} z^\prime_{ub} = - A r^\prime_{ub}z^\prime_{ub}\operatorname{sign}(r_{ub}y_{ub}z_{ub} - r^\prime_{ub} y^\prime_{ub} z^\prime_{ub}) + \lambda_b^\prime{}^\star 
  \end{equation}
\end{subequations}

This yields\footnote{Recall that $y_{ub}$ stands for the uplink allocation granted by base station $b \in \mathcal{B}_{\textsc{dl}}$ to user $u$ and $y^\prime_{ub}$ stands for the uplink allocation granted by base station $b \in \mathcal{B}_{\textsc{ul}}$ to user $u$, where $b$ may not be the same for each case.}
\begin{subequations}

  \begin{equation}  \label{eq:solution-dl}
    y^\star_{ub} = \begin{cases}
      \bigl(\frac{r_{ub}^{1-\alpha}}{A r_{ub} \operatorname{sign}(r_{ub}y_{ub} - r^\prime_{ub} y^\prime_{ub})+ \lambda_b^\star}\bigr)^{1/\alpha}, & \quad\text{if user $u$ is associated to $b$ in the downlink} \\ 
      0, & \quad\text{otherwise}
    \end{cases}
  \end{equation}

  \begin{equation}  \label{eq:solution-ul}
    y_{ub}^\prime{}^\star = \begin{cases}
      \bigl(\frac{r_{ub}^\prime{}^{1-\alpha}}{- A r^\prime_{ub} \operatorname{sign}(r_{ub}y_{ub} - r^\prime_{ub} y^\prime_{ub})+ \lambda_b^\prime{}^\star}\bigr)^{1/\alpha}, & \quad\text{if user $u$ is associated to $b$ in the uplink} \\
      0, & \quad\text{otherwise}.
    \end{cases}
  \end{equation}
\end{subequations}

Substituting the above equations into the primal feasibility conditions (\ref{eq:kkt-primal-feas}) we have, for each base station

\begin{subequations}

  \begin{equation}\label{eq:multipliers-dl}
    \sum_u \bigl(\frac{r_{ub}^{1-\alpha}}{A r_{ub} \operatorname{sign}(r_{ub}y_{ub} - r^\prime_{ub} y^\prime_{ub})+ \lambda_b^\star}\bigr)^{1/\alpha} = 1
  \end{equation}

  \begin{equation}  \label{eq:multipliers-ul}
	\sum_u \bigl(\frac{r_{ub}^\prime{}^{1-\alpha}}{- A r^\prime_{ub} \operatorname{sign}(r_{ub}y_{ub} - r^\prime_{ub} y^\prime_{ub})+ \lambda_b^\prime{}^\star}\bigr)^{1/\alpha} = 1
  \end{equation}
\end{subequations}

which are the implicit expressions for $\lambda_b^\star$ and $\lambda_b^\prime{}^\star$ multipliers.

When $\alpha = 1$ the utility function is $U_{\alpha}(R) = \log(R)$. We can specialise the above solution as follows

\begin{subequations}
  \begin{equation}  \label{eq:solution-dl-log}
    y^\star_{ub} = \begin{cases}
      \frac{1}{A r_{ub}\operatorname{sign}(r_{ub}y_{ub} - r^\prime_{ub} y^\prime_{ub})+ \lambda_b^\star}, & \quad\text{if $u$ is associated to $b$ in the
        downlink} \\ 
      0, & \quad\text{otherwise}
    \end{cases}
  \end{equation}
  \begin{equation}  \label{eq:solution-ul-log}
    y_{ub}^\prime{}^\star = \begin{cases}
      \frac{1}{- A r^\prime_{ub}\operatorname{sign}(r_{ub}y_{ub} - r^\prime_{ub} y^\prime_{ub})+ \lambda_b^\prime{}^\star}, & \quad\text{if $u$ is associated to $b$ in the uplink} \\
      0, & \quad\text{otherwise}.
    \end{cases}
  \end{equation}
\end{subequations}

And the same idea holds for the multipliers

\begin{subequations}

  \begin{equation}\label{eq:multipliers-dl-log}
    \sum_u \frac{1}{A r_{ub}\operatorname{sign}(r_{ub}y_{ub} - r^\prime_{ub} y^\prime_{ub})+ \lambda_b^\star} = 1
  \end{equation}
  \begin{equation} \label{eq:multipliers-ul-log}
	\sum_u \frac{1}{- A r^\prime_{ub}\operatorname{sign}(r_{ub}y_{ub} - r^\prime_{ub} y^\prime_{ub})+ \lambda_b^\prime{}^\star} = 1 .
  \end{equation}
\end{subequations}

\subsection{Explicit solution of $f_\alpha^{\textsc{ssa}}$}
Although the solution obtained is in a non-explicit form, this does not preclude the posibility to obtain a closed solution to the \textsc{ssa} problem with optimal resource allocation and user-balanced load for any real case. Our first approach consisted of solving the set of nonlinear equations \eqref{eq:solution-dl}, \eqref{eq:solution-ul}, \eqref{eq:multipliers-dl}, \eqref{eq:multipliers-ul} using a \textsc{matlab} toolbox for such purpose. We have to take into account that the number of equations grows with the number of base stations and users. In particular, the number of equations is one per user and link \eqref{eq:solution-dl}, \eqref{eq:solution-ul}, and another one per base station and link \eqref{eq:multipliers-dl}, \eqref{eq:multipliers-ul}. Therefore, the total number of nonlinear equations is: 

\begin{equation}
2\cdot|\mathcal{B}| + 2\cdot|\mathcal{U}|,
\end{equation}

where $|\mathcal{B}|$ denotes the total number of base stations and $|\mathcal{U}|$ is the total number of users in the network. In addition, the bigger the number of users, the more complex become \eqref{eq:multipliers-dl} and \eqref{eq:multipliers-ul}. These equations constitute the system of nonlinear equations that we must solve in order to obtain the optimal resource allocation when the association of mobile users to their respective base stations is fixed. As we expected, this seems to be a huge number of equations for the \textsc{matlab} toolbox to handle with. Such was the case that we only found a closed solution, that is, the solver managed to converge to an optimal solution only for extremely simple topologies with a low number of base stations and users and small values of $\alpha$.

Nevertheless, we were able to devise a solution to partially avoid this issue. First of all, given a user $u$, it seems clear that: 
\begin{equation}\label{eq:approximation}
\operatorname{sign}(r_{ub}y_{ub} - r^\prime_{ub} y^\prime_{ub}) = \operatorname{sign}(r_{ub} - r^\prime_{ub}),
\end{equation}
if the optimal allocation variables $y_{ub}$ and $y^\prime_{ub}$ do not invert the sign of the sum after the resource granting round. Hence, the previously mentioned set of equations can be wrtitten as follows:

\begin{subequations}
  \begin{equation} \label{eq:solution-dl-tractable}
    y^\star_{ub} = \begin{cases}
      \bigl(\frac{r_{ub}^{1-\alpha}}{A r_{ub} \operatorname{sign}(r_{ub} - r^\prime_{ub})+ \lambda_b^\star}\bigr)^{1/\alpha}, & \quad\text{if user $u$ is associated to $b$ in the downlink} \\ 
      0, & \quad\text{otherwise}
    \end{cases}
  \end{equation}
  \begin{equation}\label{eq:solution-ul-tractable}
    y_{ub}^\prime{}^\star = \begin{cases}
      \bigl(\frac{r_{ub}^\prime{}^{1-\alpha}}{- A r^\prime_{ub} \operatorname{sign}(r_{ub} - r^\prime_{ub})+ \lambda_b^\prime{}^\star}\bigr)^{1/\alpha}, & \quad\text{if user $u$ is associated to $b$ in the uplink} \\
      0, & \quad\text{otherwise} 
    \end{cases}
  \end{equation}
\end{subequations}

and

\begin{subequations}

  \begin{equation}\label{eq:multipliers-dl-tractable}
    \sum_u \bigl(\frac{r_{ub}^{1-\alpha}}{A r_{ub} \operatorname{sign}(r_{ub} - r^\prime_{ub})+ \lambda_b^\star}\bigr)^{1/\alpha} = 1
  \end{equation}
 
  \begin{equation} \label{eq:multipliers-ul-tractable}
	\sum_u \bigl(\frac{r_{ub}^\prime{}^{1-\alpha}}{- A r^\prime_{ub} \operatorname{sign}(r_{ub} - r^\prime_{ub})+ \lambda_b^\prime{}^\star}\bigr)^{1/\alpha} = 1.
  \end{equation}
\end{subequations}

Note that now, solving the system of equations has become a substantially simpler task. We can numerically solve \eqref{eq:multipliers-dl-tractable} and \eqref{eq:multipliers-ul-tractable} and thus obtain the uplink and downlink multipliers for each base station. Afterwards, we may substitute the multipliers back into \eqref{eq:solution-dl-tractable} and \eqref{eq:solution-ul-tractable} in order to compute the optimal resource allocation for each user. A pair of custom algorithms Alg.~\ref{alg:alg1-dl} and Alg.~\ref{alg:alg1-ul} were implemented for this purpose.
Consequently, if \eqref{eq:approximation} holds for all users, this solution is also a solution of the original problem, that is, it satisfies \eqref{eq:solution-dl}, \eqref{eq:solution-ul}, \eqref{eq:multipliers-dl} and \eqref{eq:multipliers-ul}. If it holds for most of the users in the network, this solution seems to be a good approximation of the optimal one, as we are about to check in the following section. As this particular study is not the main goal of this work, no further actions were taken in order to try getting an exact solution.
\vspace*{0.5cm}
\begin{algorithm}[H]\label{alg:alg1-dl}
\SetKwData{F}{f}\SetKwData{UserAllocation}{DownlinkUserAllocation}
 \KwData{$A$; $\alpha-fairness$; User-Rates $\in
  \mathbb{R}_+^{|\mathcal{U}| \times  |\mathcal{B_{\textsc{dl}}}|}$.}
 \KwResult{\textsc{BS}s' downlink multipliers of size $1\times |\mathcal{B_{\textsc{dl}}}|$; DLUserAllocation $\in \mathbb{R}_+^{|\mathcal{U}| \times  |\mathcal{B_{\textsc{dl}}}|} $.}
 \BlankLine
 $Downlink\textsc{bs}Multipliers \leftarrow 0 $\;

\For{$bs \leftarrow  1\, \KwTo\, |\mathcal{B_{\textsc{dl}}}|$}{
$\F \leftarrow$ 0\;
\For{$user \leftarrow  1\, \KwTo\, |\mathcal{U}|$}{

\If{ $user$ is associated to  $bs$ in downlink}{
$\F = \F + \UserAllocation$\tcc*{DownlinkUserAllocation = \eqref{eq:solution-dl-tractable}}
}

}

\tcc{Compute \textsc{bs} multiplier $\lambda$}
$\lambda = $ Numerically solve $\F$ for $\lambda$, e.g., using bisection\;

\For{$user \leftarrow  1\, \KwTo\, |\mathcal{U}|$}{

\eIf{ $user$ is associated to  $bs$}{
$\UserAllocation$ = \eqref{eq:solution-dl-tractable}\tcc*{Replacing $\lambda$ with the value computed before.}
}{
$\UserAllocation = 0$\;
}

}

}
 \caption{Computing optimal allocation for downlink.}
\end{algorithm}

\begin{algorithm}[H]\label{alg:alg1-ul}
\SetKwData{F}{f}\SetKwData{UserAllocation}{UplinkUserAllocation}
 \KwData{$A$; $\alpha-fairness$; User-Rates $\in
  \mathbb{R}_+^{|\mathcal{U}| \times  |\mathcal{B_{\textsc{ul}}}|}$.}
 \KwResult{\textsc{BS}s' uplink multipliers of size $1\times |\mathcal{B_{\textsc{ul}}}|$; UplinkUsersAllocation $\in \mathbb{R}_+^{|\mathcal{U}| \times  |\mathcal{B_{\textsc{ul}}}|} $.}
 \BlankLine
 $Uplink\textsc{bs}Multipliers \leftarrow 0 $\;

\For{$bs \leftarrow  1\, \KwTo\, |\mathcal{B_{\textsc{ul}}}|$}{
$\F \leftarrow$ 0\;
\For{$user \leftarrow  1\, \KwTo\, |\mathcal{U}|$}{

\If{ $user$ is associated to  $bs$ in uplink}{
$\F = \F + \UserAllocation$\tcc*{UplinkUserAllocation = \eqref{eq:solution-dl-tractable}}
}

}

\tcc{Compute \textsc{bs} multiplier $\lambda^\prime$}
$\lambda^\prime = $ Numerically solve $\F$ for $\lambda^\prime$, e.g., using bisection\;

\For{$user \leftarrow  1\, \KwTo\, |\mathcal{U}|$}{

\eIf{ $user$ is associated to  $bs$}{
$\UserAllocation$ = \eqref{eq:solution-dl-tractable}\tcc*{Replacing $\lambda$ with the value computed before.}
}{
$\UserAllocation = 0$\;
}

}

}

 \caption{Computing optimal allocation for uplink.}
\end{algorithm}

\subsection{Validating the solution}
Our goal is to assess the validity of the obtained solution as well as to show how the different values of the variables that parameterise the solution affect both the performance of the overall network and each one of the users. Namely, we are going to present several performance measurements that may help a hypothetical mobile operator to choose the appropriate values for $A$ and $\alpha$ depending on the particular network needs in every specific moment. As we will show next, on the one hand $A$ parameter controls the weight we give to the network asymmetry penalisation term. On the other hand, the $\alpha$ parameter in the proportional fair utility function (see~\ref{eq:utility_function}) will change the way we grant resources to a user depending on his instantaneous rate.

With the aim of performing a systematic comparison, we make use of the network simulator described in Chapter~\ref{chap:dude}. Simulation parameters concerning the channel model used in simulations are shown in table~\ref{parametersB}.


\begin{savenotes}
\begin{table}[!htb]
\caption{Deployment parameters}
\centering
\label{parametersB}
\begin{tabular}{cl}
\multicolumn{1}{l}{}                                                                                                                                                  &                                                                     \\ \hline
\multicolumn{1}{|c|}{\cellcolor[HTML]{EFEFEF}{\color[HTML]{343434} Area of interest}}                                                                                 & \multicolumn{1}{l|}{$1000$ m x $1000$ m}                                \\ \hline
\multicolumn{1}{|c|}{\cellcolor[HTML]{EFEFEF}{\color[HTML]{343434} }}                                                                                                 & \multicolumn{1}{l|}{$\lambda_{Macrocells} = 3$}                         \\
\multicolumn{1}{|c|}{\cellcolor[HTML]{EFEFEF}{\color[HTML]{343434} }}                                                                                                 & \multicolumn{1}{l|}{$\lambda_ {Femtocells}$ = 30} \\
\multicolumn{1}{|c|}{\multirow{-3}{*}{\cellcolor[HTML]{EFEFEF}{\color[HTML]{343434} \begin{tabular}[c]{@{}c@{}}Network deployment\\ (PPP intensities)\end{tabular}}}} & \multicolumn{1}{l|}{$\lambda_ {Users} = 200$}                         \\ \hline
\multicolumn{1}{|c|}{\cellcolor[HTML]{EFEFEF}{\color[HTML]{343434} }}                                                                                                 & \multicolumn{1}{l|}{Macrocell DL/UL = 20 MHz}                       \\
\multicolumn{1}{|c|}{\multirow{-2}{*}{\cellcolor[HTML]{EFEFEF}{\color[HTML]{343434} Channel bandwidth}}}                                                              & \multicolumn{1}{l|}{Femtocell DL/UL = 1 GHz}                        \\ \hline
\multicolumn{1}{|c|}{\cellcolor[HTML]{EFEFEF}{\color[HTML]{343434} }}                                                                                                 & \multicolumn{1}{l|}{MBS = 46 dBm}                                   \\
\multicolumn{1}{|c|}{\cellcolor[HTML]{EFEFEF}{\color[HTML]{343434} }}                                                                                                 & \multicolumn{1}{l|}{FBS = 20 dBm}                                   \\
\multicolumn{1}{|c|}{\multirow{-3}{*}{\cellcolor[HTML]{EFEFEF}{\color[HTML]{343434} Transmit power}}}                                                                 & \multicolumn{1}{l|}{Device = 20 dBm}                                \\ \hline
\multicolumn{1}{|c|}{\cellcolor[HTML]{EFEFEF}{\color[HTML]{343434} Path-loss exponent}}                                                                               & \multicolumn{1}{l|}{$4$}                                      \\ \hline
\multicolumn{1}{|c|}{\cellcolor[HTML]{EFEFEF}{\color[HTML]{343434} Propagation constant}}                                                                             & \multicolumn{1}{l|}{1}                                              \\ \hline
\multicolumn{1}{|c|}{\cellcolor[HTML]{EFEFEF}{\color[HTML]{343434} Noise level}}                                                                                      & \multicolumn{1}{l|}{$- 106$ dBm}                                      \\ \hline
\end{tabular}
\end{table}
\end{savenotes}

We should note that the simulator was initially implemented assuming uniform resource allocation for the users associated to a given base station, i.e., splitting the available bandwidth equally between the users. We will take advantage of that fact in order to compare the performance between the uniform resource allocation approach and the custom resource allocation under study.

Recall that we are assuming that the user association is fixed, i.e., once users have been associated to a base station in both uplink and downlik, they cannot change the association. To that end, we rely upon the association scheme presented in section~\ref{sec:assoc_rules} which was used to implement the simulator.

\newpage
\subsubsection{$\alpha-fairness \rightarrow 1$}
The aim of this test is to check if the suggested resource allocation solution tends to the uniform resource allocation when the parameters are set to some appropriate values. To that end, we use $A \simeq 0$ and $\alpha \rightarrow 1$. Figures~\ref{fig:uniform_uniform_user_rates} and~\ref{fig:uniform_uniform_aggregates} show the user spectral efficiencies after resource allocation and aggregate spectral efficiency of the network, respectively. Note that spectral efficiency for every user in the network for both uplink and uplink is almost identical regardless of the chosen scheme. Consequently, the sum spectral efficiency is also the same for the two of the resource allocation schemes under these conditions ($A \simeq 0$, $\alpha \simeq 1$). Notice that since $A \simeq 0$, it is not necessary to check if \eqref{eq:approximation} holds.
\begin{figure}[!htb]
  \begin{center}
    \includegraphics[scale=0.55]{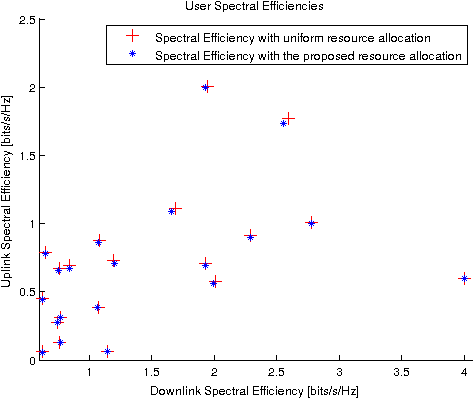}
    \caption{Users spectral efficiency after allocation.}
    \label{fig:uniform_uniform_user_rates}
  \end{center}
\end{figure}
\vspace*{-1.0cm}
\begin{figure}[!htb]
  \begin{center}
    \includegraphics[scale=0.55]{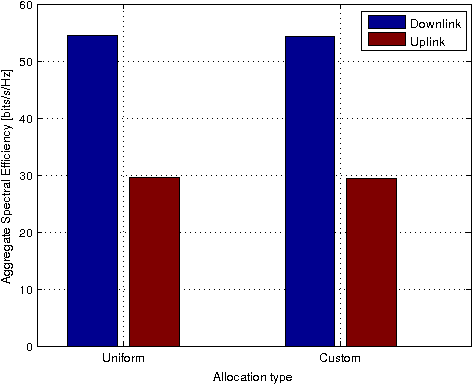}
    \caption{Aggregate spectral efficiency on the network for each allocation scheme.}
    \label{fig:uniform_uniform_aggregates}
  \end{center}
\end{figure}

\vspace*{-0.8cm}

So far, from the above results, everything seems to indicate that the scheme is working properly. We are obtaining the expected results which additionally comply with the solution produced when allocating the resources equally between the users.

\subsubsection{$\alpha-fairness \rightarrow 1$; $A = 5$}

We now give some weight to the network asymmetry term, that is, allocations where the difference between the uplink and downlink rates are smaller should be preferred over other solutions.

Obviously, in this case the fraction of resources allocated to each one of the users is not the same as for the uniform resource allocation scheme. Figure~\ref{fig:User_Spectral_Efficiencies_Uniform_UniformA5} plots the achieved spectral efficiency in uplink and downlink for each one of the users of the network, after the resource distribution. Note how either the uplink is rewarded at the expense of the downlink or the difference between uplink and downlink spectral efficiencies is smaller, for each network user.

\begin{figure}[!htb]
\centering     
\subfigure[Users spectral efficiency after allocation.]{\label{fig:User_Spectral_Efficiencies_Uniform_UniformA5}\includegraphics[scale=0.54]{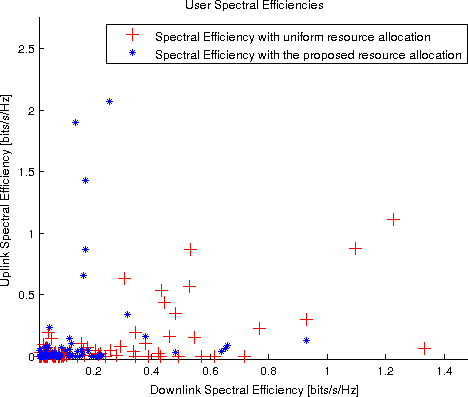}}
\subfigure[Aggregate spectral efficiency on the network for each allocation scheme.]{\label{fig:uniform_uniform_A5}\includegraphics[scale=0.54]{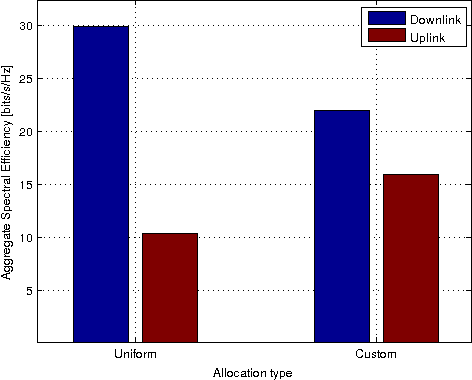}}
\caption{Simulation results for $\alpha \rightarrow 1$ and $A = 5$.}
\label{fig:uniform_uniform_A5_all}
\end{figure}

We can confirm the abovementioned facts in figures \ref{fig:donwlink_boxplot_uniformA5} and \ref{fig:uplink_boxplot_UniformA5_all}. On the one hand, we observe that the $2$nd and $3$rd quartiles are lower in the downlink spectral efficiency when using our custom resource allocation algorithm. On the other hand, uplink outliers now achieve higher spectral efficiencies (see Figure~\ref{fig:uplink_boxplot_UniformA5}). As a consequence, the gap between the aggregate rates of both links is reduced by half, as shown in Figure~\ref{fig:uniform_uniform_A5}. Furthermore, Figure~\ref{fig:uplink_boxplot_UniformA5_detail} shows a detailed view of the uplink spectral efficiency boxplot. Notice, at the right-hand side of the plot, that the mean is higher in the case of using the resource allocation scheme under study.

In summary, when we use uniform resource allocation under the simulation topology, we face a situation in which the uplink is clearly degraded in comparison with the downlink. However, when applying our custom algorithm so as to grant resources to network users, we can see the effect of the asymmetry penalisation term (\textsc{rhs} of \eqref{eq:main_v3}) which enables us to partially alleviate the uplink situation. As a consequence, the average difference between uplink and downlink rates of a user has also decreased ($\sim$ 10\% in this case). 
Finally, it is worth saying that in this particular case, $85$\% of the users satisfied equation \eqref{eq:approximation}. This value hardly ever goes below $70$\% but tests of significance might be performed anyway as future work in order to confirm this fact. Due to this, the obtained solution is only an approximation of the optimal one. Nevertheless, note that we are facing the expected and desired behaviour of the network allocation process.

\begin{figure}[!htb]
  \begin{center}
    \includegraphics[scale=0.65]{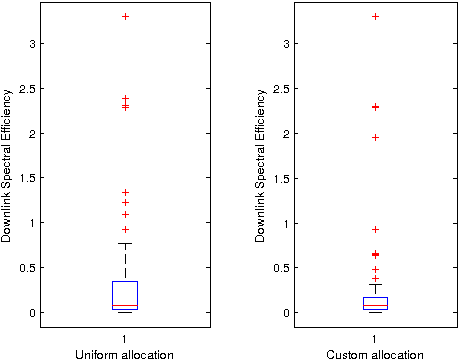}
    \caption{Downlink spectral efficiency boxplot.}
    \label{fig:donwlink_boxplot_uniformA5}
  \end{center}
\end{figure}

\begin{figure}[!htb]
\centering     
\subfigure[Entire boxplot.]{\label{fig:uplink_boxplot_UniformA5}\includegraphics[scale=0.56]{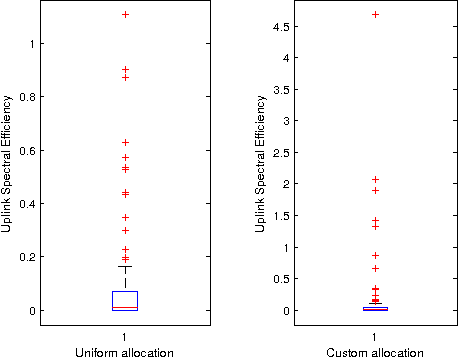}}
\subfigure[Boxplot detail.]{\label{fig:uplink_boxplot_UniformA5_detail}\includegraphics[scale=0.56]{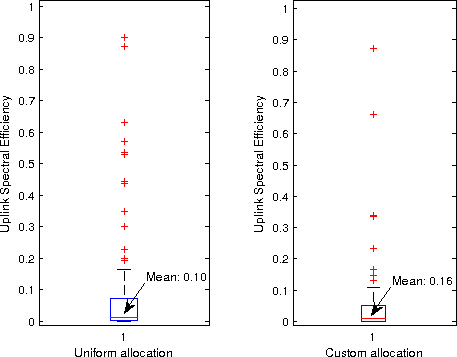}}
\caption{Uplink spectral efficiency boxplots.}
\label{fig:uplink_boxplot_UniformA5_all}
\end{figure}

To conclude with the test section, we report the simulation results of some illustrative cases concerning $\alpha$ parameter. In particular, we want to shed light on the behaviour of the algorithm when either the fairness value is too high or it tends to zero.

\newpage
\subsubsection{$\alpha-fairness \rightarrow 0$}
Now, we set the fairness parameter $\alpha$ to a value near zero ($\alpha = 0.15$) so that we can have an idea of the network behaviour when rate imbalance is not of utmost. To that end, we set parameter $A$ to zero. In figures \ref{fig:alphazero} and \ref{fig:alphazeroperformance} are shown the results of the simulations.

\begin{figure}[!htb]
\centering     
\subfigure[Downlink resource allocation.]{\label{fig:res_alloc_zero_downlink}\includegraphics[scale=0.54]{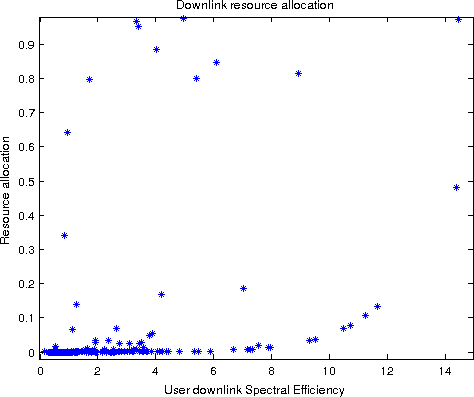}}
\subfigure[Uplink resource allocation.]{\label{fig:res_alloc_zero_uplink}\includegraphics[scale=0.54]{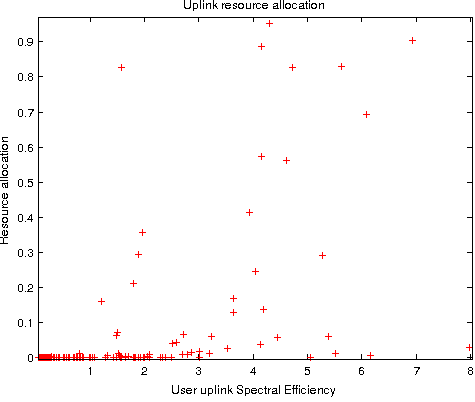}}
\caption{Resource allocation vs spectral efficiency. $\alpha \simeq 0$.}
\label{fig:alphazero}
\end{figure}

\begin{figure}[!htb]
\centering     
\subfigure[Users' spectral efficiency.]{\label{fig:uniform_zero_spectral_efficiencies}\includegraphics[scale=0.54]{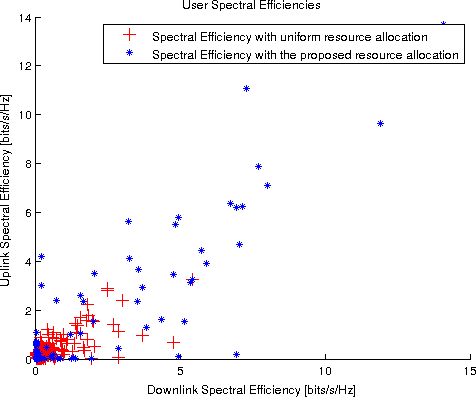}}
\subfigure[Aggregate spectral efficiency.]{\label{fig:uniform_zero_aggregates}\includegraphics[scale=0.54]{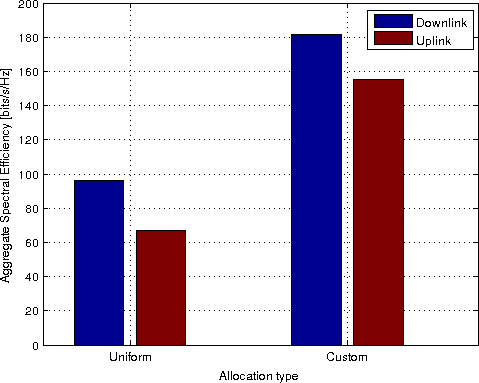}}
\caption{Individual and network performance measurements.}
\label{fig:alphazeroperformance}
\end{figure}

Notice in figures \ref{fig:res_alloc_zero_downlink} and \ref{fig:res_alloc_zero_uplink} that the higher the spectral efficiency of a given user, the more resources are granted to him. Outliers on the previously mentioned figures are explained due to the fact that they are associated to an almost idle base station. That is the reason why even they achieve low spectral efficiency, they perceive a big amount of resources. Figure \ref{fig:uniform_zero_spectral_efficiencies} shows some users are taking advantage of the purposed allocation scheme at the expense of some others who are obtaining lower uplink and downlin spectral efficiencies. As expected, these facts lead to a throughput maximisation scenario as it is confirmed in figure \ref{fig:uniform_zero_aggregates}. Note that both uplink and downlink aggregates are higher when the custom allocation algorithm is used.

\subsubsection{$\alpha-fairness > 1 $}
Our last test goal is to assess the network behaviour when parameter $\alpha$ is greater than $1$ and ultimately, what happens when it tends to infinity. To this end, we are going to test two different scenarios. In the first one, we will set ($\alpha = 4$ and $A = 0$). Finally, we set ($\alpha = 4$ and $A = 4$) so that we can check if the asymmetry term influences the solution when the fairness parameter is greater than $1$. Figure \ref{fig:alphafour} shows the simulation results for the first case:

\begin{figure}[!htb]
\centering     
\subfigure[Downlink resource allocation.]{\label{fig:res_alloc_four_downlink}\includegraphics[scale=0.54]{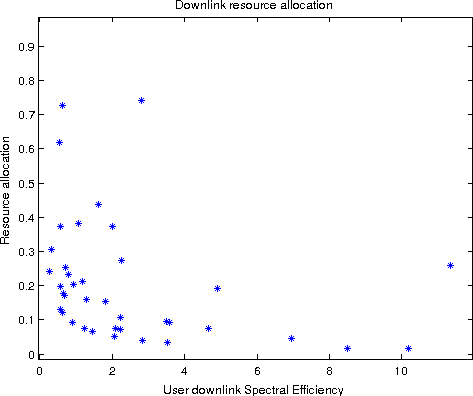}}
\subfigure[Uplink resource allocation.]{\label{fig:res_alloc_four_uplink}\includegraphics[scale=0.54]{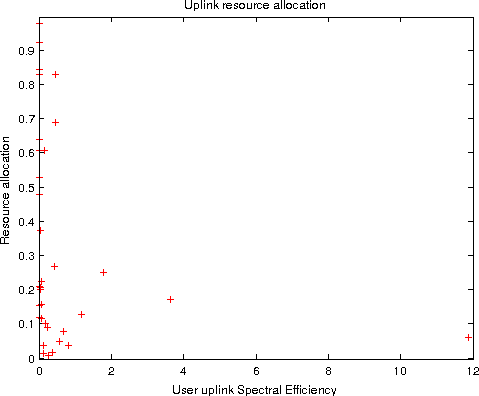}}
\caption{Resource allocation vs spectral efficiency. $\alpha \simeq 4; A = 0$.}
\label{fig:alphafour}
\end{figure}

It seems clear that now, the lower it is the spectral efficiency of a user, the bigger amount of resources are granted to him. As a consequence, both downlink and uplink aggregates are degraded but maybe this behaviour is considered to be more fair with individual users. We should also notice that the gap between downlink and uplink aggregates remains practically the same (see. Figure~\ref{fig:aggregate_four}). 

\newpage
\begin{figure}[!htb]
  \begin{center}
    \includegraphics[scale=0.65]{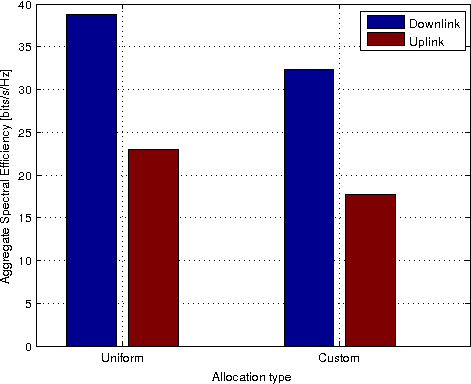}
    \caption{Network aggregates.}
    \label{fig:aggregate_four}
  \end{center}
\end{figure}

The percentage of users satisfying \eqref{eq:approximation} was $82$ \% and therefore we expect these results to be a good enough approximation to the optimal solution.

Finally, in Figure \ref{fig:aggregate_four_4} we explore the second situation where we set $\alpha = 4$ and $A = 4$.

\begin{figure}[!htb]
\centering     
\subfigure[Downlink resource allocation.]{\label{fig:res_alloc_four_4_downlink}\includegraphics[scale=0.54]{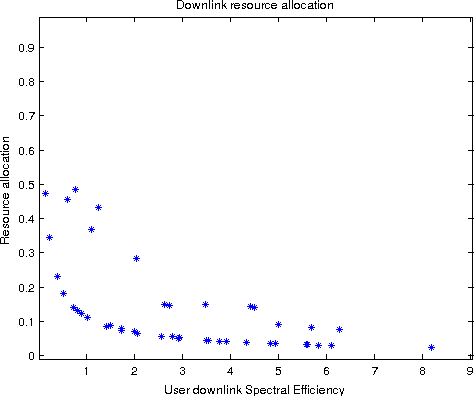}}
\subfigure[Network aggregates.]{\label{fig:aggregate_four_4}\includegraphics[scale=0.54]{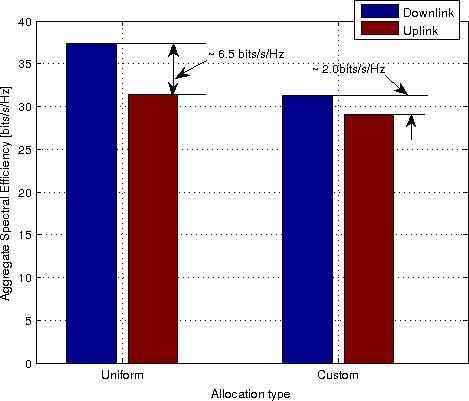}}
\caption{Resource allocation and network aggregates. $\alpha \simeq 4; A = 4$.}
\label{fig:alphafour_4}
\end{figure}

As shown in Figure~\ref{fig:res_alloc_four_4_downlink} the resource allocation criterion continues being the same. More resources are granted to those users with a lower spectral efficiency. Nevertheless, Figure~\ref{fig:aggregate_four_4} shows how the gap between uplink and downlink aggregates has become narrower. This time, the percentage of users satisfying \eqref{eq:approximation} was $85$\%. This, leads us to think that these results are a good enough approximations to the optimal ones.

\subsection{Discussion of results}

Due to the fact that we are permitting decoupled access, that is, we allow users to associate to different base stations in uplink and downlink, we may understand $A$ as a parameter which affects the behaviour of the network as a whole. Have in mind that now we are using DUDe, controling the behaviour of the network throught the $\alpha$ parameter is not really straightforward. That is the reason why the $A$ parameter seems to be useful in order to couple the control of the network's behaviour as a whole even if the uplink and downlink are decoupled for each user. 
 On the contrary, $\alpha-fairness$ parameter, has more to do with the allocations within a single base stations i.e. it controls how we grant resources to the users associated to a given base station basing on its instantaneous rates. From simulation results we may draw the following conclusions:
 
 In every performed test, parameter A helped to reduce the gap between uplink and downlik aggregates on the network as a whole. A value of the fairness parameter $\alpha$ lower than unity, leads to a throughput maximisation scenario where eventually a \textsc{bs} would only allocate resources to the user with the highest spectral efficiency associated to it (when $\alpha = 0$). On the other hand, as $\alpha$ grows above unity, more resources are granted to the users with lower spectral efficiency. To conclude, as we have already mentioned, we may not choose a single fairness and asymmetry parameters combination which is better that the rest but select dynamically the most appropriate configuration according to the needs of the network in every specific moment.
 
 \newpage
 \section{\textsc{ssa} problem with joint cell association and resource allocation problem}\label{sec:joint}
In this section, we study the tractability of the \emph{single station association} (\textsc{ssa}) maximisation problem when we no longer assume fixed association vectors, that is, we allow changes in the association of mobile users to their respective serving base stations. This fact implies that we do not employ the association scheme explained in section~\ref{sec:assoc_rules} anymore. Having in mind that the elements of feasible association vectors are binary (see. \eqref{eq:feasible-dl} and \eqref{eq:feasible-ul}) our problem becomes a combinatorial one. Any brute force solution for the complete problem stated above has complexity $\Theta(|\mathcal{B}|^{|\mathcal{U}|})$, where $|\mathcal{B}|$ denotes the total number of base stations and $|\mathcal{U}|$ is the total number of users in the network. This seems an unaffordable computational effort even when the number of users is small. Consequently, a problem relaxation appears to be the best option.

 \subsection{Integer relaxation of the problem}
 
We formulate the complete optimisation problem considering an integer relaxation for the association vectors. To that end, define two association restriction vectors $\mathbf{x}_1 = (x_{ub}, u \in \mathcal{U}, b \in \mathcal{B}_{\textsc{dl}})$
and $\mathbf{x}_2 = (x^\prime_{ub}, u \in \mathcal{U}, b \in
\mathcal{B}_{\textsc{ul}})$, one for each link. The elements of vectors $\mathbf{x}_1$ and
$\mathbf{x}_2$ are now real numbers. We denote the problem $f_{\alpha}^{\textsc{rs,o}}$(\textsc{rs,o} stands for Relaxed \textsc{ssa} problem with optimal allocation). Under these conditions the sets $\mathcal{X_{\textsc{dl}}}$ and
$\mathcal{X_{\textsc{ul}}}$ of feasible associations are

\begin{align}
  \label{eq:feasible-dl-relaxed}
  \mathcal{X_{\textsc{dl}}} &= \lbrace \mathbf{x}_1: \mathbf{x}_1 \in
  \mathbb{R}_+^{|\mathcal{U}| \times |\mathcal{B}_{\textsc{dl}}|}, \sum_{b \in
    \mathcal{B}_\text{DL}} x_{ub} = 1, \forall \, u \in \mathcal{U} \rbrace \\
  \label{eq:feasible-ul-relaxed}
  \mathcal{X_{\textsc{ul}}} &= \lbrace \mathbf{x}_2: \mathbf{x}_2 \in
  \mathbb{R}_+^{|\mathcal{U}| \times |\mathcal{B}_{\textsc{ul}}|}, \sum_{b \in
  \mathcal{B}_\text{UL}} x_{ub} = 1, \forall \, u \in \mathcal{U} \rbrace.
\end{align}

Below, we study the convexity of the relaxed problem.

 \subsubsection{Non - convexity of $f_{\alpha}^{\textsc{rs,o}}$}
 
For general $\alpha$, we can express $f_{\alpha}^{\textsc{rs,o}}$ as follows
 
 \begin{equation}\label{eq:relaxed-formulation}
 f_{\alpha}^{\textsc{rs,o}} \equiv \max_{\substack{\mathbf{y}_1, \mathbf{y}_2\\ \mathbf{x}_1, \mathbf{x}_2}} \sum_{ub} \bigl[ \frac{(r_{ub} y_{ub})^{1-\alpha}}{1-\alpha} x_{ub}\quad + \quad \frac{(r^\prime_{ub} y^\prime_{ub})^{1-\alpha}}{1-\alpha} x^\prime_{ub}  \bigr] - A\sum_{ub} \bigl[|r_{ub} y_{ub} x_{ub} - r^\prime_{ub} y^\prime_{ub} x^\prime_{ub} |\bigr] 
 \end{equation}
 
 On the one hand, each term $R_u(\mathbf{y}, \mathbf{x}) = r_{ub} y_{ub} x_{ub}$ is of the form $g(x,y) = x\cdot y$, which is not a convex set in $x,y\,\in[0, 1]$. The Hessian matrix of $g(x, y)$ is
 
 \begin{equation}\label{eq:xy-hessian}
\nabla^2 g(x,y) = \begin{bmatrix} %
0 & 1 \\
1 & 0 
\end{bmatrix},
 \end{equation}

whose eigenvalues are $\lbrace\lambda_1, \lambda_2\rbrace = {1, -1}$ and thus, \eqref{eq:xy-hessian} is not positive (semi) definite over $x,y\,\in[0, 1]$. Consequently, $g(x, y)$ is not convex. Figure \ref{fig:xy} shows that  $g(x, y)$ is indeed a hyperbolic paraboloid, which is formed by two parabolas of opposite curvatures (one convex and the other concave, curved in opposite directions).

On the other hand, each summand of the \textsc{lhs} of \eqref{eq:relaxed-formulation} is of the form $f(x,y) = x^{1-\alpha}\cdot y$. The Hessian matrix of this scalar-valued function is

 \begin{equation}\label{eq:xy-hessian-f}
H = \nabla^2 f(x,y) = \begin{bmatrix} %
-\alpha(1-\alpha)yx^{-\alpha -1} & (1-\alpha)x^{-\alpha} \\
\frac{\partial (\partial f/\partial x)}{\partial y} = \frac{\partial (\partial f/\partial y)}{\partial x} & 0 
\end{bmatrix},
 \end{equation}
 
 which has eigenvalues

 \begin{equation}
\lbrace\lambda_1, \lambda_2\rbrace = \frac{1}{2}\bigl(-\alpha(1-\alpha)yx^{-\alpha-1} \pm \sqrt{(\alpha - \alpha^2)^2y^2x^{-2\alpha -2} + 4 (1-\alpha)^2x^{-2\alpha}}\bigr).
 \end{equation}

We may rewrite eigenvalues as follows

\begin{equation}
\lbrace\lambda_1, \lambda_2\rbrace = -\frac{1}{2}(1-\alpha)x^{-\alpha-1}\bigl[\alpha y \pm \sqrt{\alpha^2 y^2 + 4x^2} \bigr].
 \end{equation}

Analysing their ratio, we obtain

\begin{equation}
\frac{\lambda_1}{\lambda_2} = \frac{1 + \frac{\sqrt{\alpha^2 y^2+4x^2}}{\alpha y}}{1 - \frac{\sqrt{\alpha^2 y^2+4x^2}}{\alpha y}} = \frac{1 + \sqrt{\frac{\alpha^2 y^2+4x^2}{\alpha^2 y^2}}}{1 - \sqrt{\frac{\alpha^2 y^2+4x^2}{\alpha^2 y^2}}}
\end{equation}

and it is clear that the numerator is positive. Conversely, the denominator is always negative since

\begin{equation}
1 - \sqrt{\frac{\alpha^2 y^2+4x^2}{\alpha^2 y^2}} = 1 - \sqrt{1 + \frac{4x^2}{\alpha^2 y^ 2}} < 0
\end{equation}

Therefore, the eigenvalues are of mixed signs and the Hessian matrix is not postive (semi) definite. Consequently, $f(x,y)$ is not convex. Figure \ref{fig:xalphay} shows the region of interest. As a result, due to the fact that $f_{\alpha}^{\textsc{rs,o}}$ is a weighted sum of non-convex functions as the ones studied above, we can state that $f_{\alpha}^{\textsc{rs,o}}$ is, in general, a non-convex optimisation problem.

\begin{figure}[!htb]
\centering     
\subfigure[$g(x,y)$.]{\label{fig:xy}\includegraphics[scale=0.54]{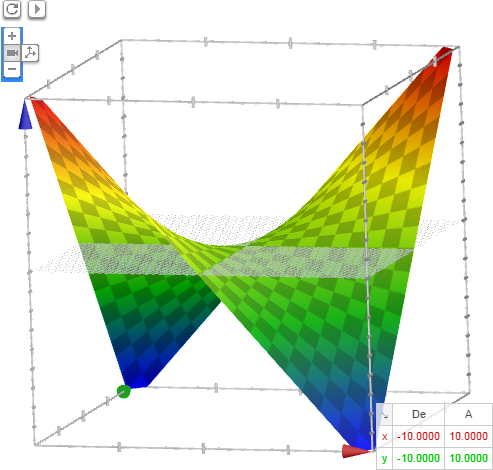}}
\subfigure[$f(x,y)$.]{\label{fig:xalphay}\includegraphics[scale=0.54]{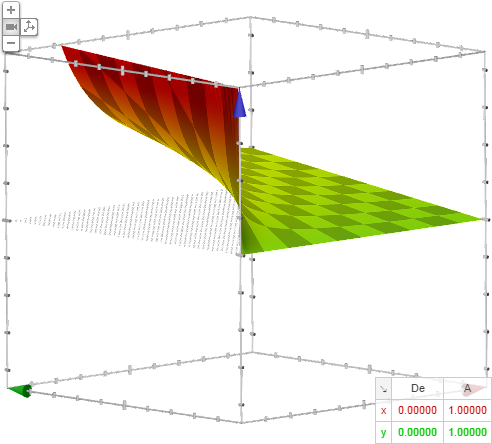}}
\caption{Regions under study.}
\label{fig:regions}
\end{figure}

\subsection{Problem reformulation}

Despite the results from the last section establish that we are facing a non-convex optimisation problem, we study a problem reformulation which may be convex. To that, end we attach additional constraints to the problem so as to capture the \textsc{ssa} constraint in a different way

\begin{equation}
  f_\alpha^{\textsc{ssa}} \equiv \max_{\mathbf{y}_1, \mathbf{y}_2} \sum_{u,b} \bigl( U_{\alpha}(r_{ub} y_{ub}) + U_{\alpha}(r^\prime_{ub} y^\prime_{ub}) \bigr) - A
  |\sum_{u,b} r_{ub} y_{ub}  - \sum_{u,b} r^\prime_{ub} y^\prime_{ub}| 
\end{equation}

such that

\begin{subequations}
  \begin{equation}
    \label{eq:constraint-mult-dl}
     y_{ua}y_{ub} = 0, \forall\, u\, \in \mathcal{U}, a\neq b \in
    \mathcal{B}_{\textsc{dl}}
  \end{equation}
  \begin{equation}
    \label{eq:constraint-mult-ul}
 y_{uc}y^\prime_{ub} = 0, \forall\, u\, \in \mathcal{U}, c\neq b \in
    \mathcal{B}_{\textsc{ul}},
  \end{equation}
\end{subequations}

along with all the remaining constraints from the original problem.

Note that the new constraints \eqref{eq:constraint-mult-dl} and \eqref{eq:constraint-mult-ul} are of the form $g(x,y) = x\cdot y$. We have already proven that $g(x,y)$ is non-convex at the previous section (see. eq~\eqref{eq:xy-hessian}). Nevertheless, we can tackle the non-convexity of the constraints by reworking them in a more suitable way.  

\begin{equation}\label{eq:ssa-epsilon}
  f_\alpha^{\textsc{ssa}}(\epsilon) \equiv \min_{\mathbf{y}_1, \mathbf{y}_2} - \sum_{u,b} \bigl( U_{\alpha}(r_{ub} y_{ub}) + U_{\alpha}(r^\prime_{ub} y^\prime_{ub}) \bigr) + A
  |\sum_{u,b} r_{ub} y_{ub} - r^\prime_{ub} y^\prime_{ub}| 
\end{equation}

such that

\begin{subequations}
  \begin{equation}
    \label{eq:constraint-mult-dl-rework}
     y_{ua}y_{ub} \leq \epsilon , \forall\, u\, \in \mathcal{U}, a\neq b \in
    \mathcal{B}_{\textsc{dl}}
  \end{equation}
  \begin{equation}
    \label{eq:constraint-mult-ul-rework}
 y_{uc}y^\prime_{ub} \leq \epsilon , \forall\, u\, \in \mathcal{U}, c\neq b \in
    \mathcal{B}_{\textsc{ul}},
  \end{equation}
\end{subequations}

where $\epsilon > 0$. It is clear that we capture \textsc{ssa} constraint if we choose $\epsilon \rightarrow 0$. Unfortunately, $f_\alpha^{\textsc{ssa}}$ problem as stated above is not a geometric program as we initially expected. A geometric program (\textsc{gp}) is an optimisation problem where the objective function and inequality restrictions are posymonials and equality restrictions are monomials. A posynomial is a polinomial of the form

\begin{equation}
g(x) = \sum_i c_i \prod_j x_j^{a_{ij}}\, ,\, c_i > 0, \, a_{ij} \in \mathbb{R}.
\end{equation}

Although most of the summations in \eqref{eq:ssa-epsilon} contain terms which are posynomials, i.e. they may be expressed as polynomials with postitive coefficients for all terms when $\alpha > 1$

\begin{equation}
\min \frac{1}{\alpha - 1} \sum_{ub} (r_{ub} y_{ub})^{1-\alpha},
\end{equation} 

there will always be a term with a polynomial with a negative coefficient coming from the rate asymmetry term. Consequently, the problem turns into a signomial geometric programming (\textsc{sgp}). These kind of problems are much more difficult to solve than geometric programs because, unlike posynomials, signomials are not guaranteed to be globally convex. Recent work on that field {\cite{signomial}} shows that there exist new global optimization algorithms which are based on transformation of variables and linearization techniques. Considering the degree of difficulty and the need of implementing more sophisticated algorithms in order to solve this signomial problem, we will try to solve joint cell association and resource allocation problem using a different approach in the next Chapter.

\chapter{Decentralised algorithms for utility maximisation}
\label{chap:decentralised}
In this chapter, we will address a joint problem. On the one hand, we want to optimally associate users to their respective base stations in uplink and downlink. On the other hand, we expect to allocate resources in each one of the base stations according to the chosen utility function. Since the \textsc{ssa} policy gives rise to a combinatorial problem which is \textsc{np}-hard, we seek now a solution via relaxation of the allocation variables and proper decomposition so as to derive fast and scalable \emph{decentralised} optimal solutions of the relaxed problem. These will naturally be more complex to implement compared to those under the \textsc{ssa} policy, but still suitable for deployment in large-scale 5\textsc{g} networks.

\section{The Multi-Station Association (\textsc{msa}) problem}

In section \ref{sec:joint}, we proved that this problem is not readily tractable following the \emph{single station association} (\textsc{ssa}) approach. In order to overcome this issue we relax that assumption and suppose that each user can be associated to more than one base station per link at the same time. Under this assumption, we no longer need the association restriction vectors $\mathbf{z}_1 = (z_{ub}, u \in \mathcal{U}, b \in \mathcal{B}_{\textsc{dl}})$ and $\mathbf{z}_2 = (z^\prime_{ub}, u \in \mathcal{U}, b \in
\mathcal{B}_{\textsc{ul}})$, which were limiting the feasible associations. Instead, we may use the resource allocation variables to indicate if a user is associated to a given base station, that is, a user $u$ is associated to \textsc{bs} $b$ in downlink if $y_ub{} > 0$ and it is not otherwise. Before addressing the possible solutions of the problem we shall reformulate it taking into account the above considerations.

Let $y_{ub}$ be the fraction of resources that \textsc{bs} $u$ grants to
\textsc{mu} $u$. This resource allocation fraction may represent a certain
amount of time to transmit depending on the multiplexing scheme the
\textsc{bs} is using.  Then, the sum-rates from the downlink and the uplink
for user $u$ are, respectively,
\begin{align}
  R_{u} (\mathbf{y}_1) &= \sum_{b \in \mathcal{B}_\text{DL}} r_{ub} y_{ub} \\
  R^\prime_{u} (\mathbf{y}_2) &= \sum_{b \in \mathcal{B}_\text{UL}}
                                     r^\prime_{ub} y^\prime_{ub}
\end{align}
where
$\mathbf{y}_1 = (y_{ub}, u \in \mathcal{U}, b \in \mathcal{B}_{\textsc{dl}})$
and
$\mathbf{y}_2 = (y_{ub}, u \in \mathcal{U}, b \in
\mathcal{B}_{\textsc{ul}})$
are the $|\mathcal{U}| \times |\mathcal{B}_{\textsc{dl}}|$ and
$|\mathcal{U}| \times |\mathcal{B}_{\textsc{ul}}|$ resource allocation
matrices for the downlink and for the uplink, respectively. In the following,
scalar or vector symbols with prime superindices will denote quantities for
the uplink channels, and vectors are denoted with boldface symbols. The first
constraint we need to define concerns the set of resources a \textsc{bs} can
offer. The maximum amount of resources that a \textsc{bs} can allocate is
normalised and set to unity. Therefore, the set of feasible allocations for
each link are
\begin{align*}
  \mathcal{Y_{\textsc{dl}}} &= \lbrace \mathbf{y}_1 : \mathbf{y}_1 \in
  \mathbb{R}_+^{|\mathcal{U}| \times  |\mathcal{B}_{\textsc{dl}}|},
  \sum_{u}y_{ub} = 1, \forall \, b \in \mathcal{B}_{\textsc{dl}} \rbrace \\
  \mathcal{Y_{\textsc{ul}}} &= \lbrace \mathbf{y}_2 : \mathbf{y}_2
  \in \mathbb{R}_+^{|\mathcal{U}| \times |\mathcal{B}_{\textsc{ul}}|}, \sum_u
  y^\prime_{ub} = 1, \forall \, b \in \mathcal{B}_{\textsc{ul}} \rbrace 
\end{align*}

Our goal again is to maximise the network utility function, which is defined as the
sum of the individual users' utility function plus an additional term concerning the symmetry between uplink and downlik. For the individual utility functions, we will use the same class of $\alpha$-proportional
fair utility functions we used before, defined as follows:
\begin{equation}
  U_{\alpha}(R) = \begin{cases}
    \frac{R^{1-\alpha}}{1- \alpha}, & \quad \alpha \geq 0,\, \alpha \neq 1 \\
    \log(R), & \quad \alpha = 1.
  \end{cases}
\end{equation}

Here, $R$ denotes the rate the user is perceiving from the network either in
uplink or downlink.

Accordingly, we can formulate the canonical optimisation problem under \textsc{msa} policy as:

\begin{equation}\label{eq:main_msa}
  f_\alpha^{\textsc{msa}} \equiv \underset{\substack{\mathbf{y}_1,
      \mathbf{y}_2}}{\max} \sum_u
    \bigl[U_{\alpha}(\sum_{b_{\textsc{dl}}} r_{ub}y_{ub}) + U_{\alpha}(\sum_{b_{\textsc{ul}}}r^\prime_{ub} y^\prime_{ub})\bigr]
\end{equation}
such that
\begin{subequations}
\begin{equation}
    \label{eq:constraint-msa-absA}
\sum_{b_{\textsc{dl}}}r_{ub}y_{ub}  - \sum_{b_{\textsc{ul}}}r^\prime_{ub} y^\prime_{ub} \leq \epsilon_u, \forall\, u \in \mathcal{U}
 \end{equation}
 \begin{equation}
    \label{eq:constraint-msa-absB}
\sum_{b_{\textsc{ul}}}r^\prime_{ub} y^\prime_{ub} - \sum_{b_{\textsc{dl}}}r_{ub}y_{ub}  \leq \epsilon_u, \forall\, u \in \mathcal{U}
 \end{equation}
  \begin{equation}
    \label{eq:constraint-msa-dl}
    \sum_{u \in \mathcal{U}} y_{ub} \leq 1, \forall\, b \in
    \mathcal{B}_{\textsc{dl}}, \mathbf{y}_1 \in \mathcal{Y}_\textsc{dl}
  \end{equation}
  \begin{equation}
    \label{eq:constraint-msa-ul}
    \sum_{u \in \mathcal{U}} y^\prime_{ub} \leq 1, \forall\, b \in
    \mathcal{B}_{\textsc{ul}}, \mathbf{y}_2 \in \mathcal{Y}_\textsc{ul}
  \end{equation}
\end{subequations}

The constraints \eqref{eq:constraint-msa-absA} and \eqref{eq:constraint-msa-absB} establish a maximum amount of rate asymmetry \emph{per user}, given by the non-negative values of $\epsilon_u$ (a parameter); constraints \eqref{eq:constraint-msa-dl} and \eqref{eq:constraint-msa-ul} are the normalisation constraints enforcing each base station not to allocate more than its total resources.

\begin{remark}
There is a major difference between problem \eqref{eq:main_msa} and the \textsc{ssa} problem considered in the previous chapter: note that the penalty or regularisation term has been modet to the constraints \eqref{eq:constraint-msa-absA} - \eqref{eq:constraint-msa-absB}. Expressing that concept as a constraint instead of as a term of the objective function remarkably simplifies the decomposition and it has two immediate consequences for the sake of mathematical tractability:
\begin{itemize}
\item The objective function is crearly separable in the variables $\mathbf{y}_1$, $\mathbf{y}_2$.
\item The new constraints are linear in the variables. 
\end{itemize}

\end{remark}
 
\begin{remark}
The second difference is that the constraints \eqref{eq:constraint-msa-dl} and \eqref{eq:constraint-msa-ul} are now \emph{inequality constraints}. However, it is clear that the strict equality can be relaxed without affecting the optimal point and the optimal value of the problem. Besides, note that strict equality in the resource utilisation might be in conflict with the bounded asymmetry in \eqref{eq:constraint-msa-absA} - \eqref{eq:constraint-msa-absB}, if some the bounds $\lbrace\epsilon_u \rbrace$ were too tight. In that case, the feasible region could be empty.
\end{remark}

\subsection{Convexity}
\begin{theorem}
  \label{thm:msa-convexity}
  Choose feasible allocation schemes $\mathbf{y}_1$ and $\mathbf{y}_2$ for
  the downlink and uplink. Then, problem
  \eqref{eq:main_msa} is convex.
\end{theorem}
\begin{proof}
  If $\mathbf{y}_1$ and $\mathbf{y}_2$ are feasible allocation schemes, the
  objective function is a sum of a composition of a concave function with affine functions of $\mathbf{y}_1$, $\mathbf{y}_2$.  Therefore, the objective
  function is concave itself. The constraints ~\eqref{eq:constraint-msa-absA}, \eqref{eq:constraint-msa-absB} are linear, and the same reasoning holds true with respect to ~\eqref{eq:constraint-msa-dl} and~\eqref{eq:constraint-msa-ul}. In addition, the feasible allocation sets
  $\mathcal{Y}_\textsc{dl}$ and $\mathcal{Y}_\textsc{ul}$ are easily seen to
  be convex. Thus, the feasible set is convex.
\end{proof}

\begin{remark}
Note, however, that the utility of the aggregated rate used by a given user in the downlink or in the uplink is not strictly (or strongly) concave, since an equation of the form $\underset{b}{\sum} r_{ub}y_{ub} = C$ may have multiple solutions on $y_{ub}$.
\end{remark}

\section{Full Dual Decomposition}
Notice that since we are facing a convex optimisation problem, a local optimum of the problem is also globally optimal. In addition, duality gap is zero under the problem constraints, so the Karush-Kuhn-Tucker (KKT) conditions are necessary and sufficient for primal-dual optimality. Below, we perform a dual decomposition of the problem with the aim of finding decomposable structures that now remain unseen. Lagrange duality theory connects the original maximisation problem \eqref{eq:main_msa} with the dual maximisation problem \eqref{eq:msa_lagrange} by relaxing the former, transfering the constraints to the objective funtion via Lagrange multipliers. Therefore, the Lagrangian of the problem is defined as

\begin{equation}\label{eq:msa_lagrange}
\begin{aligned}
L(\mathbf{y}_1, \mathbf{y}_2, \boldsymbol\lambda, \boldsymbol\lambda^\prime, \boldsymbol\nu, \boldsymbol\nu^\prime) &= \sum_u
    \bigl[U_{\alpha}(\sum_{b_{\textsc{dl}}} r_{ub}y_{ub}) + U_{\alpha}( \sum_{b_{\textsc{ul}}}r^\prime_{ub} y^\prime_{ub})\bigr] - \\ 
    &- \sum_u \lambda_u\bigl(\sum_{b_{\textsc{dl}}}r_{ub}y_{ub}  - \sum_{b_{\textsc{ul}}}r^\prime_{ub} y^\prime_{ub} - \epsilon_u \bigr) - \\
    &-\sum_u \lambda^\prime_u\bigl(\sum_{b_{\textsc{ul}}}r^\prime_{ub} y^\prime_{ub} - \sum_{b_{\textsc{dl}}}r_{ub}y_{ub} - \epsilon_u \bigr) - \\ &-\sum_{b_{\textsc{dl}}} \nu_b(\sum_{u\in\mathcal{U}(b)} y_{ub}-1) - \sum_{b_{\textsc{ul}}} \nu^\prime_b(\sum_{u\in\mathcal{U}(b)} y^\prime_{ub}-1).
\end{aligned}
\end{equation}

Grouping all the summations over $u$, we have

\begin{equation}
\begin{aligned}
L(\mathbf{y}_1, \mathbf{y}_2, \boldsymbol\lambda, \boldsymbol\lambda^\prime, \boldsymbol\nu, \boldsymbol\nu^\prime) &= \sum_u
    \bigl[U_{\alpha}(\sum_{b_{\textsc{dl}}}r_{ub}y_{ub}) + U_{\alpha}(\sum_{b_{\textsc{ul}}}r^\prime_{ub} y^\prime_{ub}) - \lambda_u\bigl(\sum_{b_{\textsc{dl}}}r_{ub}y_{ub}\\
    & - \sum_{b_{\textsc{ul}}}r^\prime_{ub} y^\prime_{ub} - \epsilon_u \bigr) - \lambda^\prime_u\bigl(\sum_{b_{\textsc{ul}}}r^\prime_{ub} y^\prime_{ub} - \sum_{b_{\textsc{dl}}}r_{ub}y_{ub} - \epsilon_u \bigr) \bigr] + \\     
&+\sum_{b_{\textsc{dl}}} \nu_b(1 - \sum_{u/u\in\mathcal{U}(b)} y_{ub}) + \sum_{b_{\textsc{ul}}} \nu^\prime_b(1 - \sum_{u/u\in\mathcal{U}(b)} y^\prime_{ub}).
\end{aligned}
\end{equation}

We may rework the above expression so as to reach an appropriate Lagrangian form which allows us to decompose the problem. To that end, combine the summations over base stations (both uplink and downlink) and expand the last two terms of the expression

\begin{equation}
\begin{aligned}
L(\mathbf{y}_1, \mathbf{y}_2, \boldsymbol\lambda, \boldsymbol\lambda^\prime, \boldsymbol\nu, \boldsymbol\nu^\prime) &= \sum_u
    \bigl[U_{\alpha}(\sum_{b_{\textsc{dl}}}r_{ub}y_{ub}) + \sum_{b_{\textsc{dl}}} \bigl(- \lambda_u r_{ub}y_{ub} + \lambda^\prime_u r_{ub}y_{ub} \bigr) +\\
    &+ U_{\alpha}(\sum_{b_{\textsc{ul}}}r^\prime_{ub} y^\prime_{ub}) + \sum_{b_{\textsc{ul}}} \bigl(\lambda_u r^\prime_{ub} y^\prime_{ub} - \lambda^\prime_u r^\prime_{ub} y^\prime_{ub}\bigr) + \epsilon_u (\lambda_u + \lambda^\prime_u) \bigr] +\\     
&+\sum_{b_{\textsc{dl}}} \nu_b + \sum_{b_{\textsc{ul}}} \nu^\prime_b - \sum_{b_{\textsc{dl}}}\nu_b (\sum_{u\in\mathcal{U}(b)} y_{ub}) - \sum_{b_{\textsc{ul}}}\nu^\prime_b (\sum_{u\in\mathcal{U}(b)} y^\prime_{ub}).
\end{aligned}
\end{equation}

Now, exchanging the order of summations over base stations on the last two terms yields

\begin{equation}
\begin{aligned}
L(\mathbf{y}_1, \mathbf{y}_2, \boldsymbol\lambda, \boldsymbol\lambda^\prime, \boldsymbol\nu, \boldsymbol\nu^\prime) &= \sum_u
    \bigl[U_{\alpha}(\sum_{b_{\textsc{dl}}}r_{ub}y_{ub}) + \sum_{b_{\textsc{dl}}} r_{ub}y_{ub} (\lambda^\prime_u-\lambda_u)\, + \\
    &+ U_{\alpha}(\sum_{b_{\textsc{ul}}}r^\prime_{ub} y^\prime_{ub}) +\sum_{b_{\textsc{ul}}} r^\prime_{ub} y^\prime_{ub} (\lambda_u - \lambda^\prime_u) + \epsilon_u (\lambda_u + \lambda^\prime_u)\bigr] + \\     
&+\sum_{b_{\textsc{dl}}} \nu_b + \sum_{b_{\textsc{ul}}} \nu^\prime_b - \sum_u(\sum_{b\in \mathcal{B}_{\textsc{dl}}(u)} \nu_b y_{ub}) - 
\sum_u(\sum_{b\in \mathcal{B}_{\textsc{ul}}(u)} \nu^\prime_b y^\prime_{ub}).
\end{aligned}
\end{equation}

Finally, we combine the new summations over $u$ with the ones we had already computed. The final form of the Lagrangian is given by:

\begin{equation}
\begin{aligned}
L(\mathbf{y}_1, \mathbf{y}_2, \boldsymbol\lambda, \boldsymbol\lambda^\prime, \boldsymbol\nu, \boldsymbol\nu^\prime) &= \sum_u
    \bigl[U_{\alpha}(\sum_{b_{\textsc{dl}}}r_{ub}y_{ub}) + \sum_{b_{\textsc{dl}}} \bigl( r_{ub}y_{ub} (\lambda^\prime_u-\lambda_u) - \nu_b y_{ub} \bigr) + \\
    &+ U_{\alpha}(\sum_{b_{\textsc{ul}}}r^\prime_{ub} y^\prime_{ub}) + \sum_{b_{\textsc{ul}}} \bigl(r^\prime_{ub} y^\prime_{ub} (\lambda_u - \lambda^\prime_u) - \nu^\prime_b y^\prime_{ub}\bigr) + \epsilon_u (\lambda_u + \lambda^\prime_u)\bigr] + \\     
&+\sum_{b_{\textsc{dl}}} \nu_b + \sum_{b_{\textsc{ul}}} \nu^\prime_b.
\end{aligned}
\end{equation}

Clearly, the optimisation now separates into two different levels. At the lower level of the problem, we have subproblems for each user in the network. That is, the dual decomposition results in each mobile user $u$ solving the $u$-th Lagrangian $L_u( y_{ub}, y^\prime_{ub}, \lambda_u, \lambda^\prime_u, \boldsymbol\nu, \boldsymbol\nu^\prime)$, for the given multipliers $\boldsymbol\nu, \boldsymbol\nu^\prime$

\begin{equation}\label{eq:decentralised_lower}
\begin{aligned}
\operatorname{arg} \max_{y_{ub}, y^\prime_{ub} \geq 0} \quad \bigl[&U_{\alpha}(\sum_{b_{\textsc{dl}}}r_{ub}y_{ub}) + \sum_{b_{\textsc{dl}}} \bigl(r_{ub}y_{ub} (\lambda^\prime_u-\lambda_u) - \nu_b y_{ub} \bigr) +\\
+\, &U_{\alpha}(\sum_{b_{\textsc{ul}}}r^\prime_{ub} y^\prime_{ub}) + \sum_{b_{\textsc{ul}}} \bigl(r^\prime_{ub} y^\prime_{ub} (\lambda_u - \lambda^\prime_u) - \nu^\prime_b y^\prime_{ub}\bigr) +\\
+ &\epsilon_u (\lambda_u + \lambda^\prime_u)  \bigr]\quad \forall u,
\end{aligned}
\end{equation}

where each user knows its own multipliers $\lambda_u$, $\lambda^\prime_u$, which actually set a price to the chosen uplink and downlink allocations so as to give preference to symmetrical solutions. Since uplink and downlink terms are independent of each other, each user has to solve \emph{two independent subproblems} (albeit identical)

\begin{subequations}\label{eq:decentralised_lower_subproblems}
\begin{equation}\label{eq:decentralised_lower_subproblem-dl}
\operatorname{arg} \max_{y_{ub} \geq 0} \quad U_{\alpha}(\sum_{b_{\textsc{dl}}}r_{ub}y_{ub}) + \sum_{b_{\textsc{dl}}} \bigl(r_{ub}y_{ub} (\lambda^\prime_u-\lambda_u) - \nu_b y_{ub} \bigr)
\end{equation}

\begin{equation}\label{eq:decentralised_lower_subproblem-ul}
\operatorname{arg} \max_{y^\prime_{ub} \geq 0}\quad U_{\alpha}(\sum_{b_{\textsc{ul}}}r^\prime_{ub} y^\prime_{ub}) + \sum_{b_{\textsc{ul}}} \bigl(r^\prime_{ub} y^\prime_{ub} (\lambda_u - \lambda^\prime_u) - \nu^\prime_b y^\prime_{ub}\bigr).
\end{equation}
\end{subequations}

Note that the solution to \eqref{eq:decentralised_lower} is $\lbrace y_{ub}^\star(\lambda_u, \lambda^\prime_u, \boldsymbol\nu, \boldsymbol\nu^\prime)$, $y_{ub}^\prime{}^\star(\lambda_u, \lambda^\prime_u, \boldsymbol\nu, \boldsymbol\nu^\prime)\rbrace$ and we will explore how to compute it in the next section. The \emph{master dual problem} is therefore

\begin{equation}\label{eq:decentralised_master}
\begin{aligned}
&\underset{\boldsymbol\lambda, \boldsymbol\lambda^\prime, \boldsymbol\nu, \boldsymbol\nu^\prime}{\operatorname{minimise}} \quad g(\boldsymbol\lambda, \boldsymbol\lambda^\prime, \boldsymbol\nu, \boldsymbol\nu^\prime) = \sum_u g_u(\lambda_u, \lambda^\prime_u, \boldsymbol\nu, \boldsymbol\nu^\prime) + \boldsymbol\nu^T\boldsymbol 1 + \boldsymbol\nu^\prime{}^T \boldsymbol 1 \\
&\text{subject to } \boldsymbol\lambda, \boldsymbol\lambda^\prime, \boldsymbol\nu, \boldsymbol\nu^\prime \geq 0,
\end{aligned}
\end{equation}

where $g_u(\lambda_u, \lambda^\prime_u, \boldsymbol\nu, \boldsymbol\nu^\prime) = L_u(y_{ub}^\star, y_{ub}^\prime{}^\star, \lambda_u, \lambda^\prime_u, \boldsymbol\nu, \boldsymbol\nu^\prime)$, that is, the Lagrangian for user $u$ evaluated at the optimal point. We know, by Theorem~\ref{thm:msa-convexity}, that the problem stated in \eqref{eq:main_msa} is convex. One consequence is that the dual function $g(\boldsymbol\lambda, \boldsymbol\lambda^\prime, \boldsymbol\nu, \boldsymbol\nu^\prime)$ is differentiable in its domain. Therefore, we may use the \emph{gradient projection method} in order to solve~\eqref{eq:decentralised_master}. Direct calculation gives the partial derivatives of the dual function as

\begin{subequations}
\begin{equation}\label{eq:msa-partialBS-downlink}
\frac{\partial\, g(\boldsymbol\lambda, \boldsymbol\lambda^\prime, \boldsymbol\nu, \boldsymbol\nu^\prime)}{\partial \boldsymbol\nu} = \sum_u\sum_{b_{\textsc{dl}}/b\in \mathcal{B}_{\textsc{dl}}(u)}(-y_{ub}) + \sum_{b_{\textsc{dl}}} 1 = \sum_{b_{\textsc{dl}}} \sum_{u/u\in\mathcal{U}(b)} (-y_{ub}) + \sum_{b_{\textsc{dl}}} 1
\end{equation} 
\begin{equation}\label{eq:msa-partialBS-uplink}
\frac{\partial\, g(\boldsymbol\lambda, \boldsymbol\lambda^\prime, \boldsymbol\nu, \boldsymbol\nu^\prime)}{\partial \boldsymbol\nu^\prime} = \sum_u\sum_{b_{\textsc{ul}}/b\in \mathcal{B}_{\textsc{ul}}(u)}(-y^\prime_{ub}) + \sum_{b_{\textsc{ul}}} 1 = \sum_{b_{\textsc{ul}}} \sum_{u/u\in\mathcal{U}(b)} (-y^\prime_{ub}) + \sum_{b_{\textsc{ul}}} 1
\end{equation}
\begin{equation}\label{eq:msa-partialUser-downlink}
\frac{\partial\, g(\boldsymbol\lambda, \boldsymbol\lambda^\prime, \boldsymbol\nu, \boldsymbol\nu^\prime)}{\partial \boldsymbol\lambda} = \sum_u\bigl( \sum_{b_{\textsc{ul}}}r^\prime_{ub} y^\prime_{ub} - \sum_{b_{\textsc{dl}}} r_{ub}y_{ub} + \epsilon_u \bigr)
\end{equation}
\begin{equation}\label{eq:msa-partialUser-uplink}
\frac{\partial\, g(\boldsymbol\lambda, \boldsymbol\lambda^\prime, \boldsymbol\nu, \boldsymbol\nu^\prime)}{\partial \boldsymbol\lambda^\prime} = \sum_u\bigl(\sum_{b_{\textsc{dl}}} r_{ub}y_{ub} - \sum_{b_{\textsc{ul}}}r^\prime_{ub} y^\prime_{ub}  + \epsilon_u \bigr).
\end{equation}
\end{subequations}

Using~\eqref{eq:msa-partialBS-downlink} and~\eqref{eq:msa-partialBS-uplink}, we update dual variables $\nu_b$ and $\nu^\prime_b$ as follows

\begin{subequations}
\begin{equation}\label{eq:update-bs-multiplier-dl}
\nu_b(t+1) = \bigl[\nu_b(t) - \gamma\,(1 - \sum_{u\in\mathcal{U}(b)} y_{ub}) \bigr],\, \forall b \in \mathcal{B}_{\textsc{dl}}
\end{equation}
\begin{equation}\label{eq:update-bs-multiplier-ul}
\nu^\prime_b(t+1) = \bigl[\nu^\prime_b(t) - \gamma\,(1 - \sum_{u\in\mathcal{U}(b)} y^\prime_{ub}) \bigr],\, \forall b \in \mathcal{B}_{\textsc{ul}},
\end{equation}
\end{subequations}

where $\gamma$ is a sufficiently small positive step size and $t$ denotes the iteration index. Likewise, using~\eqref{eq:msa-partialUser-downlink} and~\eqref{eq:msa-partialUser-uplink} we get

\begin{subequations}
\begin{equation}\label{eq:update-user-multiplier-dl}
\lambda_u(t+1) = \bigl[\lambda_u(t) - \gamma\,\bigl(\sum_{b_{\textsc{ul}}}r^\prime_{ub} y^\prime_{ub} - \sum_{b_{\textsc{dl}}} r_{ub}y_{ub} + \epsilon_u \bigr) \bigr],\, \forall u \in \mathcal{U}
\end{equation}
\begin{equation}\label{eq:update-user-multiplier-ul}
\lambda^ \prime_u(t+1) = \bigl[\lambda^\prime_u(t) - \gamma\, \bigl(\sum_{b_{\textsc{dl}}} r_{ub}y_{ub} - \sum_{b_{\textsc{ul}}}r^\prime_{ub} y^\prime_{ub}  + \epsilon_u \bigr) \bigr],\, \forall u \in \mathcal{U}.
\end{equation}
\end{subequations}

The dual variables $\boldsymbol\lambda, \boldsymbol\lambda^\prime, \boldsymbol\nu, \boldsymbol\nu^\prime$ will converge to the optimal value after a high enough number of iterations and, since the duality gap for this problem has been proved to be zero, the primal variables $y_{ub}^\star(\lambda_u, \lambda^\prime_u, \boldsymbol\nu, \boldsymbol\nu^\prime)$, $y_{ub}^\prime{}^\star(\lambda_u, \lambda^\prime_u, \boldsymbol\nu, \boldsymbol\nu^\prime)$ will also converge to the optimal value.

\subsection{Solution to the subproblems}
The users' subproblems have the general form

\begin{equation}
\max_{x \geq 0} U(\boldsymbol c^T\boldsymbol x) - \boldsymbol d^T\boldsymbol x
\end{equation}

for suitable vectors $\boldsymbol c$ and $\boldsymbol d$, which is clearly a standard convex problem over $\boldsymbol x \in \mathbb{R}^n_+$. Specifically, each user has to solve the following problem for the downlink:

\begin{equation}
\max_{y_{ub} \geq 0} \quad U_{\alpha}(\sum_{b_{\textsc{dl}}}r_{ub}y_{ub}) + \sum_{b_{\textsc{dl}}} \bigl(r_{ub}y_{ub} (\lambda^\prime_u-\lambda_u) - \nu_b y_{ub} \bigr)
\end{equation}

such that

\begin{equation}
y_{ub} \geq 0,\,\forall b \in \mathcal{B_{\textsc{dl}}}
\end{equation}

Its optimal solution may be characterised by computing the Karush-Kuhn-Tucker (KKT) conditions
for the Lagrangian which, since the problem is convex, are sufficient and necessary for optimality. The Lagrangian is

\begin{equation}
L(\mathbf{y}_1, \boldsymbol\mu) = \frac{(\sum_{b_{\textsc{dl}}} r_{ub}y_{ub})^{1-\alpha}}{1-\alpha} + \sum_{b_{\textsc{dl}}} \bigl(r_{ub}y_{ub} (\lambda^\prime_u-\lambda_u) - \nu_b y_{ub} \bigr) - \sum_{b_{\textsc{dl}}} \mu_b(-y_{ub}).
\end{equation}

From this, the first-order optimality conditions of the problem are

  \begin{equation} \label{eq:subproblem-first-order}
    \begin{cases}
     \frac{\partial L(\mathbf{y}_1, \boldsymbol\mu)}{y_{ub}} = (\sum_{b_{\textsc{dl}}} r_{ub}y_{ub}^ \star)^{-\alpha} r_{ub} + r_{ub}(\lambda^\prime_u-\lambda_u) - \nu_b + \mu_b = 0 \\ 
      \mu_b \geq 0
    \end{cases}, \,\forall b \in  \mathcal{B_{\textsc{dl}}},
  \end{equation}

where $\mu_b = 0$ when $y_{ub} \geq 0$ hold with equality. Therefore, we may rewrite \eqref{eq:subproblem-first-order} as

\begin{equation}\label{eq:final-subproblem-first-order}
\frac{1}{(\sum_{b_{\textsc{dl}}} r_{ub}y_{ub}^ \star)^\alpha} \leq \frac{\nu_b - r_{ub}(\lambda^\prime_u-\lambda_u)}{r_{ub}}, \,\forall b \in  \mathcal{B_{\textsc{dl}}}.
\end{equation}

Unfortunately, the above inequality may have multiple solutions, so it is not immediately clear how to solve algorithmically the user's subproblem. However, note that the \textsc{lhs} of \eqref{eq:final-subproblem-first-order} is common for all $b$. Thus, choosing $(\sum_{b_{\textsc{dl}}} r_{ub}y_{ub}^ \star)^{-\alpha} = \underset{b}{\min} \frac{\nu_b - r_{ub}(\lambda^\prime_u-\lambda_u)}{r_{ub}}$  satisfies all of the \eqref{eq:final-subproblem-first-order} KKT conditions. To that end, let $b_i$ be the base station a user is going to associate to in downlink such that minimises $\frac{\nu_b - r_{ub}(\lambda^\prime_u-\lambda_u)}{r_{ub}}$ . Note that $b_i$ is well defined due to the strict concavity of the utility function. Once $b_i$ is known, the remaining problem is to solve

\begin{equation}
\max_{y_{ub_i} > 0} \quad U_{\alpha}(r_{ub_i} y_{ub_i}) + \bigl(r_{ub_i} y_{ub_i} (\lambda^\prime_u-\lambda_u) - \nu_{b_i} y_{ub_i} \bigr)
\end{equation}

such that $y_{ub_i} > 0$, where $r_{ub_i}$ and $\nu_{b_i}$ are the measured \textsc{sinr} to \textsc{bs} $b_i$ and the base station multiplier (i.e., hidden price), respectively. The first-order optimality conditions become

\begin{subequations}
\begin{equation}
\frac{\partial L(y_{ub_i}, \mu)}{y_{ub_i}} = (r_{ub_i} y_{ub_i}^\star)^{-\alpha} r_{ub_i} + r_{ub_i}(\lambda^\prime_u-\lambda_u) - \nu_{b_i} + \mu = 0 
\end{equation}

\begin{equation} \label{eq:subproblem-primal-feasibility}
-y_{ub_i} \leq 0
\end{equation}

\begin{equation}\label{eq:subproblem-dual-fesibility}
\mu \geq 0
\end{equation}

\begin{equation}\label{eq:subproblem-complementary-slackness}
\mu(-y_{ub_i}) = 0
\end{equation}
\end{subequations}

Assuming, $y_{ub_i}^\star > 0$ partially satisfies primal feasibility \eqref{eq:subproblem-primal-feasibility}. In addition, setting $\mu = 0$ satisfies dual feasibility \eqref{eq:subproblem-dual-fesibility} and complementary slackness \eqref{eq:subproblem-complementary-slackness}. The stationary condition yields

\begin{equation}\label{eq:MSA_User_Alloc_DL}
y_{ub_i}^\star = \bigl( \frac{r_{ub_i}^{1-\alpha}}{\nu_{b_i} - r_{ub_i}(\lambda^\prime_u-\lambda_u)} \bigr)^{1/\alpha}.
\end{equation}

Following the same reasoning for the uplink, we get

\begin{equation}\label{eq:MSA_User_Alloc_UL}
y_{ub_j}^\prime{}^\star = \bigl( \frac{r_{ub_j}^\prime{}^{1-\alpha}}{\nu^\prime_{b_j} - r_{ub_j}^\prime(\lambda_u - \lambda^\prime_u)} \bigr)^{1/\alpha},
\end{equation}
where $b_j$ is the base station which minimises $\frac{\nu^\prime_b - r^\prime_{ub}(\lambda_u - \lambda^\prime_u)}{r^\prime_{ub}}$.

\subsection{Computational complexity}
In view of individual subproblems \eqref{eq:decentralised_lower_subproblem-dl} and \eqref{eq:decentralised_lower_subproblem-ul}, the amount of information a given user $u$ needs in order to find its optimal share of resources for the downlink is:

\begin{enumerate}
\item Its own two \emph{user multipliers} (prices) $\lambda_u$ and $\lambda^\prime_u$ for the downlink and the uplink channels, respectively.
\item The vector $\boldsymbol \nu$ with all the prices of every base station which is available to serve user $u$ in downlink.
\end{enumerate}

Similarly, the information that user $u$ must know so as to calculate its optimal share of uplink resources is:

\begin{enumerate}
\item Its own two \emph{user multipliers} (prices) $\lambda_u$ and $\lambda^\prime_u$ for the downlink and the uplink channels, respectively.
\item The vector $\boldsymbol \nu^\prime$ with all the prices of every base station which is available to serve user $u$ in uplink.
\end{enumerate}

Furthermore, the maximum transmission rates $r_{ub}$ and $r^\prime_{ub}$ can be directly computed by the user using local measurements only, e.g., by estimating the \textsc{sinr} in the uplink and downlink channels at the current time and averaging over an appropriate timescale so as to filter out fast fading. Finally, user $u$ needs to update its own multipliers $\lambda_u$ and $\lambda^\prime_u$ after each iteration, using \eqref{eq:update-user-multiplier-dl} and \eqref{eq:update-user-multiplier-ul} \emph{gradient projection methods}.

Regarding the base stations, they must update both their downlink and uplink prices (multipliers) making use of \eqref{eq:update-bs-multiplier-dl} and \eqref{eq:update-bs-multiplier-ul} after each allocation round.

\section{Numerical results}

In this section, we provide the numerical results which support the validity and the performance of the algorithm explained in the last section. We begin by drawing the simulation test bed in which we performed all the simulations. Then, we present how the algorithm behaves in different scenarios, highlighting the main features of the implemented solution and identifying additional aspects which should be taken into account in real deployment.

\subsection{Test scenario}
Having in mind the high number of feasible combinations while associating users to base stations and resources to users even in small-sized networks, the optimal configuration for each one of the elements of the system is not readily recognizable. With the aim of presenting the characteristics and strengths of the algorithm in a more suitable and friendly way, we test the proposed approach in a custom deployment with a few base stations and users. For the simulation, we model the locations of the base stations and users to be fixed so as to control the signal-to-noise ratio (\textsc{sinr}) that every user is achieving from each base station. This enables us to easily validate the behaviour of the solution since all the system's parameters are deterministic and do not depend on a random deployment. In addition, we assume that each base station is capable of serving users in both uplink and downlink.

\subsection{Test cases}
Firstly, we shall enumerate the tests which have been performed under the new association and resource allocation algorithm. They are listed briefly below, while subsequent sections go into depth on each case.

\begin{enumerate}
\item Test 1 - Oscillations in the optimal solution.
\item Test 2 - Tightness of inequalities.
\item Test 3 - $\alpha$ - fairness value.
\item Test 4 - Uplink - Downlink decoupling (\textsc{dud}e).
\item Test 5 - Load balance.
\item Test 6 - Addition of new base stations.
\end{enumerate}

The numerical results of the above-mentioned tests are accompanied by some graphical results, showing the evolutions of the \textsc{bs}' multipliers. The chosen step for the \emph{gradient projection method} of both \textsc{bs}s and users' multipliers was $\gamma = 0.004$. In addition, $8000$ iterations have been proven to be enough iterations for the algorithm to converge in all test cases. Note that these analyses have been accomplished in a single computer despite the fact that this is a distributed nature algorithm\footnote{Even though the computer where the tests have been performed has a multi-core processor, \textsc{matlab} uses only one core by default.}. Namely,  this number of iterations takes less than 2 seconds to finish, since calculations are rather simple. Nevertheless, for more details on the speed and convergence of the algorithm we refer the reader to section \ref{sec:speed}. However, note that this can be done un a decentralised way since users' subproblems are independent.

\subsection{Performance evaluation}
\subsubsection*{Test 1 - Oscillations in the optimal solution}
In this example, we show the oscillation of the global optimal solution whenever the rates perceived by a user from different base stations are similar. Let $\text{Rates}_{\textsc{dl}}{}_{|\mathcal{U}| \times |\mathcal{B}_{\textsc{dl}}|}$ be a matrix containing the maximum rate values at which each user can transmit to each one of the base stations in downlink (in $bits/s/Hz$) and let $\text{Rates}_{\textsc{ul}}{}_{|\mathcal{U}| \times |\mathcal{B}_{\textsc{ul}}|}$ be the matrix containing the values for the uplink. For this simulation, we assume $4$ users and $3$ base stations. Focusing on the uplink, note that user $\#1$ (first row of \ref{eq:Test1_ratesUL}) may achieve similar performance from any of the three base stations. In addition, users $2$ and $4$ will presumably be associated to base station $\#2$.

\begin{equation}\label{eq:Test1_ratesUL}
 \text{Rates}_{\textsc{ul}} = \begin{pmatrix}
    28 & 30 & 28 \\
    0.5 & 15 & 1 \\
    30 & 1 & 5.2 \\
    0.3 & 32 & 0.5
  \end{pmatrix}
\end{equation}

Figure \ref{fig:test1_ULMult} shows the evolution of the uplink multipliers. We observe that the uplink multiplier for \textsc{bs} $\#2$ has converged to an stable value. Conversely, \textsc{bs}'s uplink multiplier has not converged for \textsc{bs}s $\#1$ and $\#3$.

\begin{figure}[!htb]
  \begin{center}
    \includegraphics[scale=0.7]{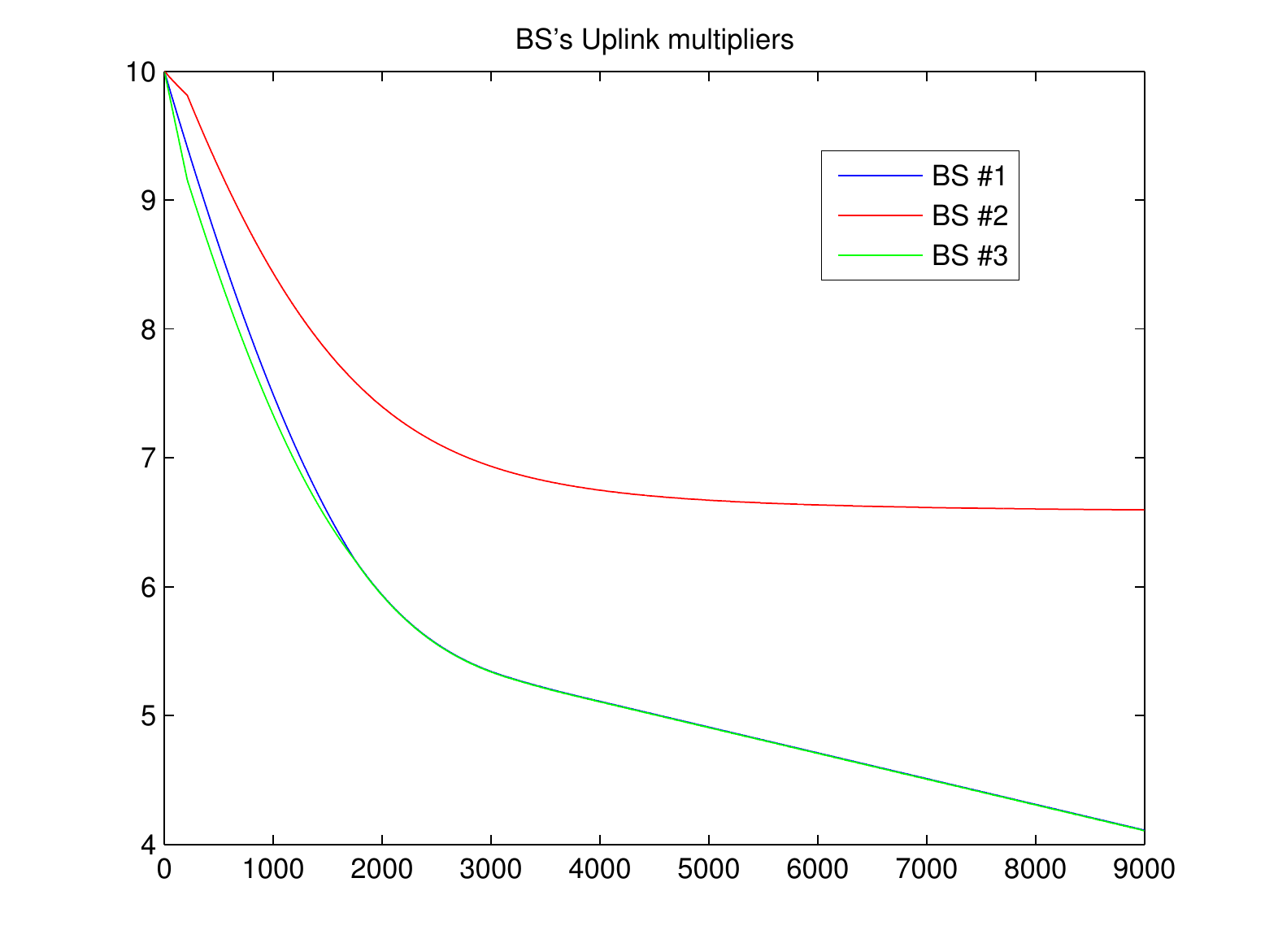}
    \caption{Base stations' uplink multipliers.}
    \label{fig:test1_ULMult}
  \end{center}
\end{figure}

This is due to the fact that User $\#1$ is constantly jumping between base station $\#1$ and $\#3$. User $\#1$ associates \textsc{bs} $\#1$ which contributes to rise the price of that base station. In the meantime, the multiplier of base station $\#3$ is reducing its value. Hence, User $\#1$ decided to switch to base station $\#3$ and the process starts over again. We can observe this in figure \ref{fig:test1_ULMult_detail}. There exist several solutions to this undesirable behaviour such as establishing a hysteresis model for changing associations, i.e, a user does not associate to a different base station if the gain does not exceed a given threshold. Another option would be setting a guard time during which a user does not consider an association change. The best choice depends, however, on many factors and the decision should be left to the operator. Despite this, note that this scenario might not be typical. We assume that in a real deployment, \textsc{bs}s will be far enough from each other and the ripple is more likely to appear at the cell edges, which is a minor part of the deployment space.

\begin{figure}[!htb]
  \begin{center}
    \includegraphics[scale=0.7]{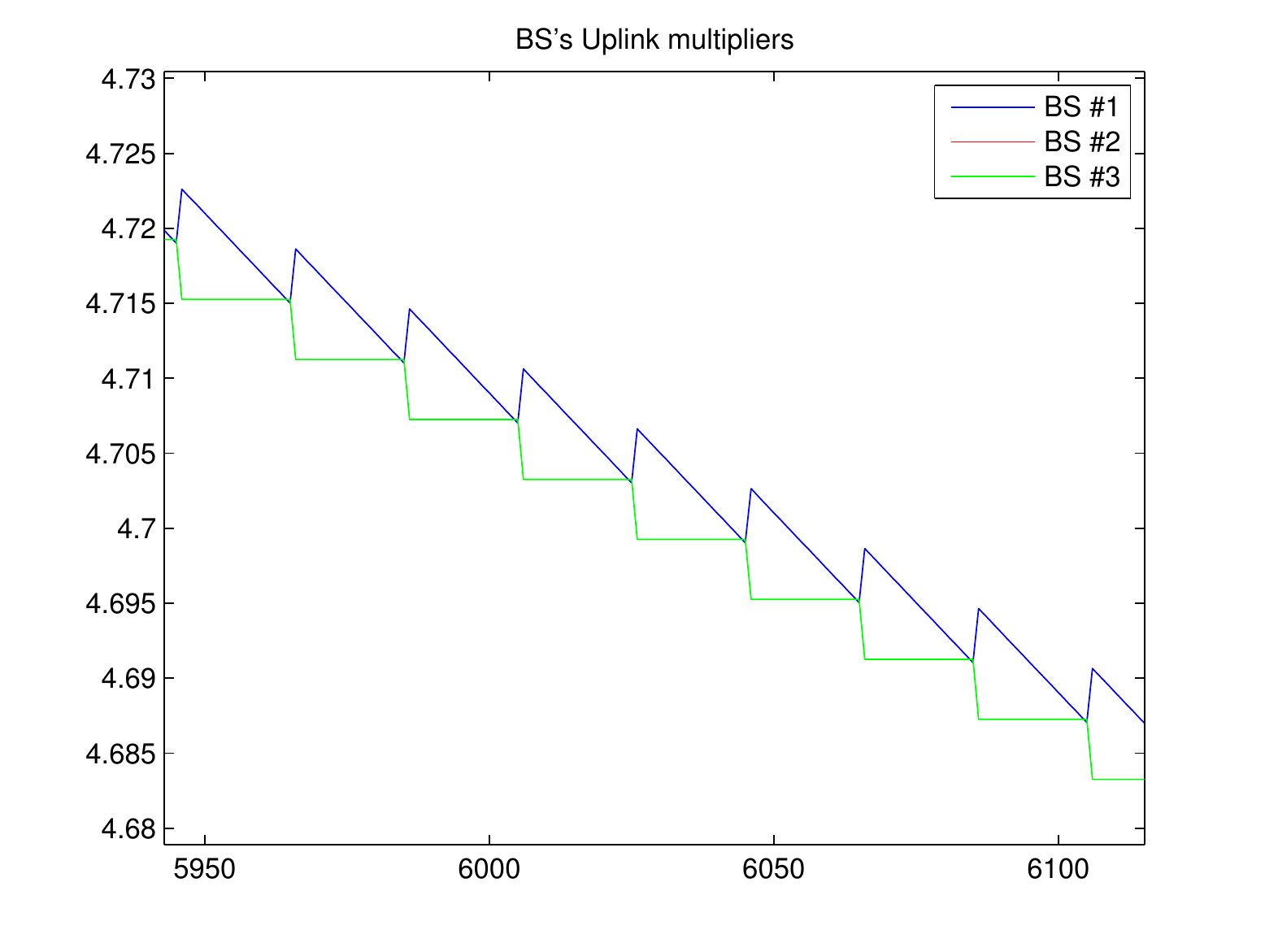}
    \caption{Base stations' uplink multipliers oscillation.}
    \label{fig:test1_ULMult_detail}
  \end{center}
\end{figure}

\subsubsection*{Test 2 - Tightness of inequalities}
The aim of this test is to check how strict the \eqref{eq:constraint-msa-absA} - \eqref{eq:constraint-msa-absB} and \eqref{eq:constraint-msa-dl} - \eqref{eq:constraint-msa-ul} constraints are. After a little thought, we realise that if the non-negative value $\epsilon_u$ is too small, the set of equations which establish a maximum amount of rate asymmetry \emph{per user} holds with equality, preventing the base stations serving these users from granting all their resources. For the simulations, we set the fairness parameter $\alpha = 0.5$ and the per-user asymmetry parameter $\epsilon_u = 2$ (we use the same value for each one of the four users). In addition, the rate matrices for uplink and downlink are shown below:

\begin{equation}
\text{Rates}_{\textsc{dl}} = \begin{pmatrix}
    8 & 1 & 29 \\
    0.5 & 15 & 1 \\
    25 & 2 & 2 \\
    8 & 28 & 0.9
  \end{pmatrix};\quad
\text{Rates}_{\textsc{ul}} = \begin{pmatrix}
    2 & 1 & 25 \\
    0.5 & 15 & 1 \\
    30 & 1 & 5.2 \\
    0.3 & 32 & 0.5
  \end{pmatrix}
\end{equation}

After running the simulation, we get the following results, depicted in Figs.~\ref{fig:test2_multipliers} and~\ref{fig:test2_allocations}

\begin{figure}[!htb]
\centering     
\subfigure[Downlink multipliers.]{\label{fig:test2_DLMult}\includegraphics[scale=0.46]{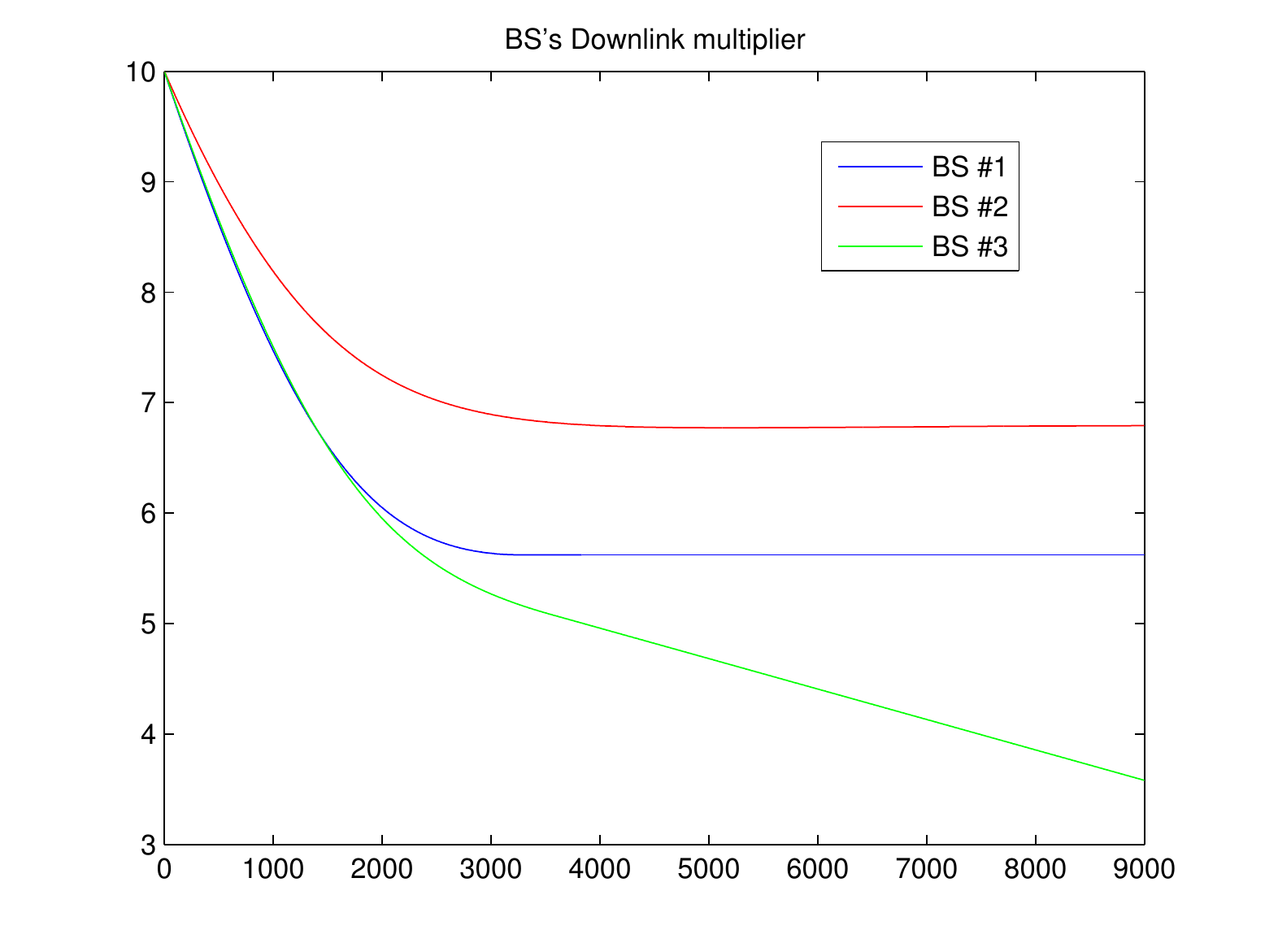}}
\subfigure[Uplink multipliers.]{\label{fig:test2_ULMult}\includegraphics[scale=0.46]{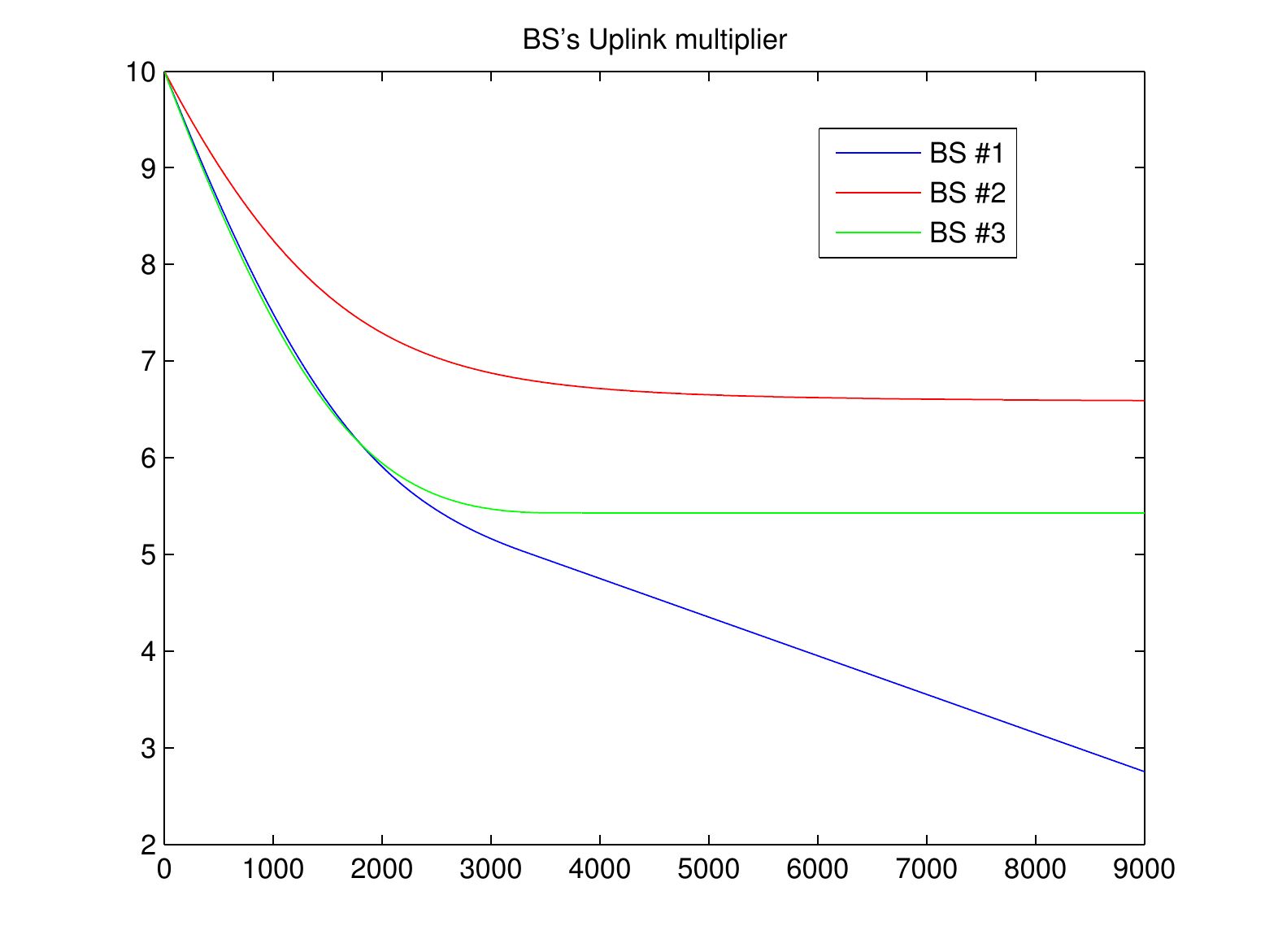}}
\caption{BSs' multipliers.}
\label{fig:test2_multipliers}
\end{figure}

\begin{figure}[!htb]
\centering     
\subfigure[Downlink allocations.]{\label{fig:test2_DLAlloc}\includegraphics[scale=0.46]{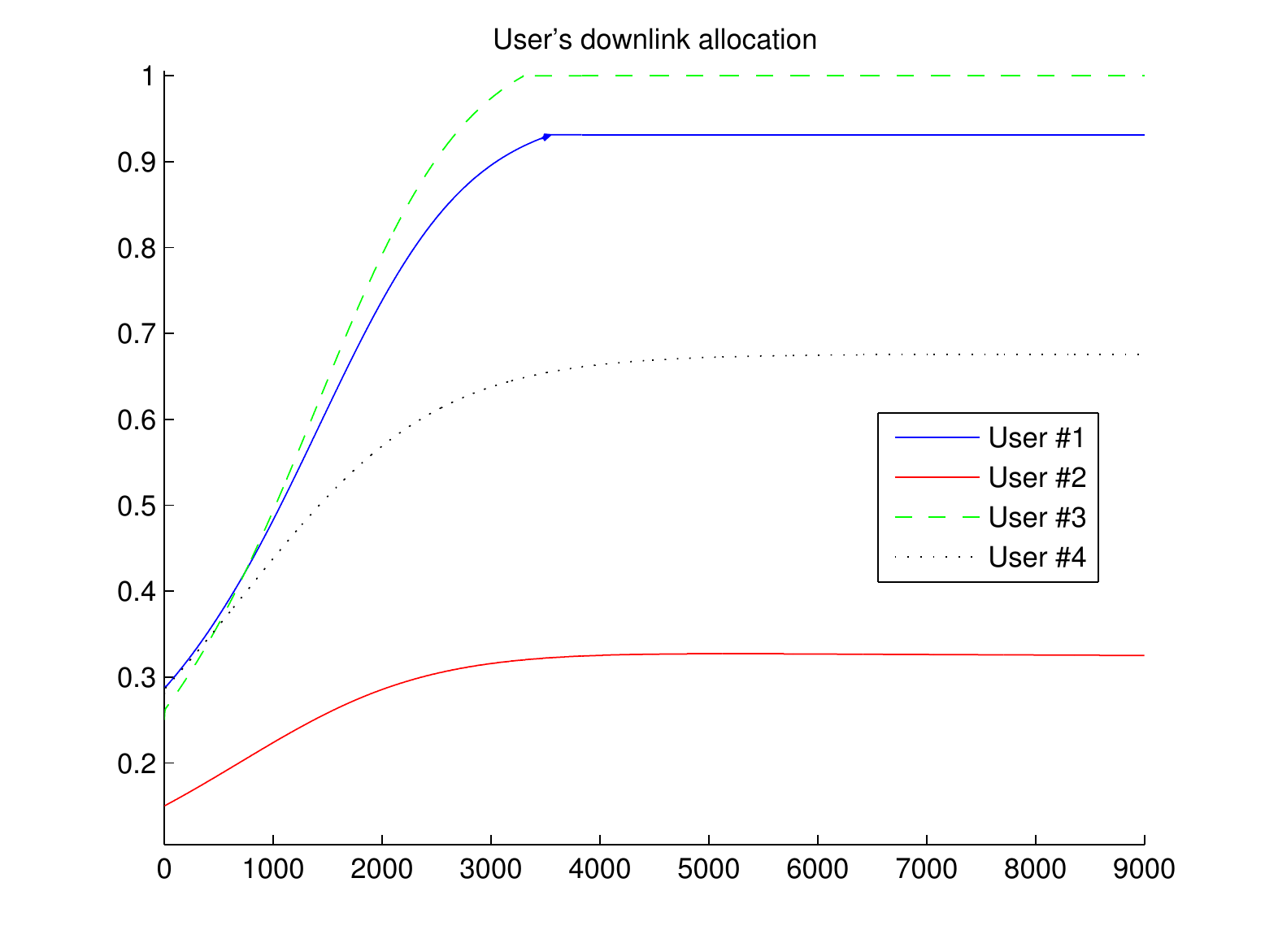}}
\subfigure[Uplink allocations.]{\label{fig:test2_ULAlloc}\includegraphics[scale=0.46]{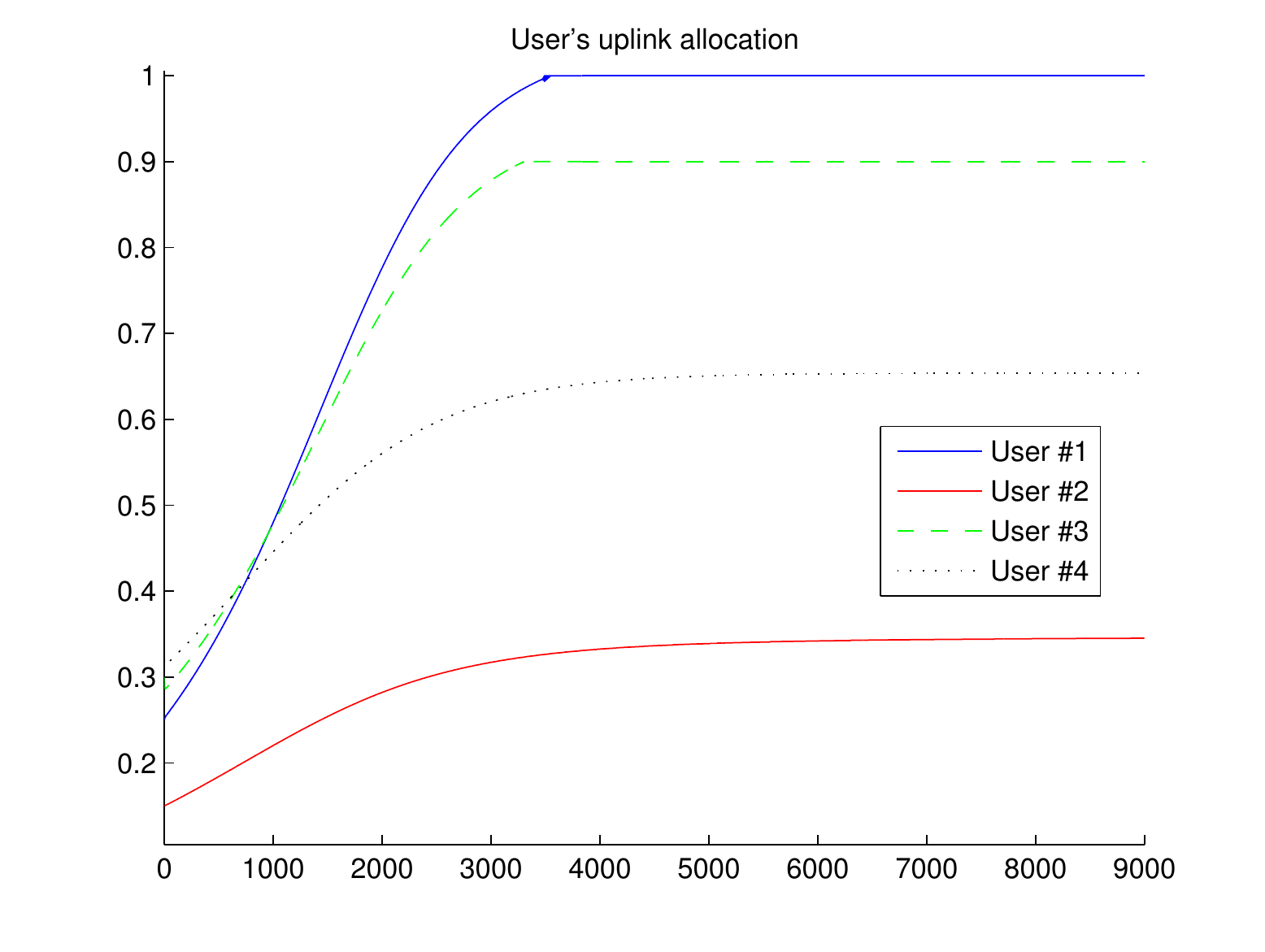}}
\caption{Users' allocations.}
\label{fig:test2_allocations}
\end{figure}

Besides, final allocation values are

\begin{equation}\label{eq:test2_allocations}
\text{Allocation}_{\textsc{dl}} = \begin{pmatrix}
    0 & 0 & 0.9311 \\
    0 & 0.3252 & 0 \\
    1 & 0 & 0 \\
    0 & 0.6757 & 0
  \end{pmatrix};
\text{Allocation}_{\textsc{ul}} = \begin{pmatrix}
    0 & 0 & 1 \\
    0 & 0.3452 & 0 \\
    0.9001 & 0 & 0 \\
    0 & 0.6537 & 0
  \end{pmatrix}
\end{equation}

As shown in figure~\ref{fig:test2_allocations}, downlink allocation for User$\#1$ does not converge to one even though it is the only user in \textsc{bs}$\#3$, regarding downlink allocation matrix in \eqref{eq:test2_allocations}. The same idea holds true for User$\#3$ in the uplink at \textsc{bs}$\#1$. This behaviour is explained by the value of $\epsilon_u$ which hinders the access to the whole pool of resources. If we examine the numerical results carefully, we may clarify where these convergence points come from. Understanding the effective downlink rate for User$\#1$  as the maximum rate multiplied by the allocation, i.e.,  $\text{User\#1}_\text{DLrate} = \text{Rates}_{\textsc{dl}}(1,\, Chosen\, \textsc{bs}) \cdot \text{Allocation}_{\textsc{dl}}(1,\, Chosen\, \textsc{bs})$, the effective downlink and uplink rates for User $\#1$ are:

\begin{subequations}
\begin{equation}
\textbf{\text{User\#1}}_\text{DLrate} = 29\, [bits/s/Hz]\, \cdot\, 0.9311 = 27.0019 \, [bits/s/Hz].
\end{equation}
\begin{equation}
\textbf{\text{User\#1}}_\text{ULrate} = 25\, [bits/s/Hz]\, \cdot\, 1 = 25 \, [bits/s/Hz].
\end{equation}
\end{subequations}

As the reader may have already noticed, asymmetry for User $\#1$ is $27.0019 - 25 \simeq 2 = \epsilon_u$. Following an analogous procedure for User$\#3$ yields: $|25 - 27.003| \simeq 2 = \epsilon_u$. This is why some base stations are not distributing all their resources. Despite the fact that this behaviour matches the mathematical model, it might not be desirable in a real deployment. As a consequence, the performance of the overall system is degraded. To overcome this issue, we shall make slight modifications to the original algorithm. Namely, we are going to focus on equations \eqref{eq:MSA_User_Alloc_DL} - \eqref{eq:MSA_User_Alloc_UL}. Recall that optimal downlink allocation was

\begin{equation}\label{eq:up}
y_{ub_i}^\star = \bigl( \frac{r_{ub_i}^{1-\alpha}}{\nu_{b_i} - r_{ub_i}(\lambda^\prime_u-\lambda_u)} \bigr)^{1/\alpha}
\end{equation}

As we have illustrated, in some cases the algorithm is not able to grant all the available resources due to the asymmetry constraint. This is caused by the effect of the denominator in \eqref{eq:up}. Continuing the example for User$\#1$, as the excess rate causing the asymmetry points towards the downlink, $\lambda_u$ will be greater than $\lambda^\prime_u$ for this user. Therefore, $r_{ub_i}(\lambda^\prime_u-\lambda_u)$ is going to be more and more negative each iteration, as the base station grants more resources. On the other hand, $\nu_{b_i}$ will keep getting smaller, in order to allocate more resources to that user. Ultimately, the \textsc{lhs} and the \textsc{rhs} of the denominator reach a point of equillibrium in which each one will compensate any change on the other so as to enforce the asymmetry constraint. To avoid these problems, we confine the effect of the user's multipliers to the decision of choosing a base station since they have no redeeming features on the resource allocation process. Hence, focusing on the downlink, the decision of which \textsc{bs} to associate to remains the same, i.e., a user will associate to the \textsc{bs} in downlink which minimises

\begin{equation*}
\frac{\nu_b - r_{ub}(\lambda^\prime_u-\lambda_u)}{r_{ub}}.
\end{equation*}

Conversely, the optimal allocation for a user in downlink becomes

\begin{equation}
y_{ub_i}^\star = \bigl( \frac{r_{ub_i}^{1-\alpha}}{\nu_{b_i}} \bigr)^{1/\alpha}
\end{equation}

This way, we can assure that all the resources will be used and we take into account rate asymmetry while choosing a base station. The same reasoning applies to the uplink.

\subsubsection*{Test 3 - $\alpha$ - fairness value}
We now focus on exploring the effect of the fairness parameter on the system performance, especially in those cases where a base station is serving more than one user a the same time. Let $\epsilon_u = 2$ be the asymmetry parameter which is the same for all users. Also, uplink and downlink rate matrices are

\begin{equation}
\text{Rates}_{\textsc{dl}} = \begin{pmatrix}
    8 & 1 & 29 \\
    0.5 & 15 & 1 \\
    25 & 2 & 2 \\
    8 & 28 & 0.9
  \end{pmatrix};\quad
\text{Rates}_{\textsc{ul}} = \begin{pmatrix}
    8 & 1 & 25 \\
    0.5 & 15 & 1 \\
    30 & 1 & 5.2 \\
    0.3 & 32 & 0.5
  \end{pmatrix},
\end{equation}

which are pretty similar to those of Test 2 and will remain unchanged for the three subtests with $\alpha = 0.5$, $\alpha = 1$ and $\alpha = 2$. Firstly, we will observe the resource allocation when the $\alpha$ - fairness parameter is set to $0.5$. Figures~\ref{fig:test3_multipliers_half} and~\ref{fig:test3_allocations_half} show the evolution of the multipliers and the amount of resources granted to each user, respectively.

\begin{figure}[!htb]
\centering     
\subfigure[Downlink multipliers.]{\label{fig:test3_DLMult_half}\includegraphics[scale=0.46]{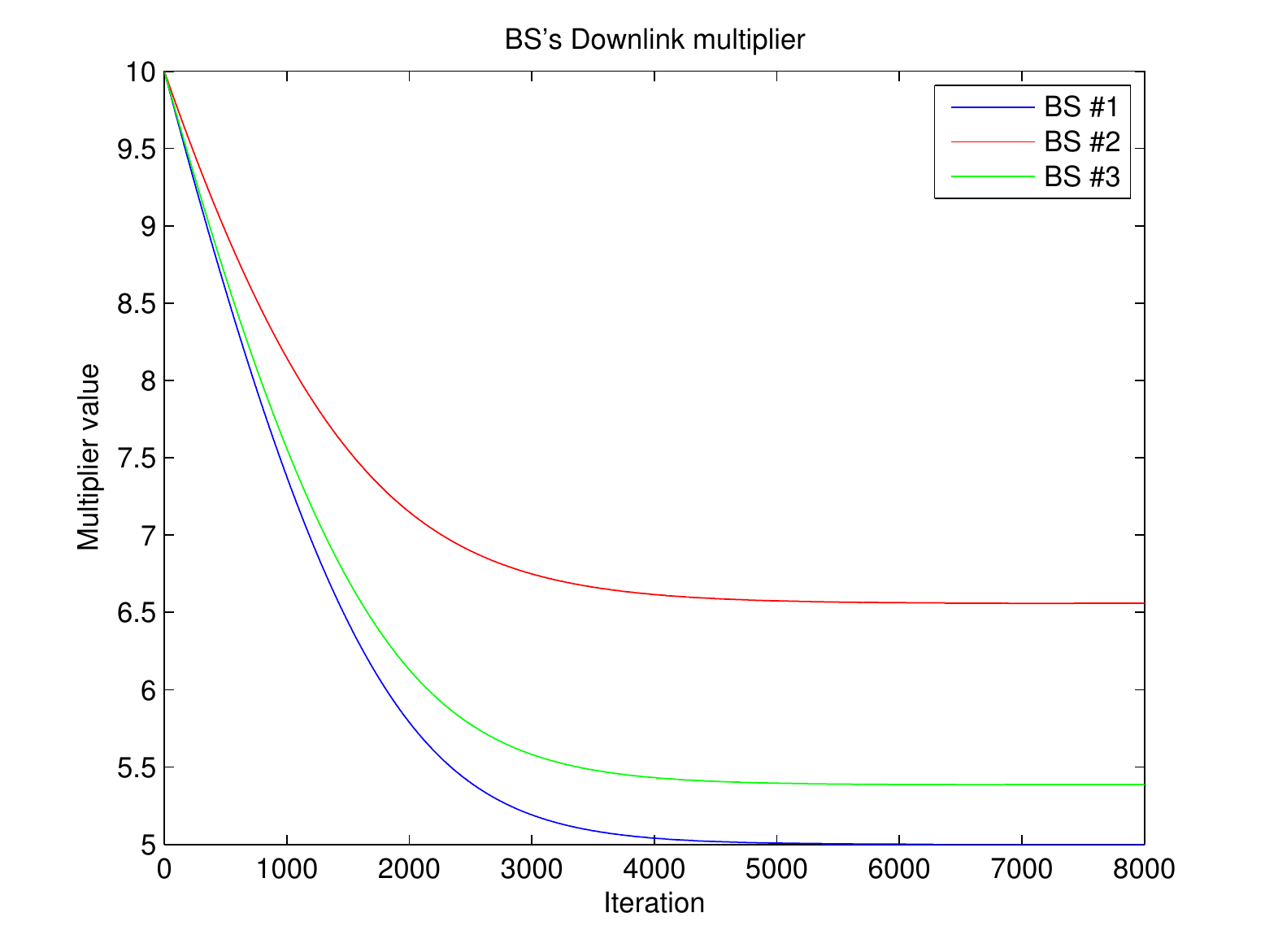}}
\subfigure[Uplink multipliers.]{\label{fig:test3_ULMult_half}\includegraphics[scale=0.46]{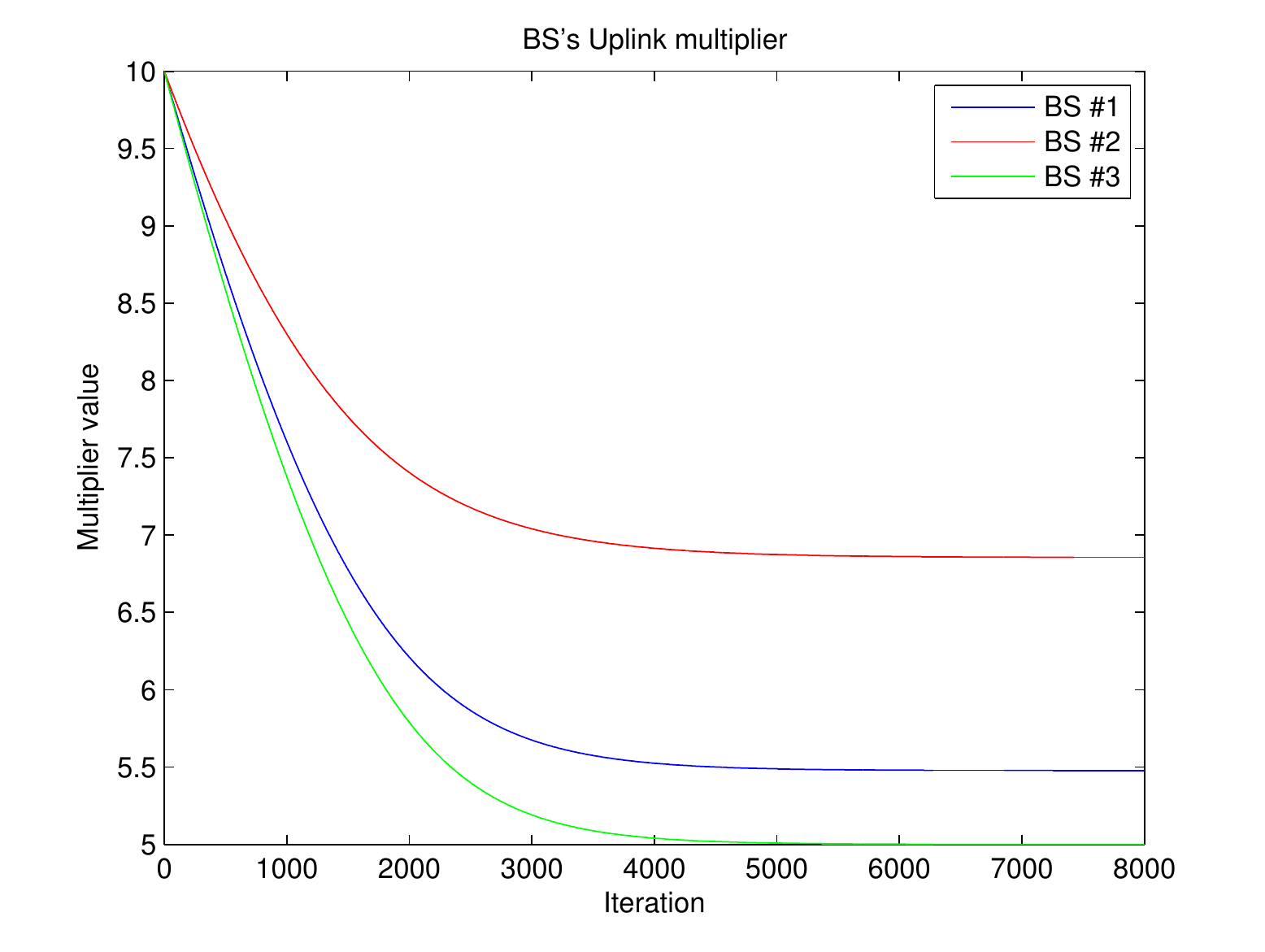}}
\caption{BSs' multipliers.}
\label{fig:test3_multipliers_half}
\end{figure}

\begin{figure}[!htb]
\centering     
\subfigure[Downlink allocations.]{\label{fig:test3_DLAlloc_half}\includegraphics[scale=0.46]{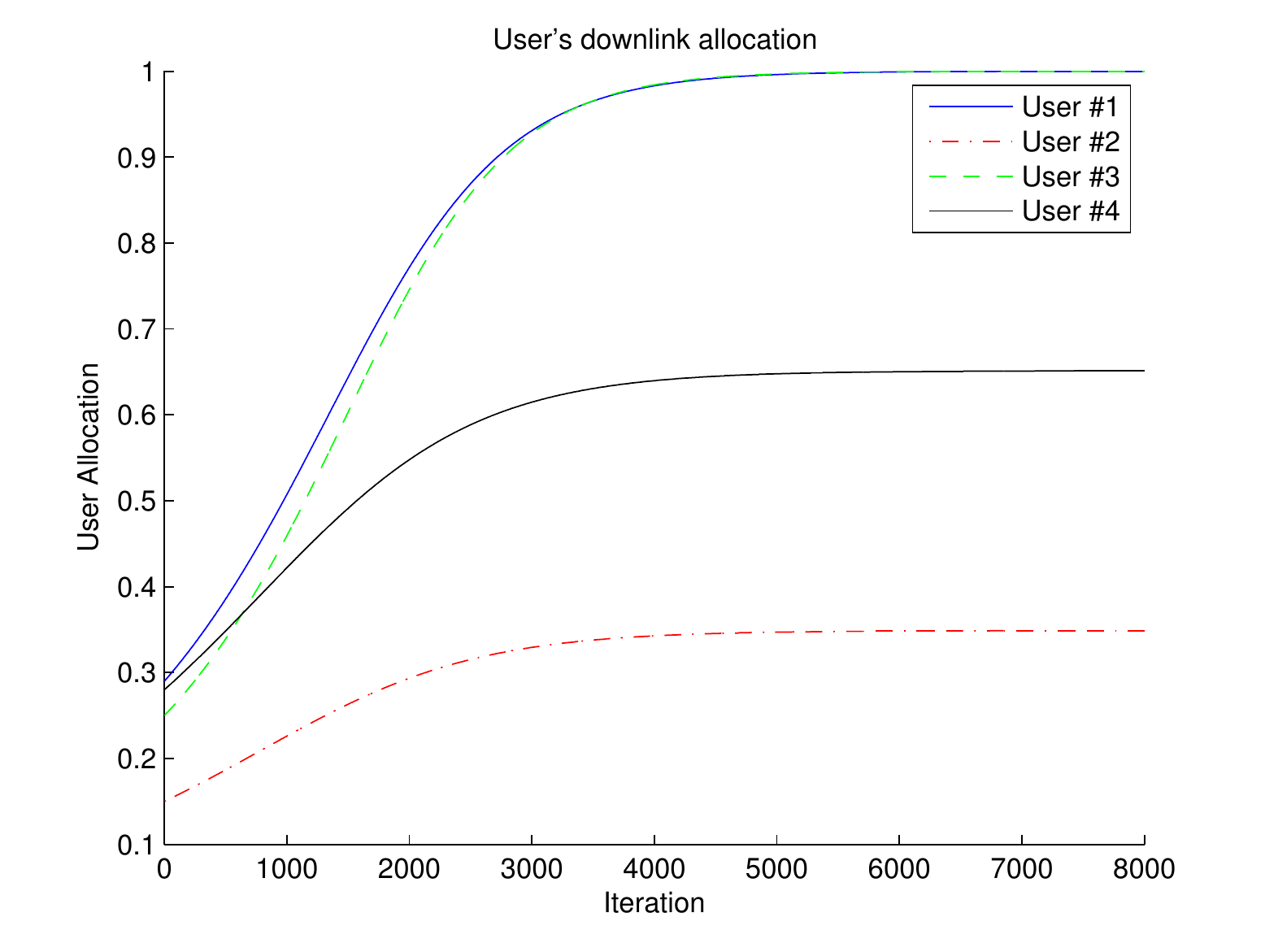}}
\subfigure[Uplink allocations.]{\label{fig:test3_ULAlloc_half}\includegraphics[scale=0.46]{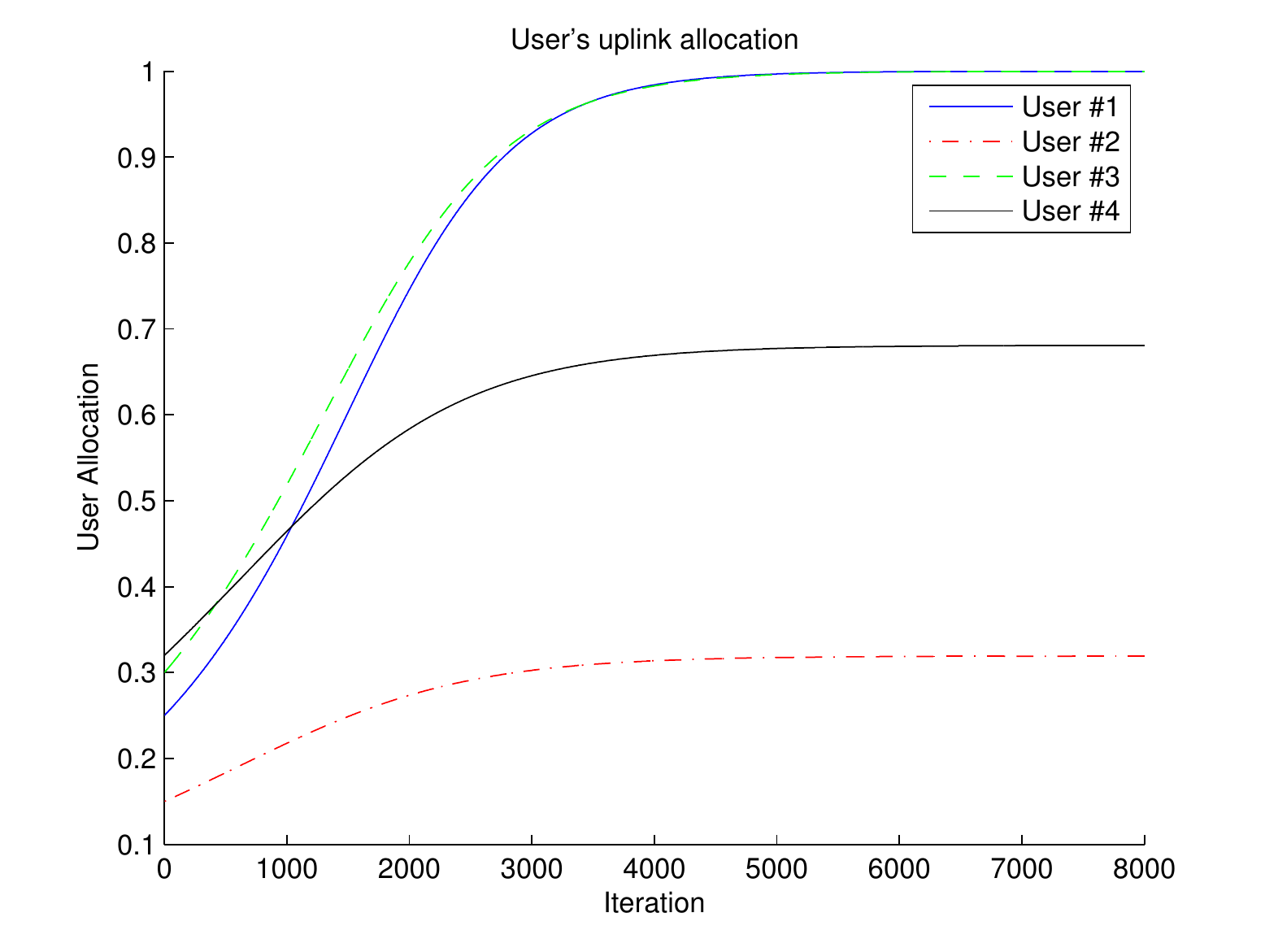}}
\caption{Users' allocations.}
\label{fig:test3_allocations_half}
\end{figure}

Note that \textsc{bs}s' multipliers have converged to a stable value for both links. In addition, it is worth mentioning that now, all the resources are being used without affecting the association decision. Regarding the resouce allocation, we shall inspect the allocation matrices.

\vspace*{-0.5cm}
\begin{equation}\label{eq:test3_allocations}
\text{Allocation}_{\textsc{dl}} = \begin{pmatrix}
    0 & 0 & 1 \\
    0 & 0.3488 & 0 \\
    1 & 0 & 0 \\
    0 & 0.6511 & 0
  \end{pmatrix};\quad
\text{Allocation}_{\textsc{ul}} = \begin{pmatrix}
    0 & 0 & 1 \\
    0 & 0.3191 & 0 \\
    1 & 0 & 0 \\
    0 & 0.6807 & 0
  \end{pmatrix}.
\end{equation}

Observe that \textsc{bs}$\#2$ is serving two of the four users in the system. As we expected since $\alpha < 1$,  we are facing a throughtput maximisation scenario where those users which are perceiving a greater spectral efficiency (user$\#4$) receive a larger amount of resources.

Now, we study how resources are allocated when the $\alpha$ - fairness parameter is equal to $1$. Recall that the rate matrices remain unchanged. As shown in figure~\ref{fig:test3_multipliers_one}, both uplink and downlink multipliers converge to a stable value. Note, in figure~\ref{fig:test3_allocations_one}, that users $\#2$ and $\#4$ are receiving the same amount of resources as a result of equally dividing the available resources of \textsc{bs}$\#2$ among the associated users. This can be easily checked by examining the final allocation matrices.

\begin{figure}[!htb]
\centering     
\subfigure[Downlink multipliers.]{\label{fig:test3_DLMult_one}\includegraphics[scale=0.46]{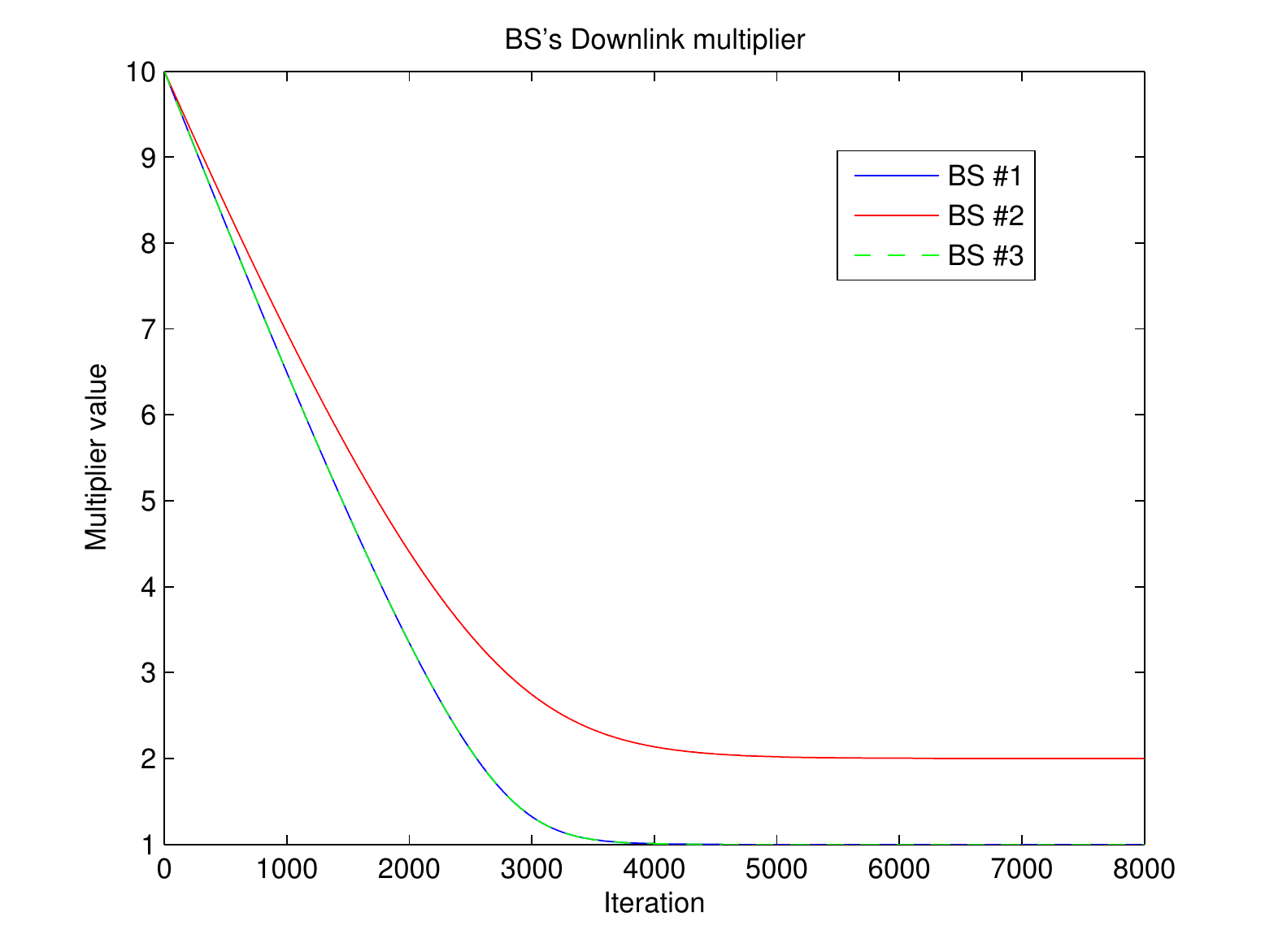}}
\subfigure[Uplink multipliers.]{\label{fig:test3_ULMult_one}\includegraphics[scale=0.46]{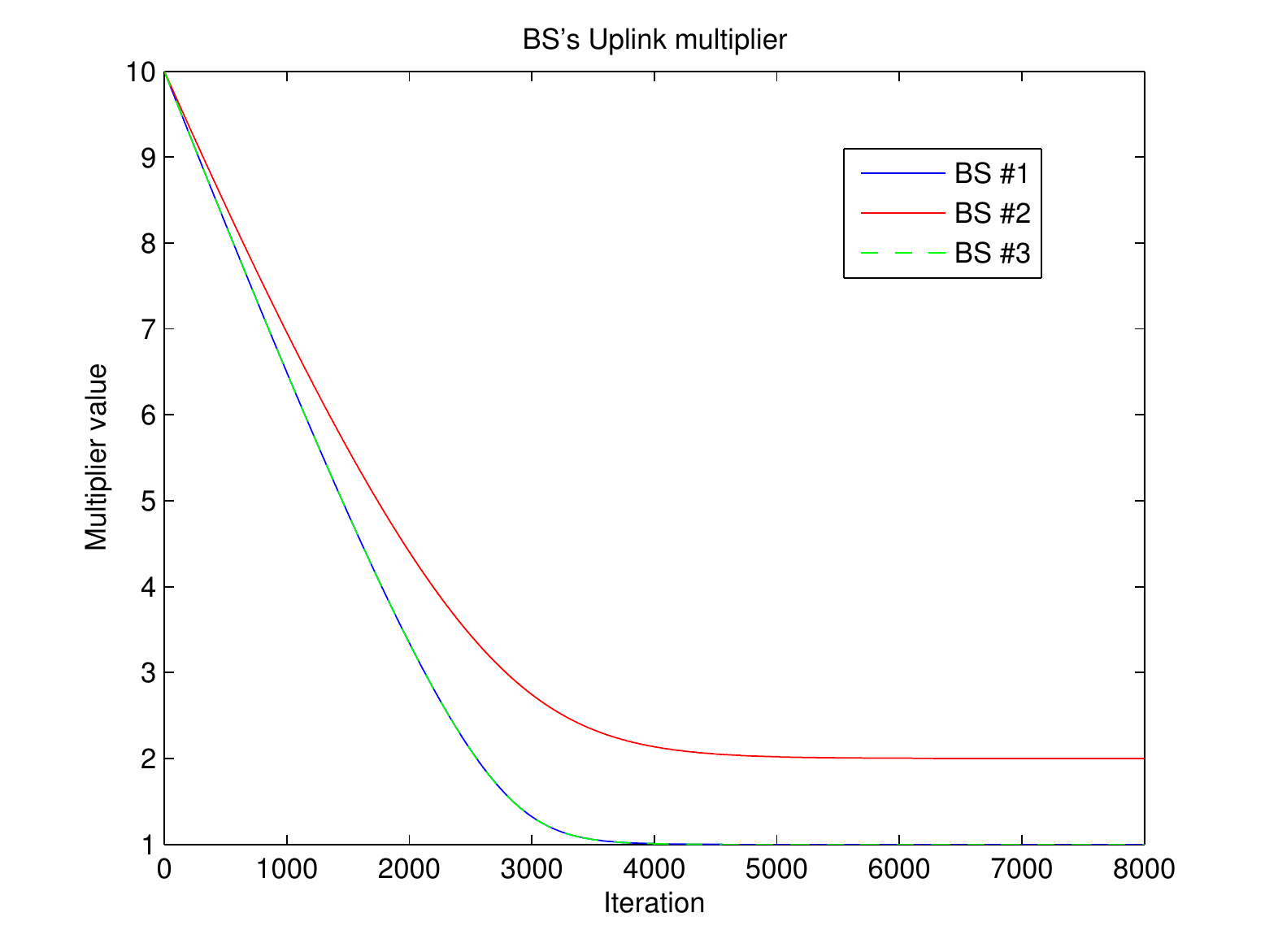}}
\caption{BSs' multipliers.}
\label{fig:test3_multipliers_one}
\end{figure}

\begin{figure}[!htb]
\centering     
\subfigure[Downlink allocations.]{\label{fig:test3_DLAlloc_one}\includegraphics[scale=0.46]{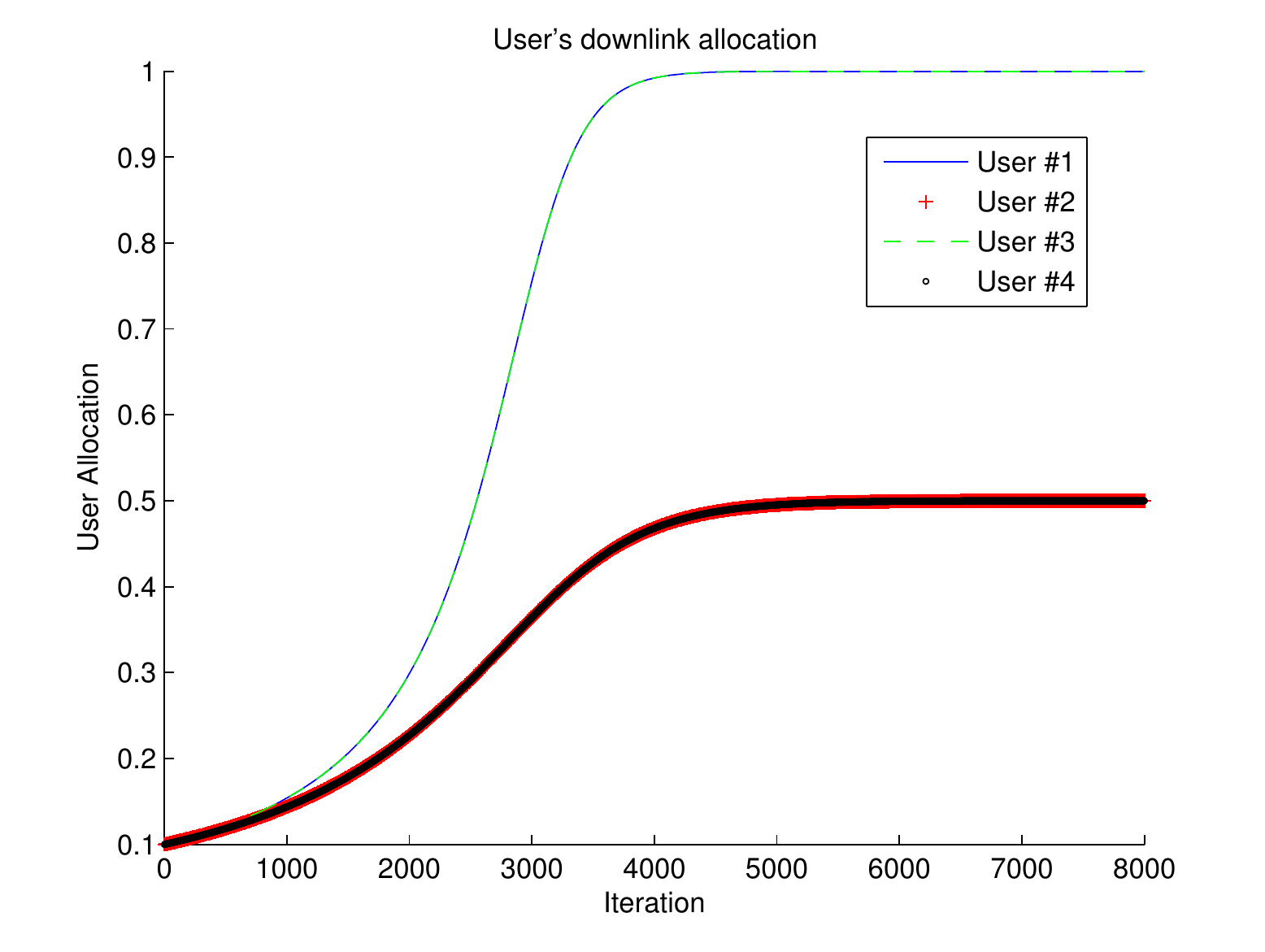}}
\subfigure[Uplink allocations.]{\label{fig:test3_ULAlloc_one}\includegraphics[scale=0.46]{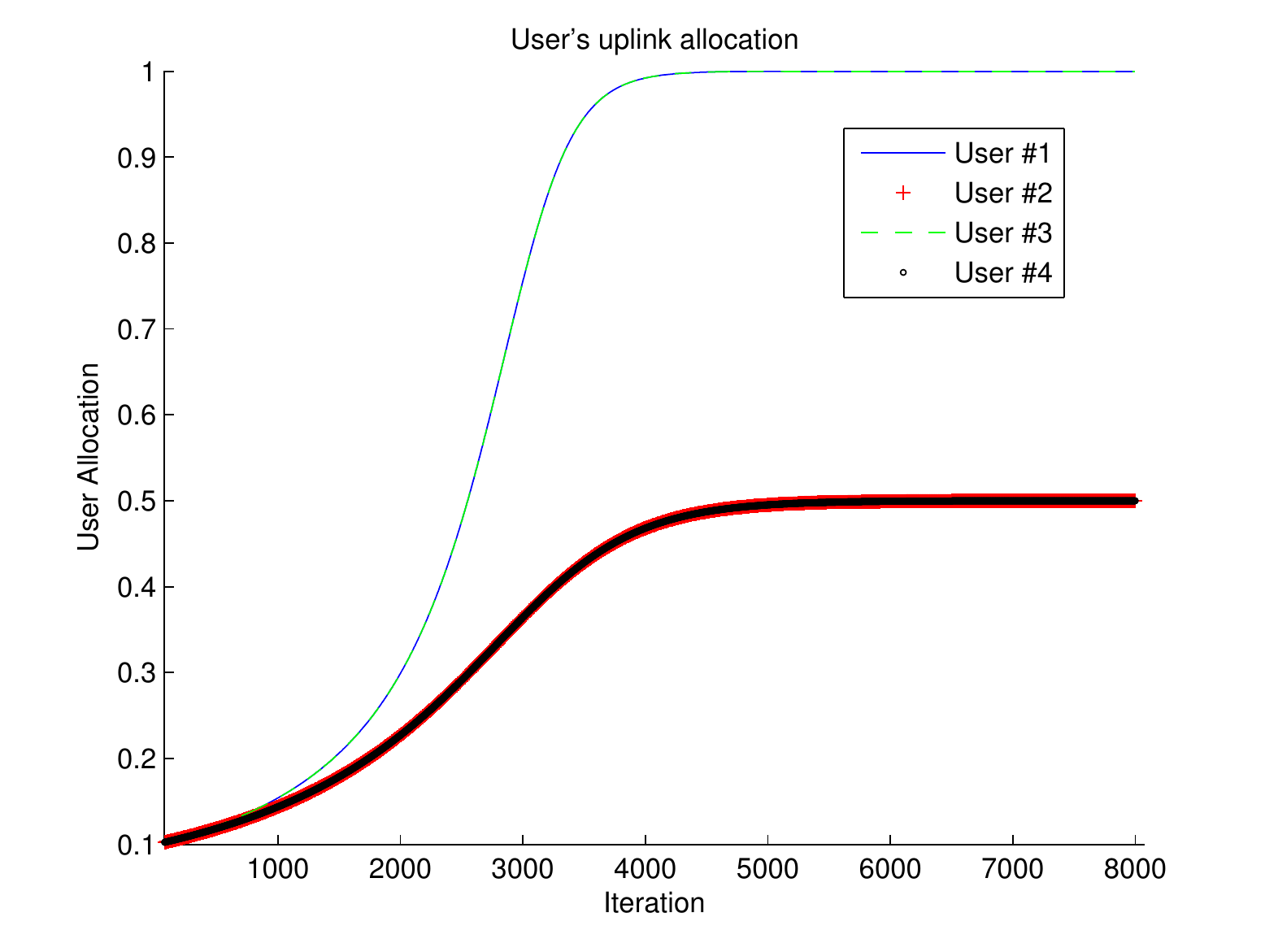}}
\caption{Users' allocations.}
\label{fig:test3_allocations_one}
\end{figure}

\begin{equation}\label{eq:test3_allocations_one}
\text{Allocation}_{\textsc{dl}} = \begin{pmatrix}
    0 & 0 & 1 \\
    0 & 0.5 & 0 \\
    1 & 0 & 0 \\
    0 & 0.5 & 0
  \end{pmatrix};\quad
\text{Allocation}_{\textsc{ul}} = \begin{pmatrix}
    0 & 0 & 1 \\
    0 & 0.5 & 0 \\
    1 & 0 & 0 \\
    0 & 0.5 & 0
  \end{pmatrix}.
\end{equation}

Finally, we evaluate the behaviour of the system when $\alpha > 1$. To that end, we set $\alpha = 2$ and repeat the same experiment. Again, figures~\ref{fig:test3_multipliers_two} and~\ref{fig:test3_allocations_two} show the evolution of the multipliers and allocations through the iterative process, respectively.

\begin{figure}[!htb]
\centering     
\subfigure[Downlink multipliers.]{\label{fig:test3_DLMult_two}\includegraphics[scale=0.46]{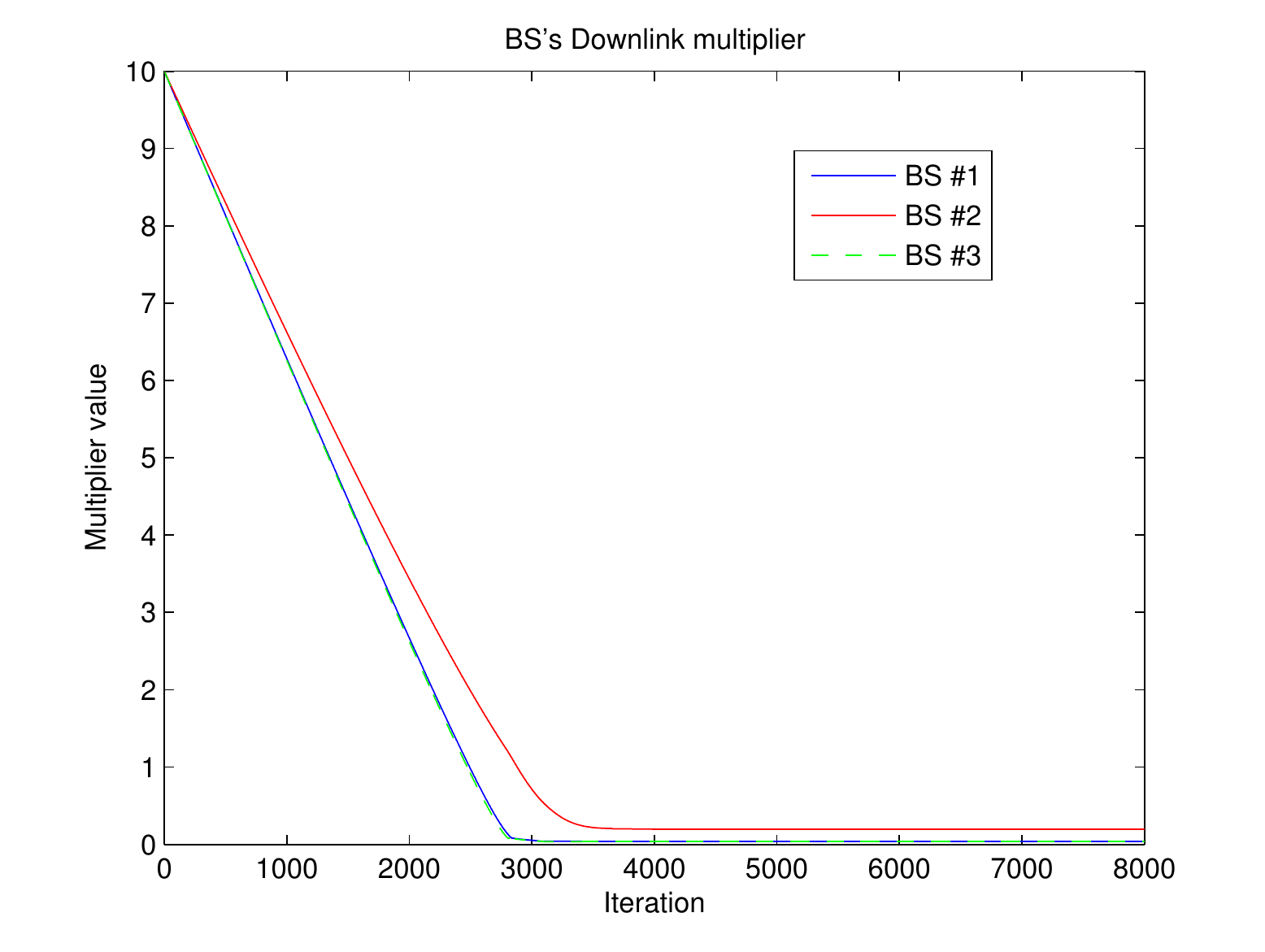}}
\subfigure[Uplink multipliers.]{\label{fig:test3_ULMult_two}\includegraphics[scale=0.46]{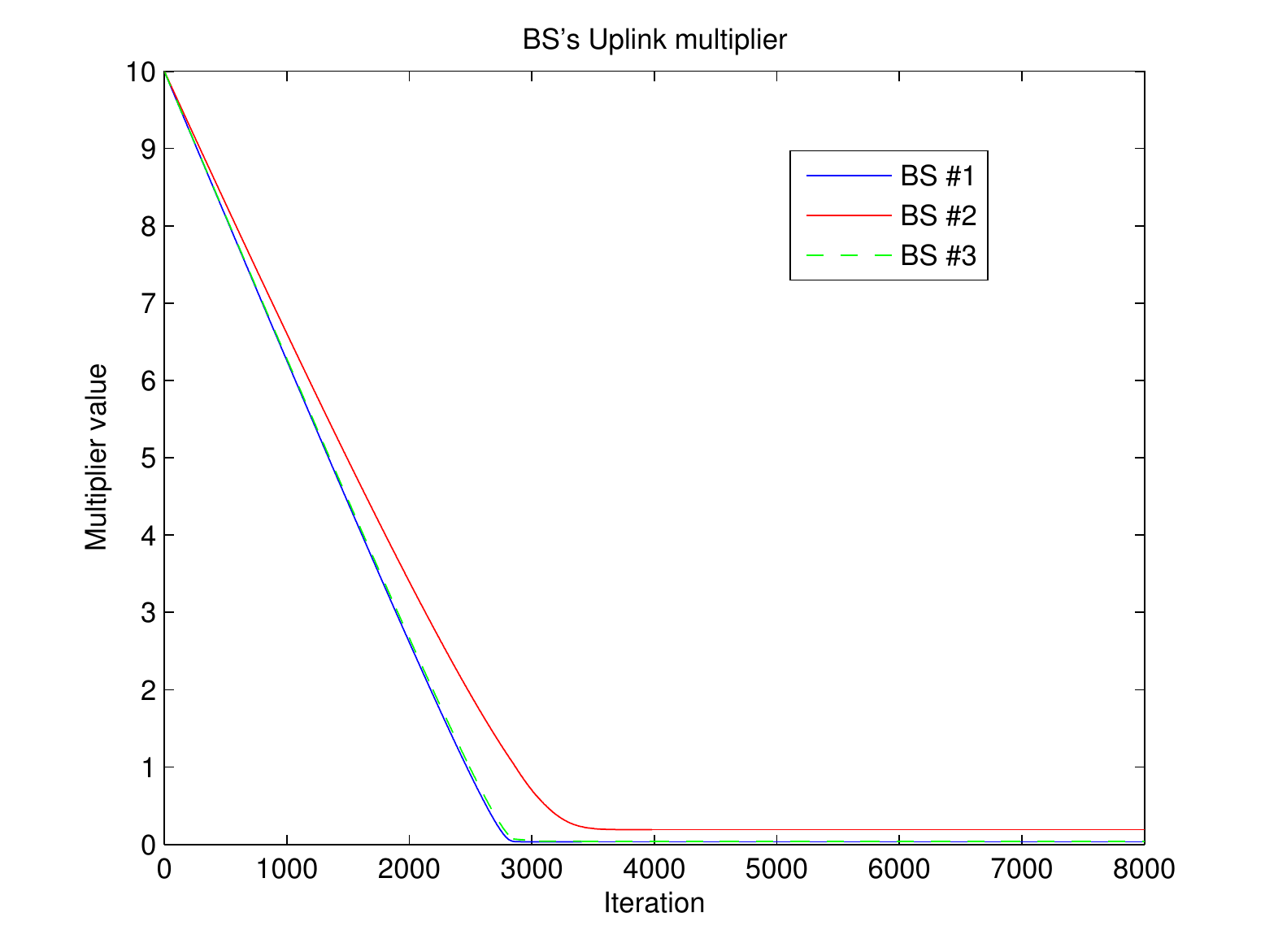}}
\caption{BSs' multipliers.}
\label{fig:test3_multipliers_two}
\end{figure}

\begin{figure}[!htb]
\centering     
\subfigure[Downlink allocations.]{\label{fig:test3_DLAlloc_two}\includegraphics[scale=0.46]{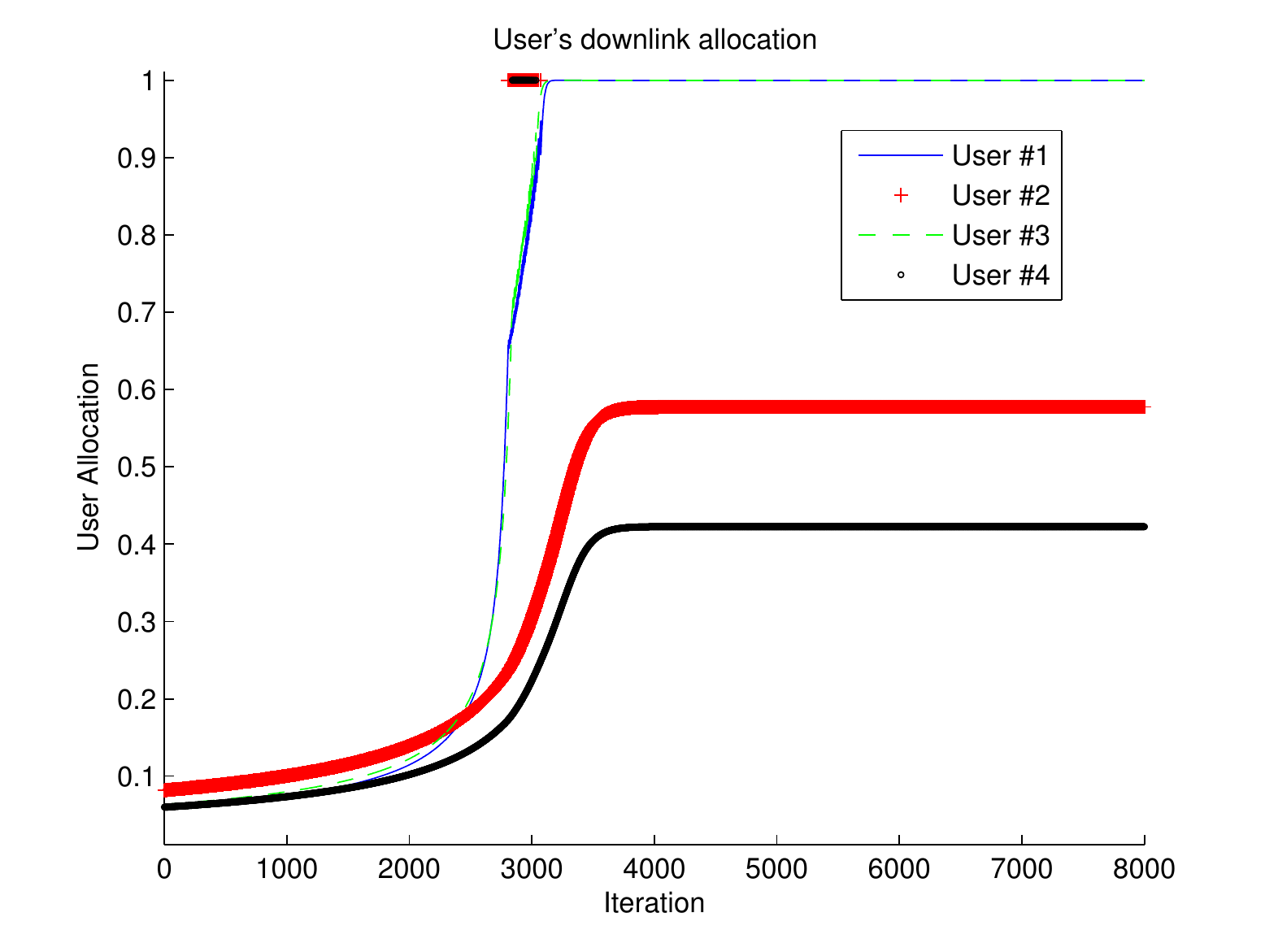}}
\subfigure[Uplink allocations.]{\label{fig:test3_ULAlloc_two}\includegraphics[scale=0.46]{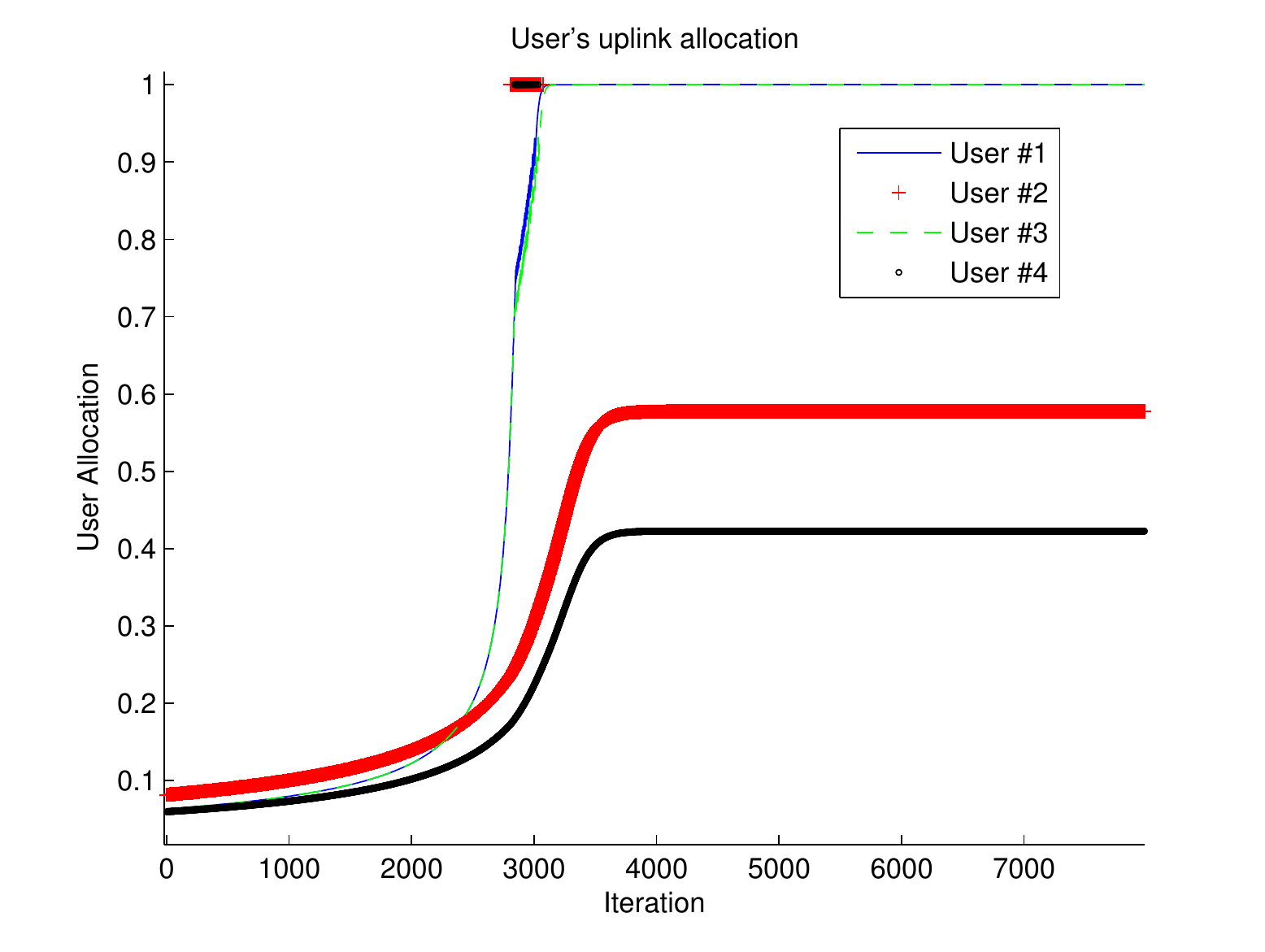}}
\caption{Users' allocations.}
\label{fig:test3_allocations_two}
\end{figure}

Now, as one would expect when the fairness parameter is greater than one, the lower the spectral efficiency of a given user, the more resources are granted to him. As an example, notice in \eqref{eq:test3_allocations_two} that user$\#2$ is receiving more resources than user$\#4$,  with whom he is sharing the base station.

\begin{equation}\label{eq:test3_allocations_two}
\text{Allocation}_{\textsc{dl}} = \begin{pmatrix}
    0 & 0 & 1 \\
    0 & 0.5774 & 0 \\
    1 & 0 & 0 \\
    0 & 0.4226 & 0
  \end{pmatrix};\quad
\text{Allocation}_{\textsc{ul}} = \begin{pmatrix}
    0 & 0 & 1 \\
    0 & 0.5936 & 0 \\
    1 & 0 & 0 \\
    0 & 0.4064 & 0
  \end{pmatrix}.
\end{equation}

\newpage
All the tests we have performed so far suggest that the implementation is working properly, since we are obtaining reasonable results. From here, we are going to test some more complex scenarios.

\subsubsection*{Test 4 - Uplink - Downlink decoupling (\textsc{dud}e)}
Untill now, we reported scenarios where the best option for the users was associating to a single base station in both downlink and uplink but we have not seen any example of decoupled access yet. In order to reveal this underlying feature, we adjust the rate matrices so as to encourage some users to embrace this paradigm. Besides, $\alpha$ is set to $0.5$. As with the previous test scenarios presented in this section, we reproduce the rate matrices below

 \begin{equation}
\text{Rates}_{\textsc{dl}} = \begin{pmatrix}
    8 & 1 & 29 \\
    0.5 & 15 & 1 \\
    25 & 2 & 2 \\
    8 & 28 & 0.9
  \end{pmatrix};\quad
\text{Rates}_{\textsc{ul}} = \begin{pmatrix}
    25 & 1 & 0.5 \\
    0.5 & 15 & 1 \\
    30 & 1 & 0.1 \\
    0.3 & 32 & 0.5
  \end{pmatrix}.
\end{equation}

Figure~\ref{fig:test4_multipliers} shows the change in \textsc{bs}s' multipliers during the simulation. Additionally, the amount of resources granted to each user can be checked in figure~\ref{fig:test4_allocations}.

\begin{figure}[!htb]
\centering     
\subfigure[Downlink multipliers.]{\label{fig:test4_DLMult}\includegraphics[scale=0.46]{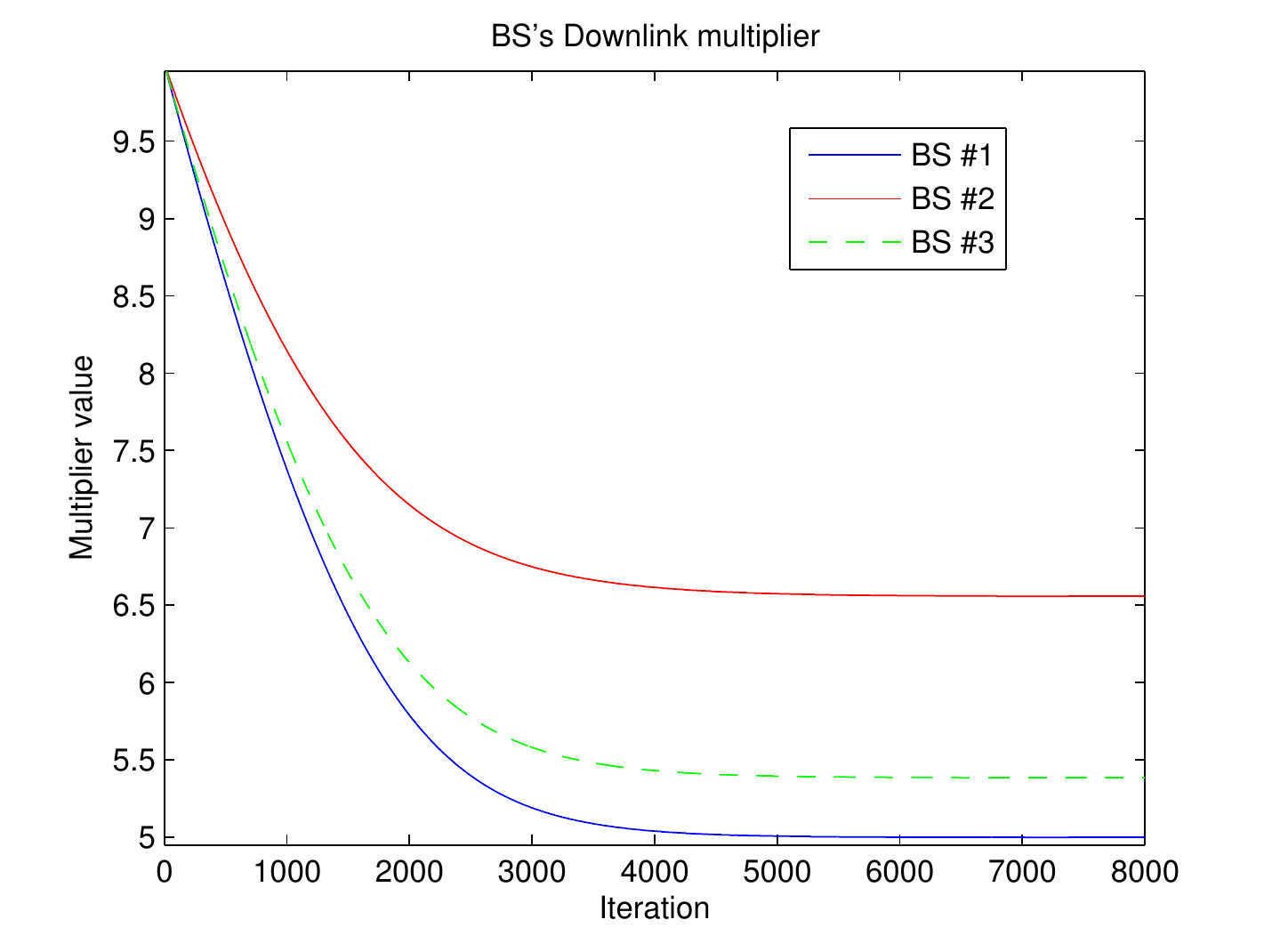}}
\subfigure[Uplink multipliers.]{\label{fig:test4_ULMult}\includegraphics[scale=0.46]{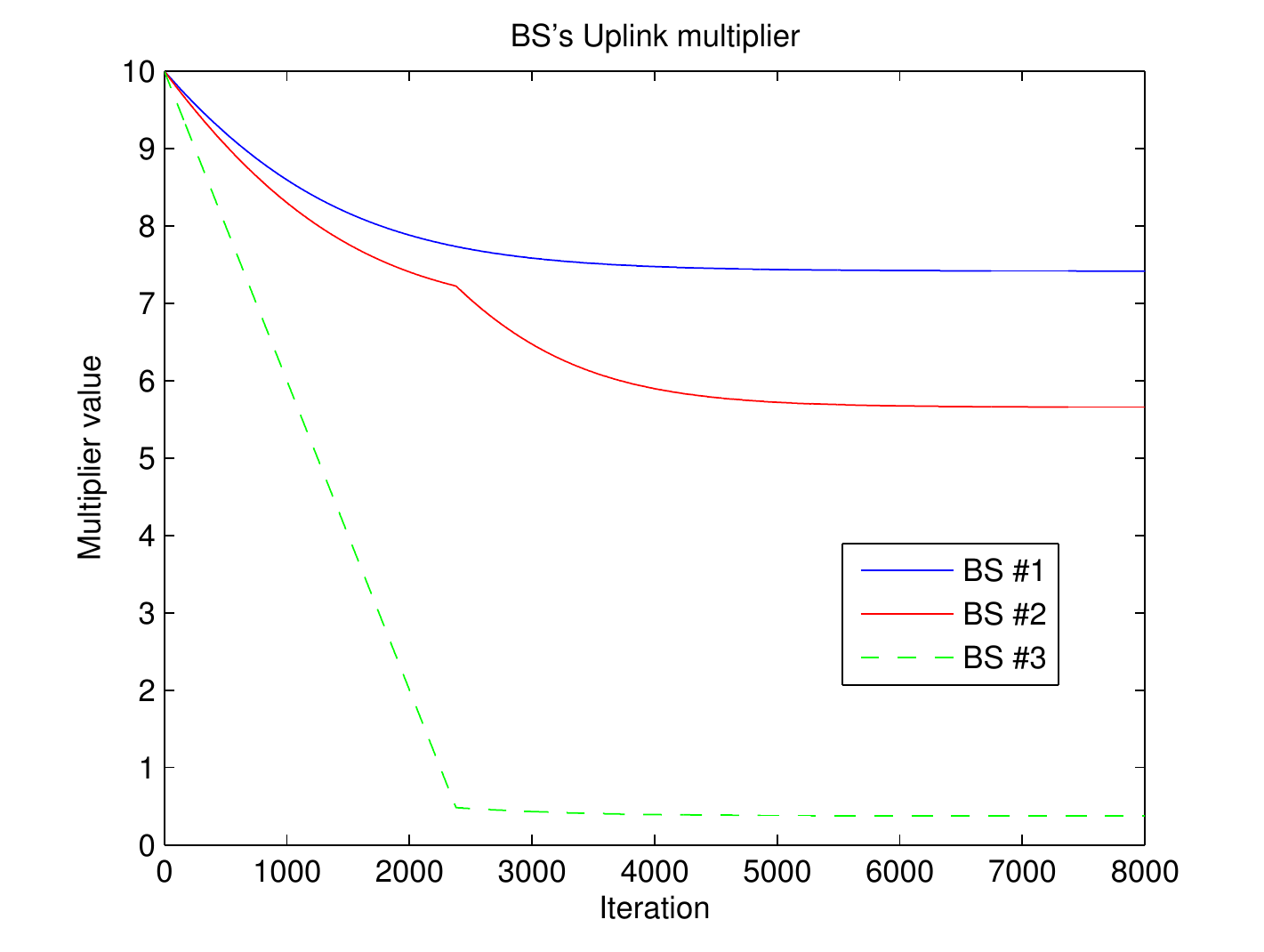}}
\caption{BSs' multipliers.}
\label{fig:test4_multipliers}
\end{figure}

\begin{figure}[!htb]
\centering     
\subfigure[Downlink allocations.]{\label{fig:test4_DLAlloc}\includegraphics[scale=0.46]{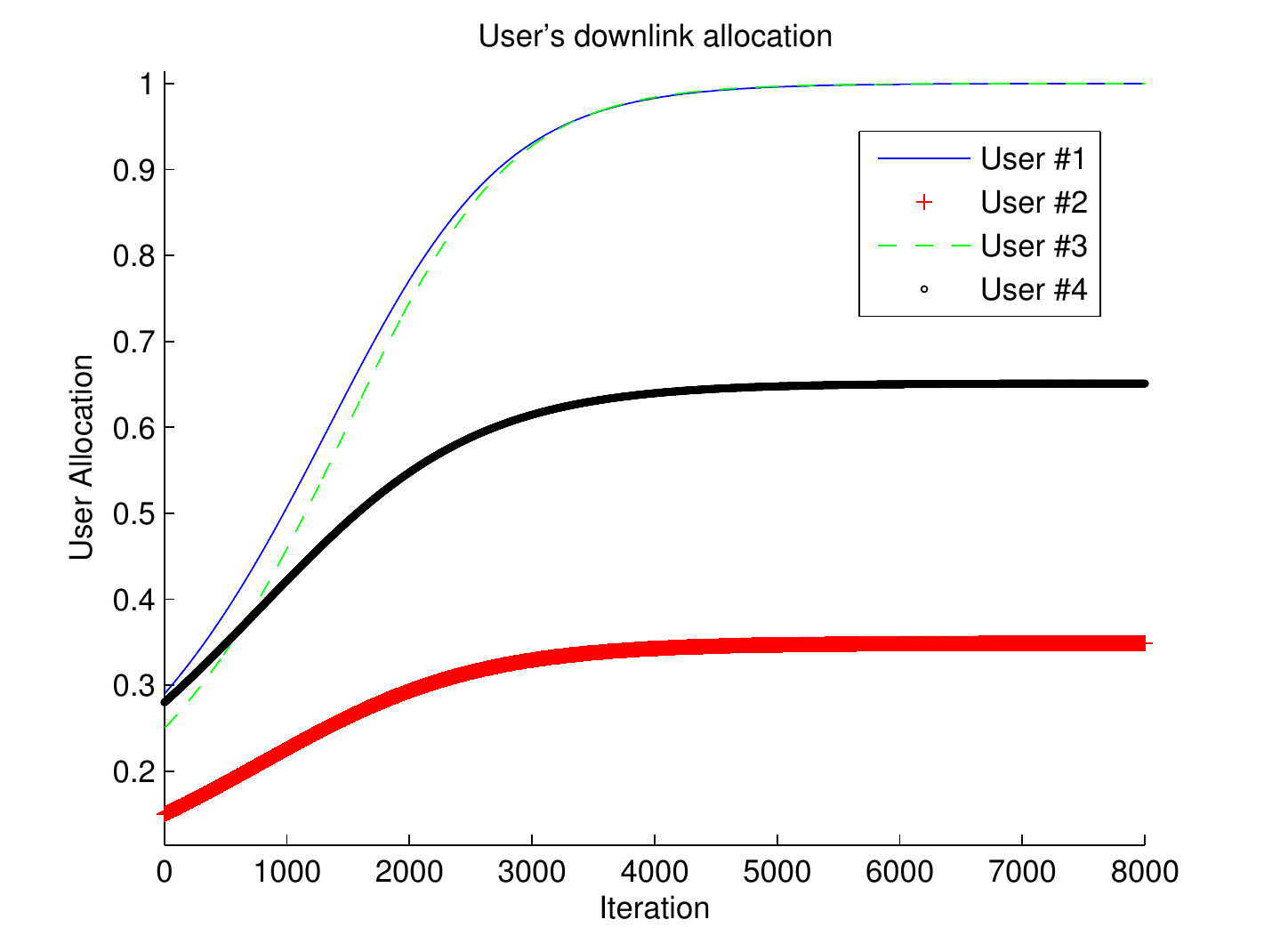}}
\subfigure[Uplink allocations.]{\label{fig:test4_ULAlloc}\includegraphics[scale=0.46]{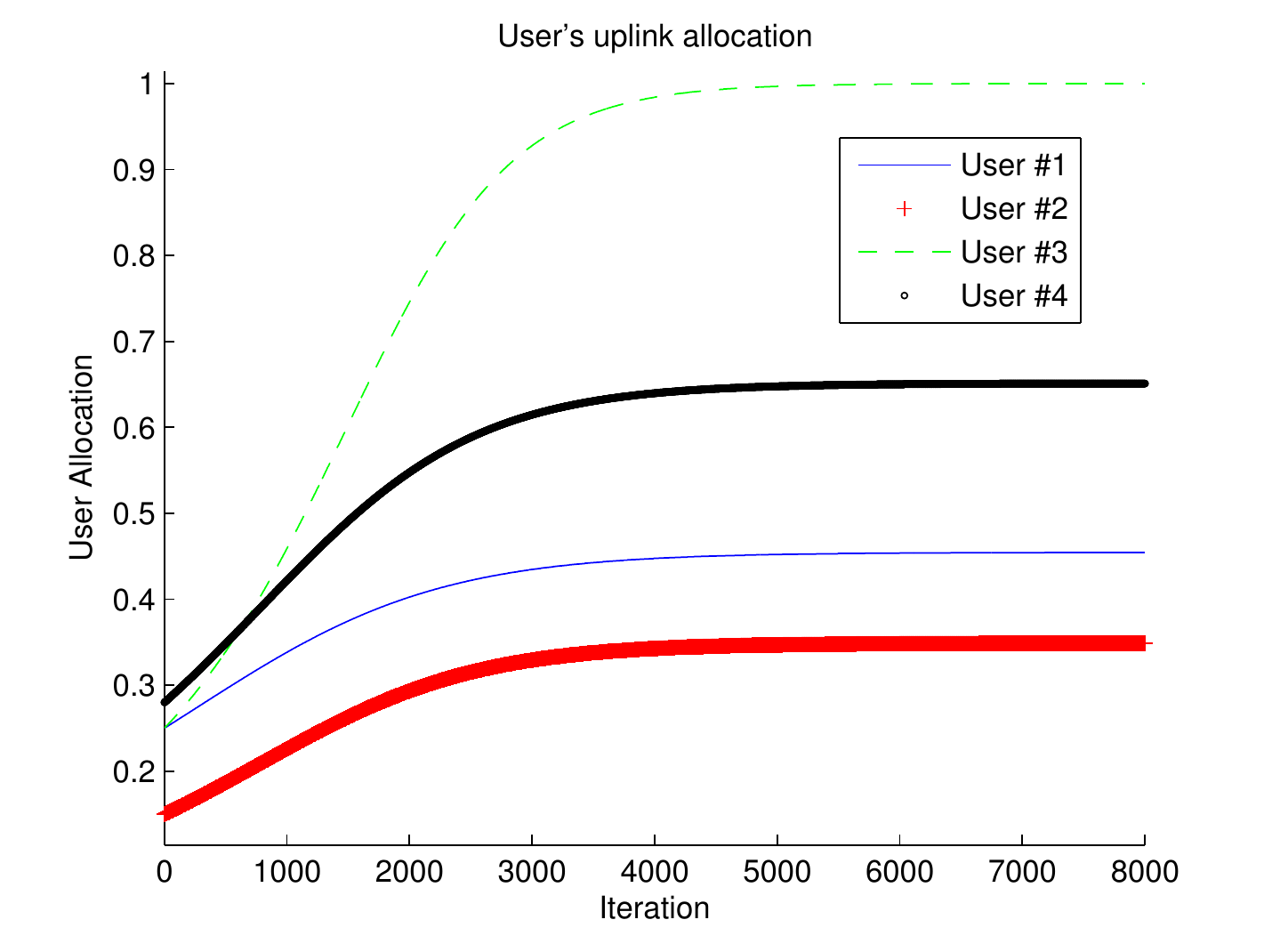}}
\caption{Users' allocations.}
\label{fig:test4_allocations}
\end{figure}

It is worth noting that some of the multipliers have converged to a higher value than others. Closer inspection reveals that base stations with a higher multiplier value are those which are serving more users. This can be verified effortlessly by checking the final allocation matrices below

\begin{equation}\label{eq:test4_allocations}
\text{Allocation}_{\textsc{dl}} = \begin{pmatrix}
    0 & 0 & 1 \\
    0 & 0.3488 & 0 \\
    1 & 0 & 0 \\
    0 & 0.6511 & 0
  \end{pmatrix};\quad
\text{Allocation}_{\textsc{ul}} = \begin{pmatrix}
    0.4544 & 0 & 0 \\
    0 & 0 & 1 \\
    0.5453 & 0 & 0 \\
    0 & 0.9997 & 0
  \end{pmatrix}.
\end{equation}

We may understand the value of the multipliers as the price of a given base station. They are used by the base stations in order to communicate the load status to all the users.

With regard to the uplink-downlink decoupling, the behaviour of user$\#1$ must be stressed. Notice that user$\#1$ associates bases station$\#3$ in downlink. On the contrary, associating to \textsc{bs}$\#1$ is preferred in the uplink. In addition, this change in User$\#1$ 's association, triggers another adjustment for User$\#2$, who still associates \textsc{bs}$\#2$ in downlink but now chooses \textsc{bs}$\#3$ in uplink. Although at first sight, this change might seem a great loss for User$\#2$, the overall performance of the system is better. Let us compare this scenario to the one explained in \emph{Test 3 ; $\alpha = 0.5$}, since the downlink allocation matrix is exactly the same (see \eqref{eq:test3_allocations}). In the latter case, User$\#2$ was granted (in uplink) a fraction of the maximum rate at \textsc{bs}$\#2$ which was equal to

\begin{equation}
\textbf{\text{User\#2}}_\text{ULrate} = 15\, [bits/s/Hz]\, \cdot\, 0.3191 = 4.7865 \, [bits/s/Hz].
\end{equation}

Conversely, now User$\#2$ perceives

\begin{equation}
\textbf{\text{User\#2}}_\text{ULrate}^\prime = 1\, [bit/s/Hz]\, \cdot\, 1 = 1 \, [bit/s/Hz].
\end{equation}

Note that User$\#2$ now strives to obtain a fraction of a maximum rate which is $15$ times lower. Nevertheless, we may consider the gain for other users. In Test 3, User$\#4$ received

\begin{equation}
\textbf{\text{User\#4}}_\text{ULrate} = 32\, [bit/s/Hz]\, \cdot\, 0.6807 = 21.7824 \, [bit/s/Hz].
\end{equation}

On the contrary, now

\begin{equation}
\textbf{\text{User\#4}}_\text{ULrate}^\prime = 32\, [bit/s/Hz]\, \cdot\, 1 = 32 \, [bit/s/Hz].
\end{equation}

As you may notice, this new onfiguration leads to a better use of the network. Therefore, the system performance is maximised as a whole, following the rules we established at the definition of the problem.

\subsubsection*{Test 5 - Load balance}
We observe now a simple example of load balancing when a new base station appears on the system. To do this, we include a fourth base station and fix the fairness parameter to $0.5$, seeking to meet a throughput maximisation. The number of users continues to be the same and so is the user asymmetry constraint. That being said, the new rate matrices are

 \begin{equation}
\text{Rates}_{\textsc{dl}} = \begin{pmatrix}
    8 & 1 & 29 & 1\\
    0.5 & 15 & 1 & 1\\
    25 & 2 & 2 & 1\\
    8 & 28 & 0.9 & 1
  \end{pmatrix};\quad
\text{Rates}_{\textsc{ul}} = \begin{pmatrix}
    8 & 1 & 25 & 1 \\
    0.5 & 15 & 1 & 1\\
    30 & 1 & 5.2 & 1 \\
    0.3 & 32 & 0.5 & 1
  \end{pmatrix}.
\end{equation}

Figures~\ref{fig:test5_allocations} and \ref{fig:test5_multipliers} show the results of the experiment. Additionally, final allocations can be checked in equation~\eqref{eq:test5_allocations}.

\begin{figure}[!htb]
\centering     
\subfigure[Downlink multipliers.]{\label{fig:test5_DLMult}\includegraphics[scale=0.46]{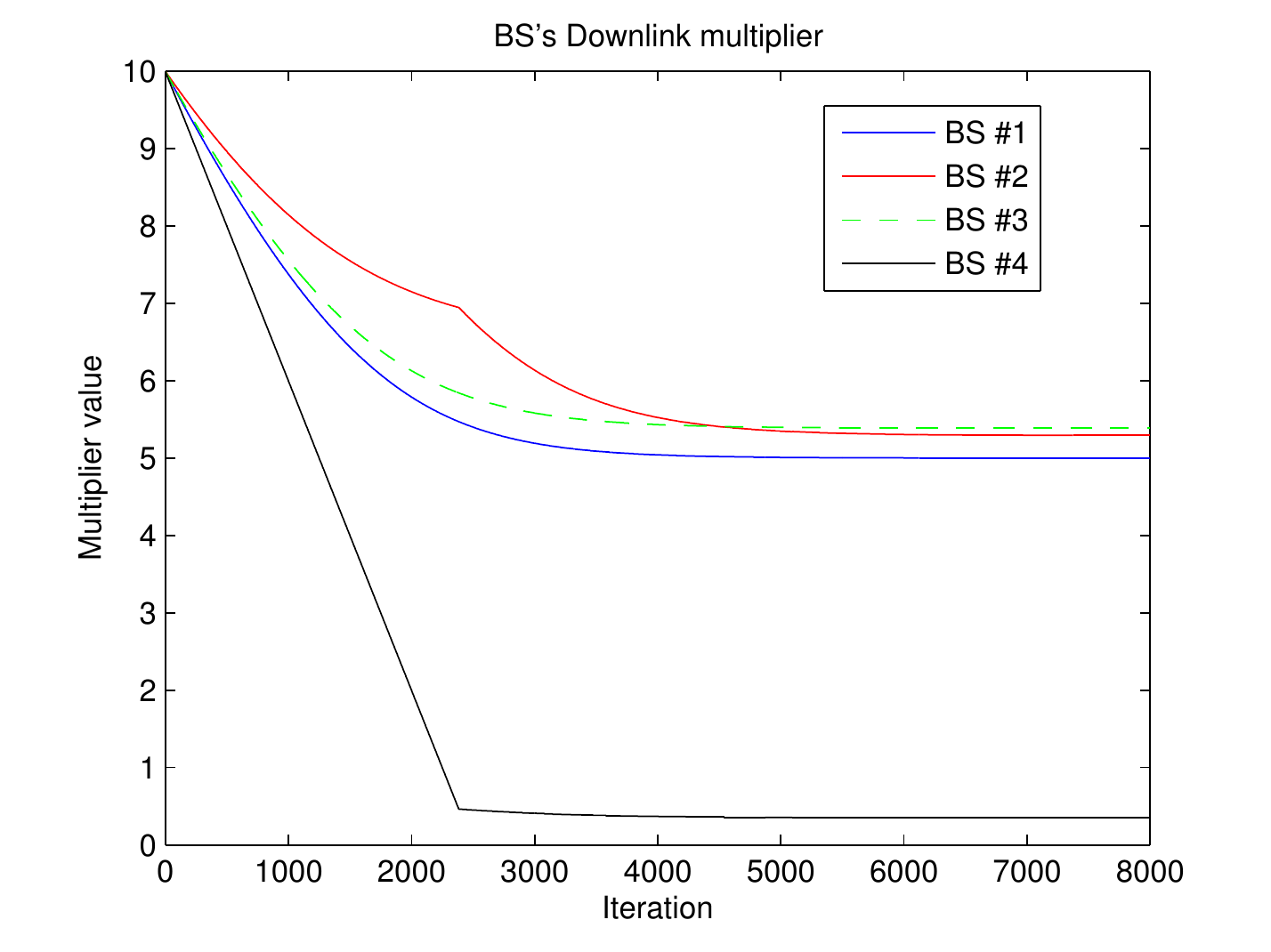}}
\subfigure[Uplink multipliers.]{\label{fig:test5_ULMult}\includegraphics[scale=0.46]{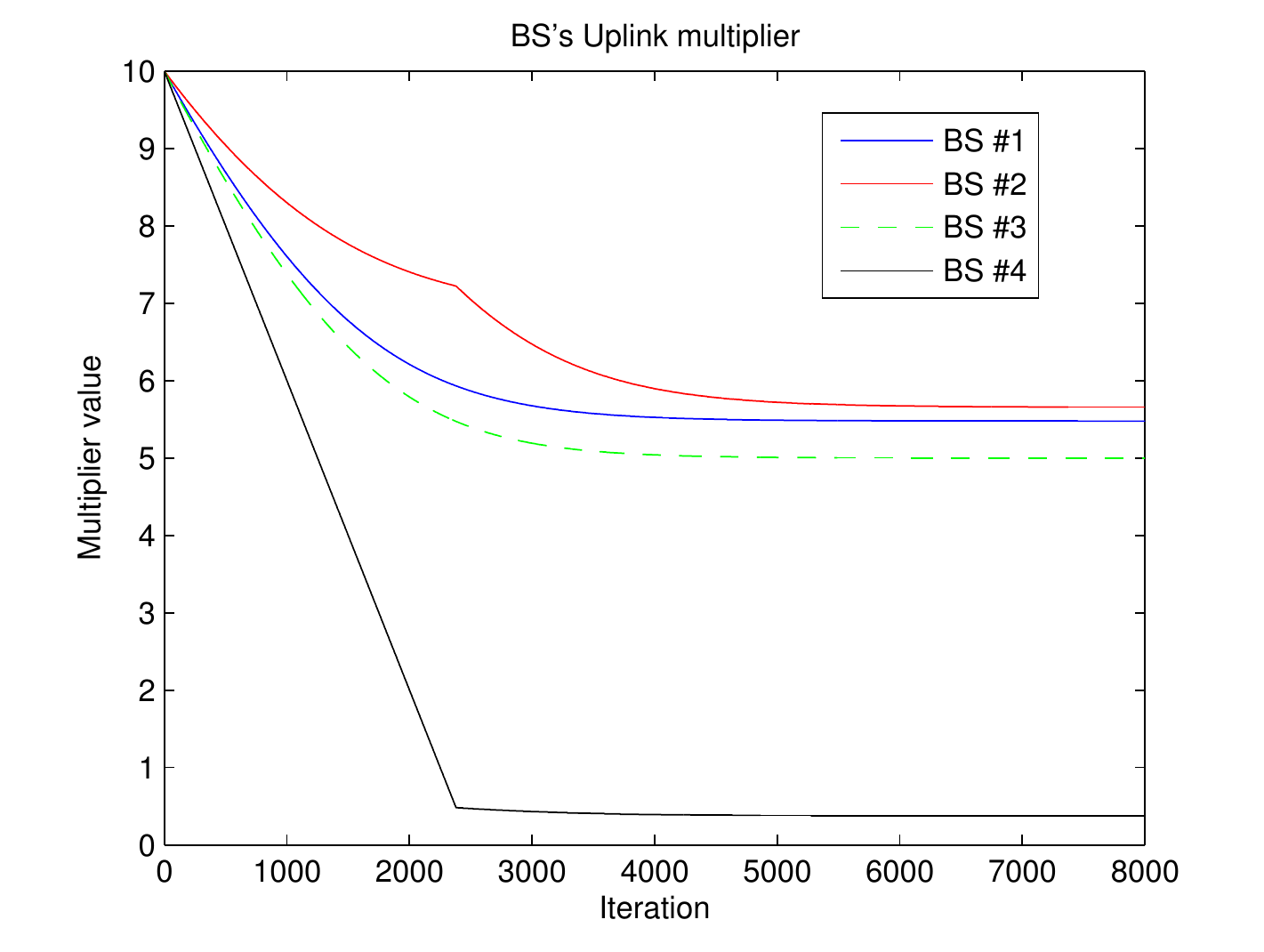}}
\caption{BSs' multipliers.}
\label{fig:test5_multipliers}
\end{figure}

\begin{figure}[!htb]
\centering     
\subfigure[Downlink allocations.]{\label{fig:test5_DLAlloc}\includegraphics[scale=0.46]{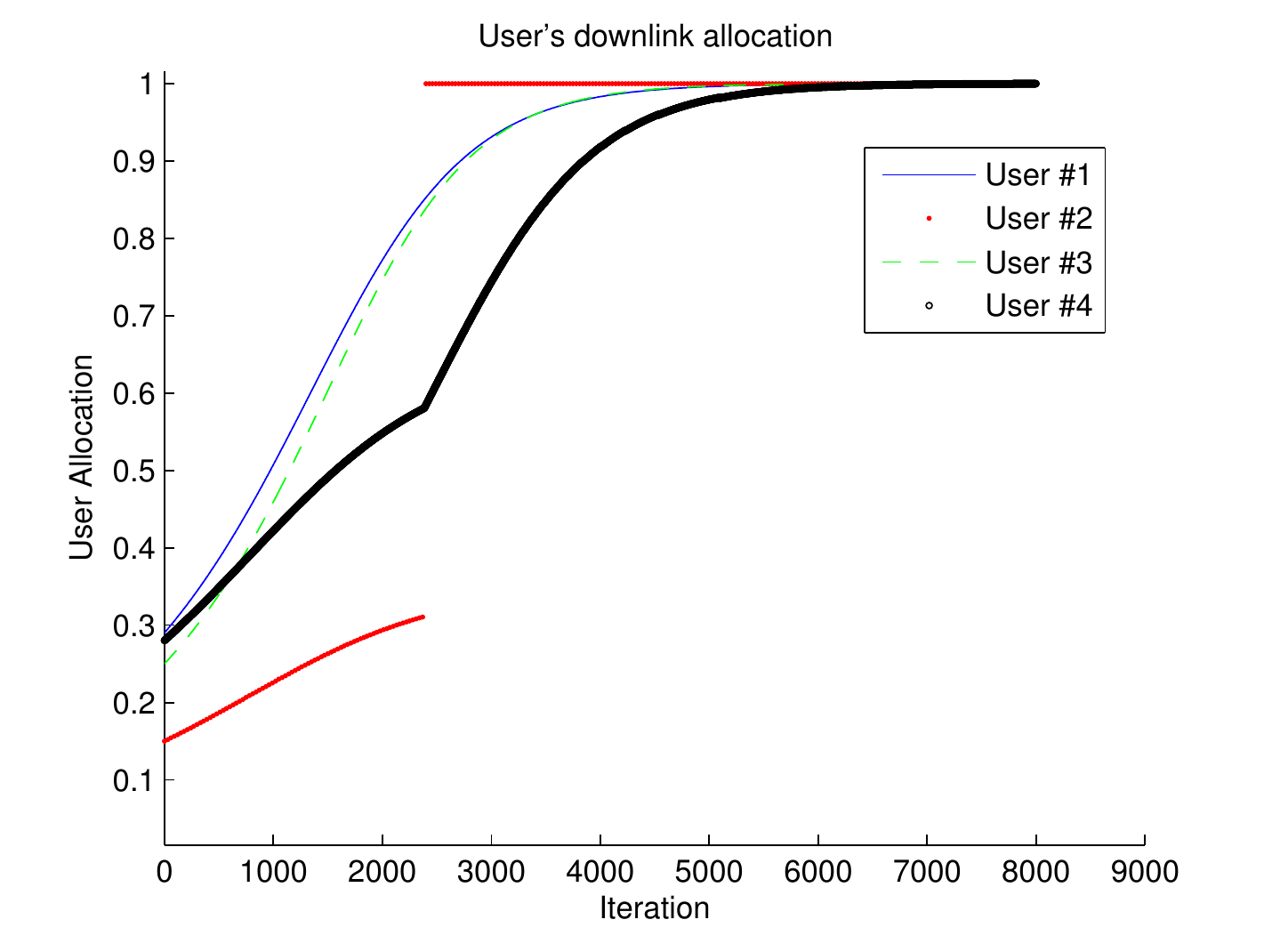}}
\subfigure[Uplink allocations.]{\label{fig:test5_ULAlloc}\includegraphics[scale=0.46]{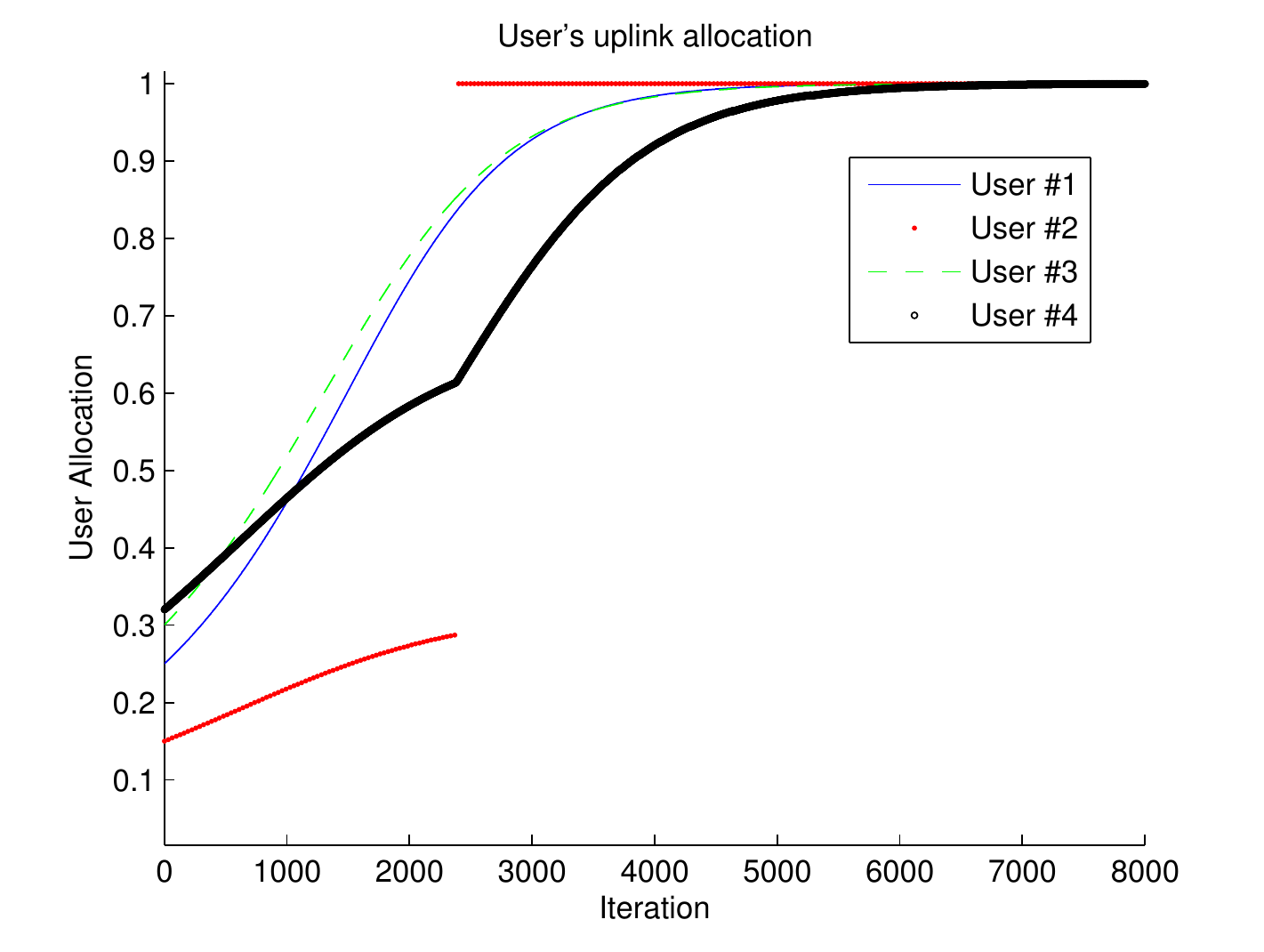}}
\caption{Users' allocations.}
\label{fig:test5_allocations}
\end{figure}

\begin{equation}\label{eq:test5_allocations}
\text{Allocation}_{\textsc{dl}} = \begin{pmatrix}
    0 & 0 & 1 & 0\\
    0 & 0 & 0 & 1 \\
    1 & 0 & 0 & 0\\
    0 & 1 & 0 & 0
  \end{pmatrix};\quad
\text{Allocation}_{\textsc{ul}} = \begin{pmatrix}
    0 & 0 & 1 & 0\\
    0 & 0 & 0 & 1 \\
    1 & 0 & 0 & 0\\
    0 & 1 & 0 & 0
  \end{pmatrix}.
\end{equation}

Note that \textsc{bs}s' multipliers are ranked according to the throughput they are granting to the users. Since \textsc{bs}$\#2$ is providing the highest throughput of the system in uplink ($32 bits/s/Hz$ to User$\#4$), its multiplier has converged to the biggest value of all. Observe, in downlink, how multipliers for base stations $2$ and $3$ converge to almost the same value, due to the fact that both are delivering a similar performance to their respective associated users.

A further inspection of the allocation matrices reveals that the introduction of the new base station has brought several consequences. Contrary to what happened in previous simulations the user with the worst rate in \textsc{bs}$\#2$, which was User$\#2$, has moved to base station number $4$, so that the remaining user (User$\#4$) can obtain the best possible performance of the \textsc{bs}. This means a great gain for latter while a slight performance degradation for the former. It should be taken into account that the entire network has encouraged this change without the need of a centralised entity which explicitly takes the decision.

\subsubsection*{Test 6 - Adding more base stations}
The aim of this test is to assess the effects of adding new base stations on the user's association decisions. Throughout this test we will carry out two different experiments. The first one consist of adding two new base stations. Then, we examine how users are relocated in order to achieve the new global optimal state. Fairness parameter remains fixed to $0.5$ and the per-user asymmetry constraint to $2$. Uplink - downlink rate matrices are

 \begin{equation}
\text{Rates}_{\textsc{dl}} = \begin{pmatrix}
    8 & 1 & 30 & 0 & 1\\
    0.5 & 15 & 1 & 0 & 1\\
    25 & 2 & 2 & 0 & 1\\
    8 & 28 & 0.9 & 0 & 1
  \end{pmatrix};\quad
\text{Rates}_{\textsc{ul}} = \begin{pmatrix}
    8 & 1 & 20 & 0 & 1 \\
    0.5 & 15 & 1 & 0 & 1\\
    27 & 1 & 5.2 & 0 & 1 \\
    0.3 & 32 & 0.5 & 0 & 1
  \end{pmatrix}.
\end{equation}

Furthermore, figures~\ref{fig:test6_multipliers_first} and~\ref{fig:test6_allocations_first} show, respectively, the progress of the multipliers and the evolution of granted allocations during the simulation. 

\begin{figure}[!htb]
\centering     
\subfigure[Downlink multipliers.]{\label{fig:test6_DLMult_first}\includegraphics[scale=0.46]{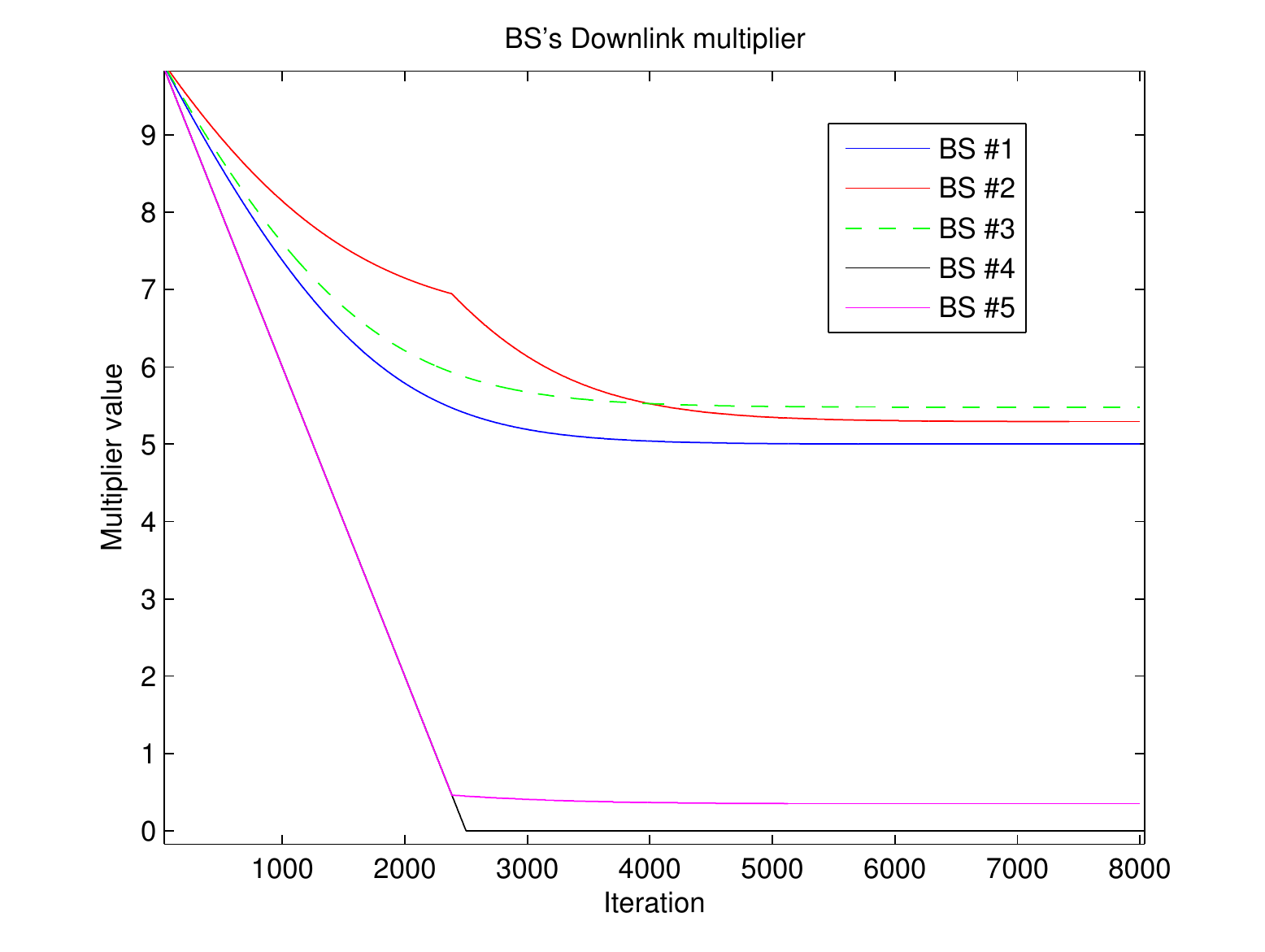}}
\subfigure[Uplink multipliers.]{\label{fig:test6_ULMult_first}\includegraphics[scale=0.46]{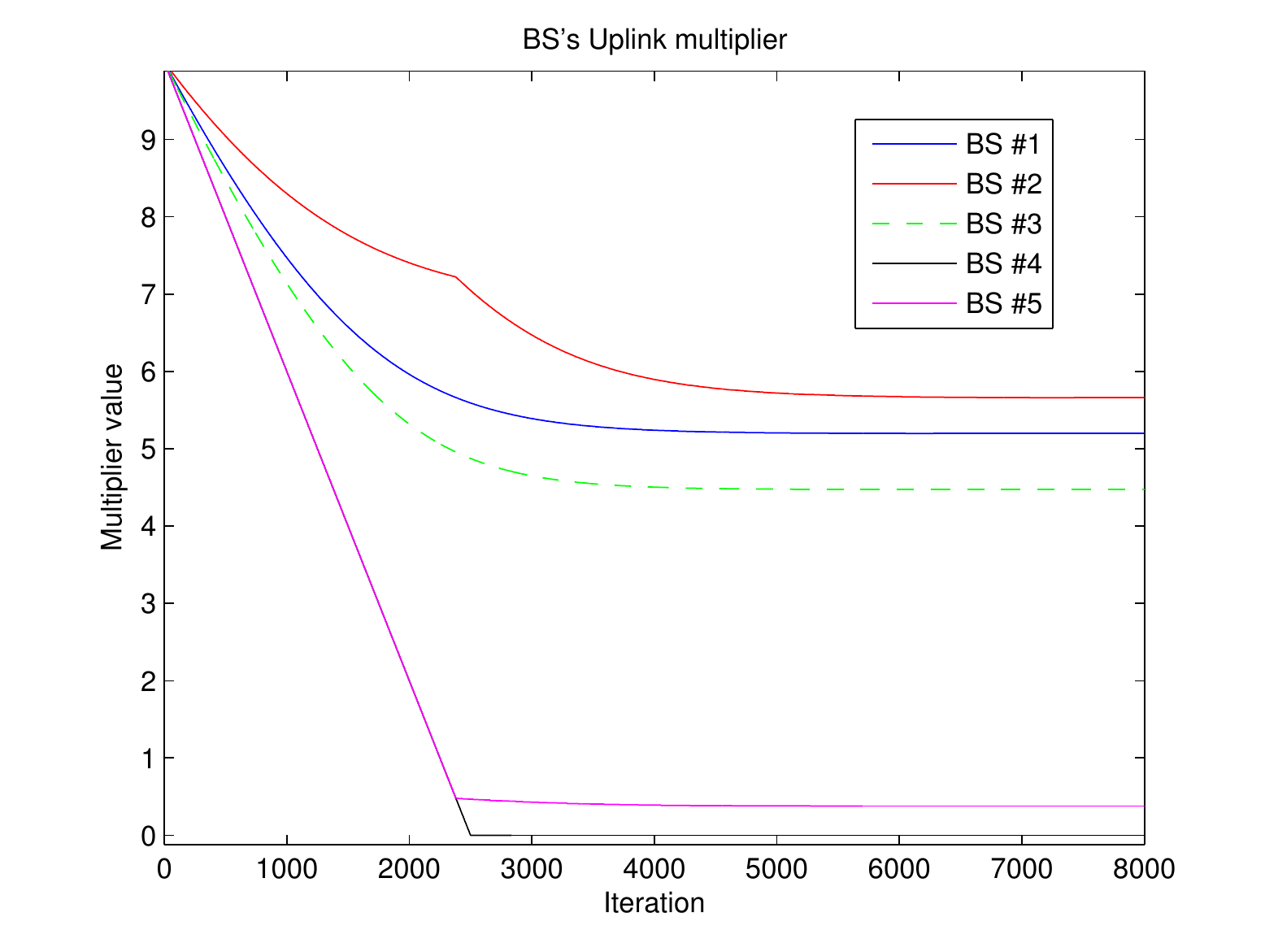}}
\caption{BSs' multipliers.}
\label{fig:test6_multipliers_first}
\end{figure}

\begin{figure}[!htb]
\centering     
\subfigure[Downlink allocations.]{\label{fig:test6_DLAlloc_first}\includegraphics[scale=0.46]{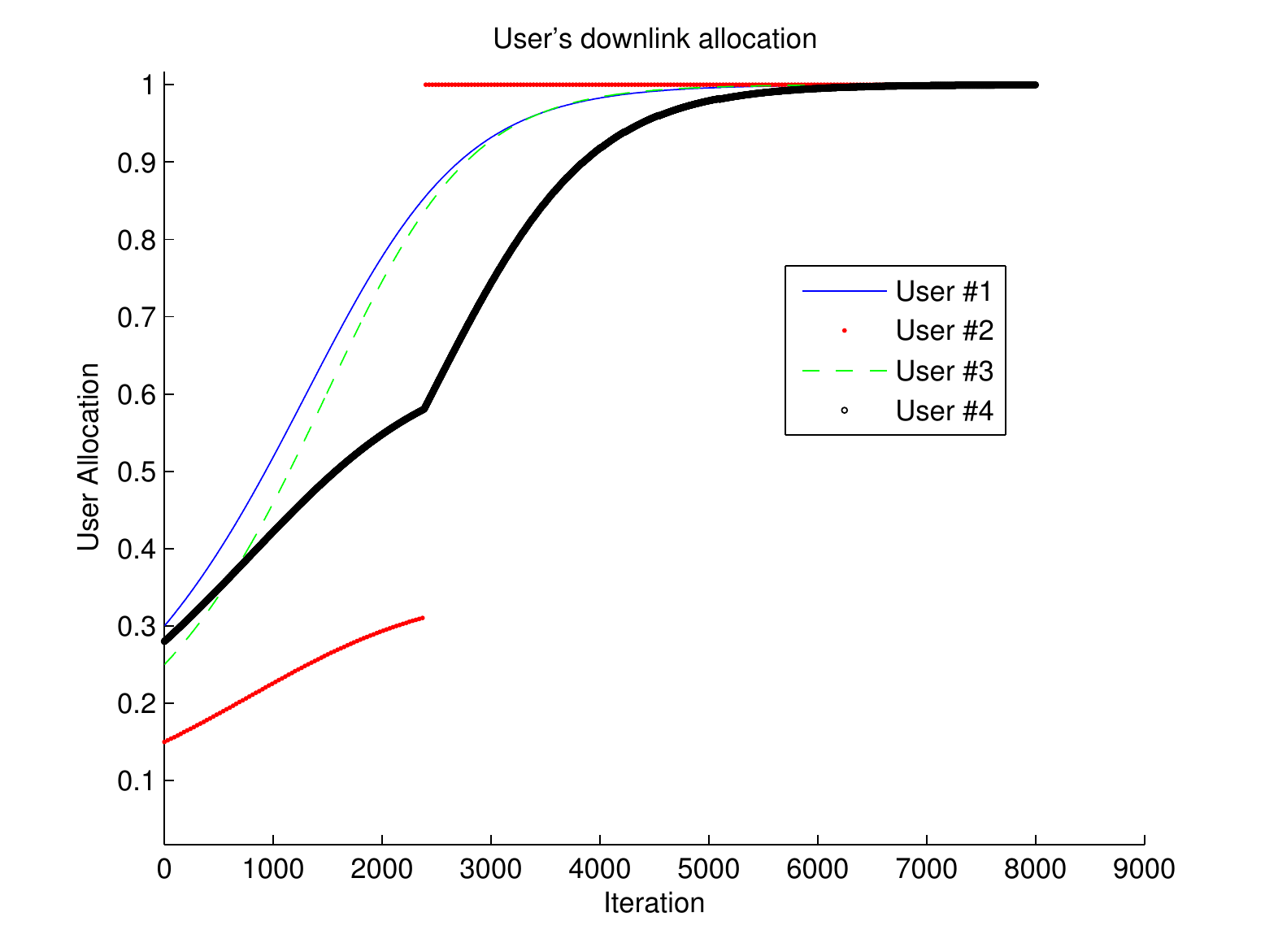}}
\subfigure[Uplink allocations.]{\label{fig:test6_ULAlloc_first}\includegraphics[scale=0.46]{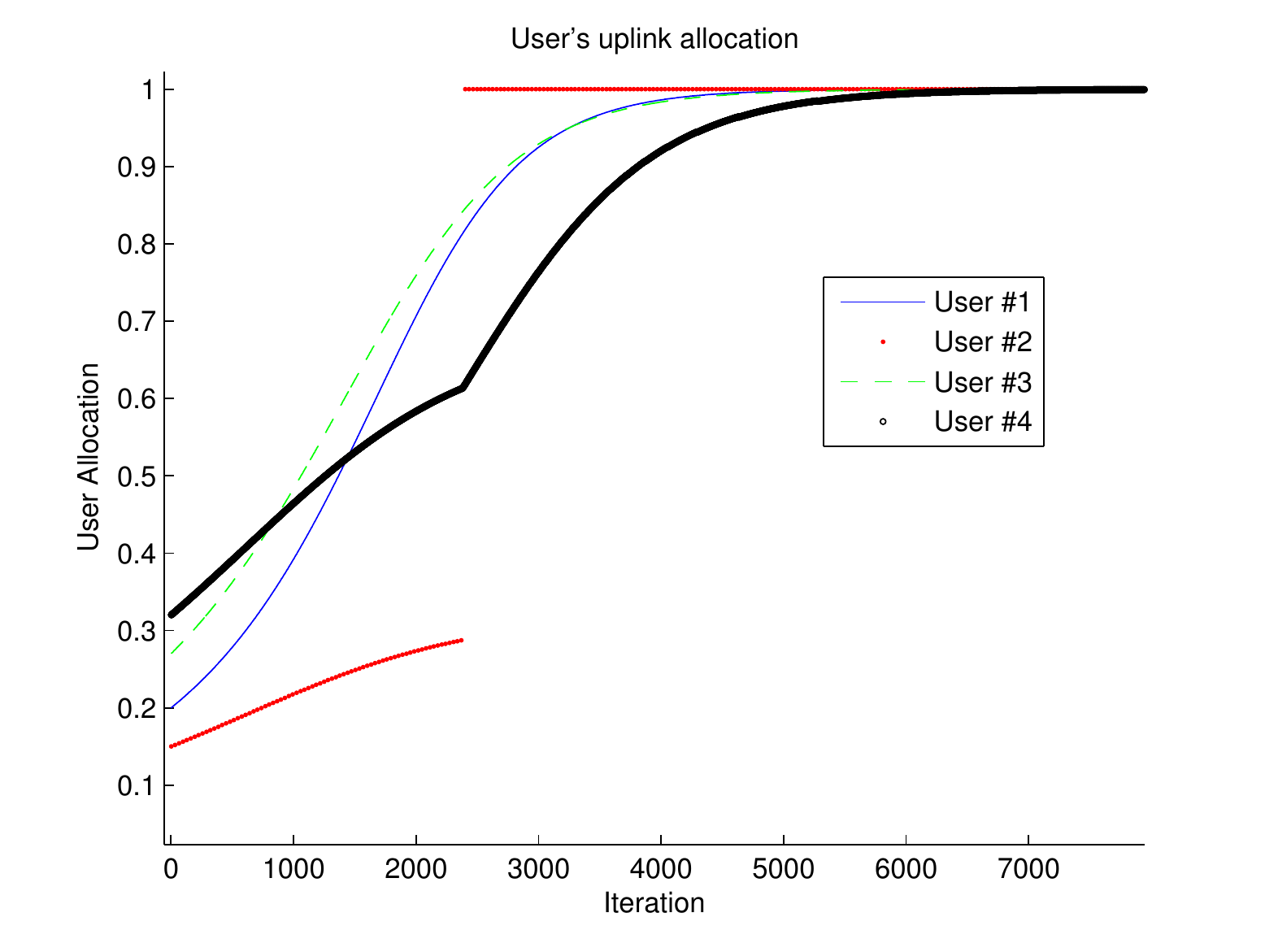}}
\caption{Users' allocations.}
\label{fig:test6_allocations_first}
\end{figure}

We find here a scenario that resembles the last one, that is, since there are empty base stations User$\#2$ is invited to leave base station number $2$. In this case, observe that User$\#2$ has got two options that would represent the same gain for him (\textsc{bs}$\#3$ and \textsc{bs}$\#5$) in both the uplink and the downlink. A closer look at the final allocation matrices reveals that the user has associated to \textsc{bs}$\#5$ (see. \eqref{eq:test6_allocations_first}). Otherwise, User$\#2$ joining \textsc{bs}$\#3$ would have hugely decreased utility and uplink-downlink symmetry since User$\#3$ is already using that base station.

\begin{equation}\label{eq:test6_allocations_first}
\text{Allocation}_{\textsc{dl}} = \begin{pmatrix}
    0 & 0 & 1 & 0 & 0\\
    0 & 0 & 0 & 0 & 1 \\
    1 & 0 & 0 & 0 & 0\\
    0 & 1 & 0 & 0 & 0
  \end{pmatrix};\quad
\text{Allocation}_{\textsc{ul}} = \begin{pmatrix}
    0 & 0 & 1 & 0 & 0\\
    0 & 0 & 0 & 0 & 1 \\
    1 & 0 & 0 & 0 & 0\\
    0 & 1 & 0 & 0 & 0
  \end{pmatrix}.
\end{equation}

Finally, for the second expermient, we will again add two new base stations but in this case we observe that both of them are going to be used. The fairness parameter remains fixed to $0.5$ and the per-user asymmetry constraint to $2$. The rate matrices are these:

 \begin{equation}
\text{Rates}_{\textsc{dl}} = \begin{pmatrix}
    8 & 1 & 30 & 0 & 1\\
    0.5 & 15 & 1 & 8 & 2\\
    25 & 2 & 2 & 0 & 1\\
    8 & 28 & 0.9 & 0 & 1
  \end{pmatrix};\quad
\text{Rates}_{\textsc{ul}} = \begin{pmatrix}
    8 & 1 & 20 & 0 & 1 \\
    0.5 & 15 & 1 & 1 & 3\\
    27 & 1 & 5.2 & 0 & 1 \\
    0.3 & 32 & 0.5 & 0 & 1
  \end{pmatrix}.
\end{equation}

Again, we include the plots showing the convergence of the multipliers (Figure~\ref{fig:test6_multipliers_second}) and the final resource allocation (Figure~\ref{fig:test6_allocations_second}).

\begin{figure}[!htb]
\centering     
\subfigure[Downlink multipliers.]{\label{fig:test6_DLMult_second}\includegraphics[scale=0.46]{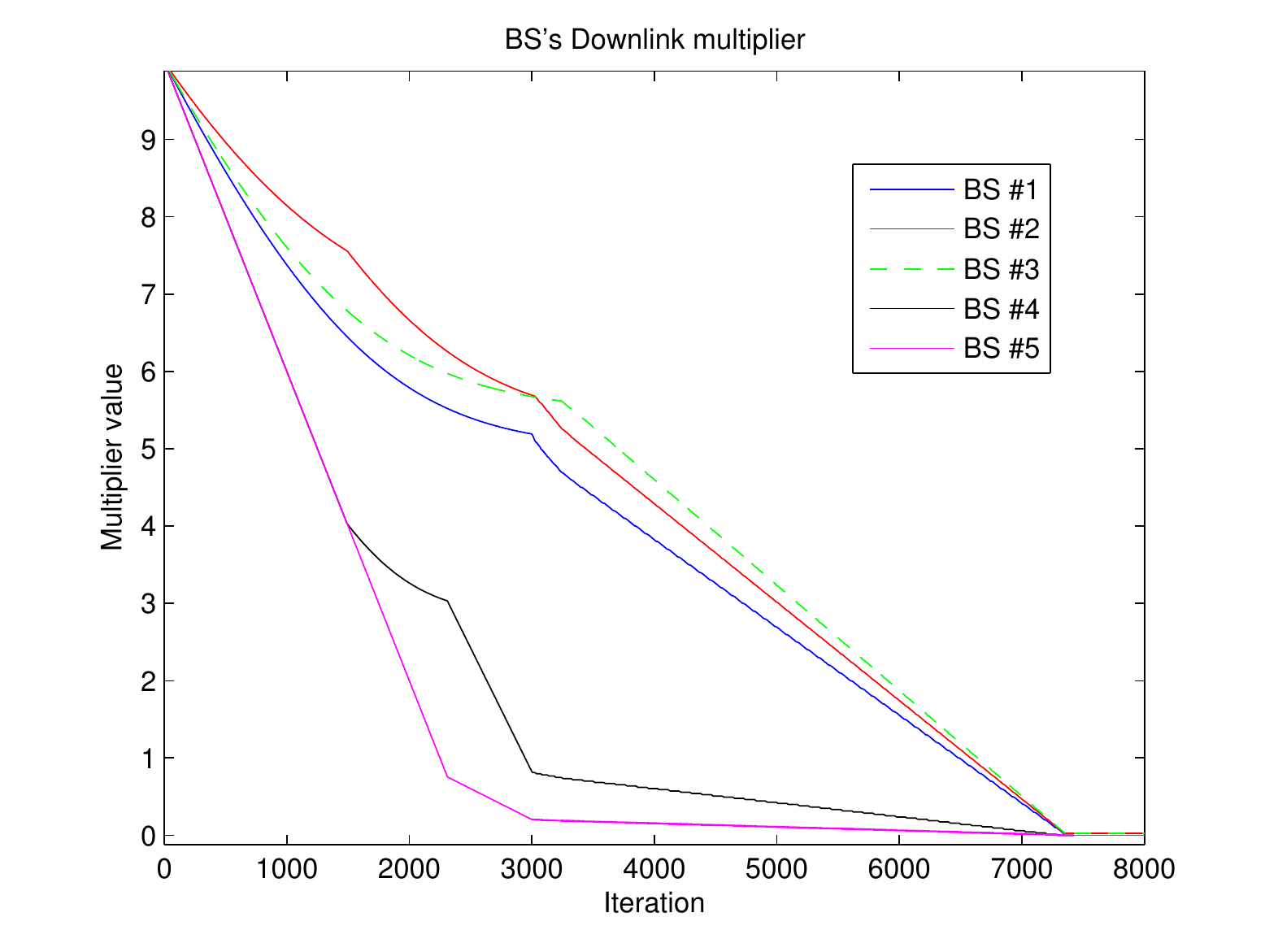}}
\subfigure[Uplink multipliers.]{\label{fig:test6_ULMult_second}\includegraphics[scale=0.46]{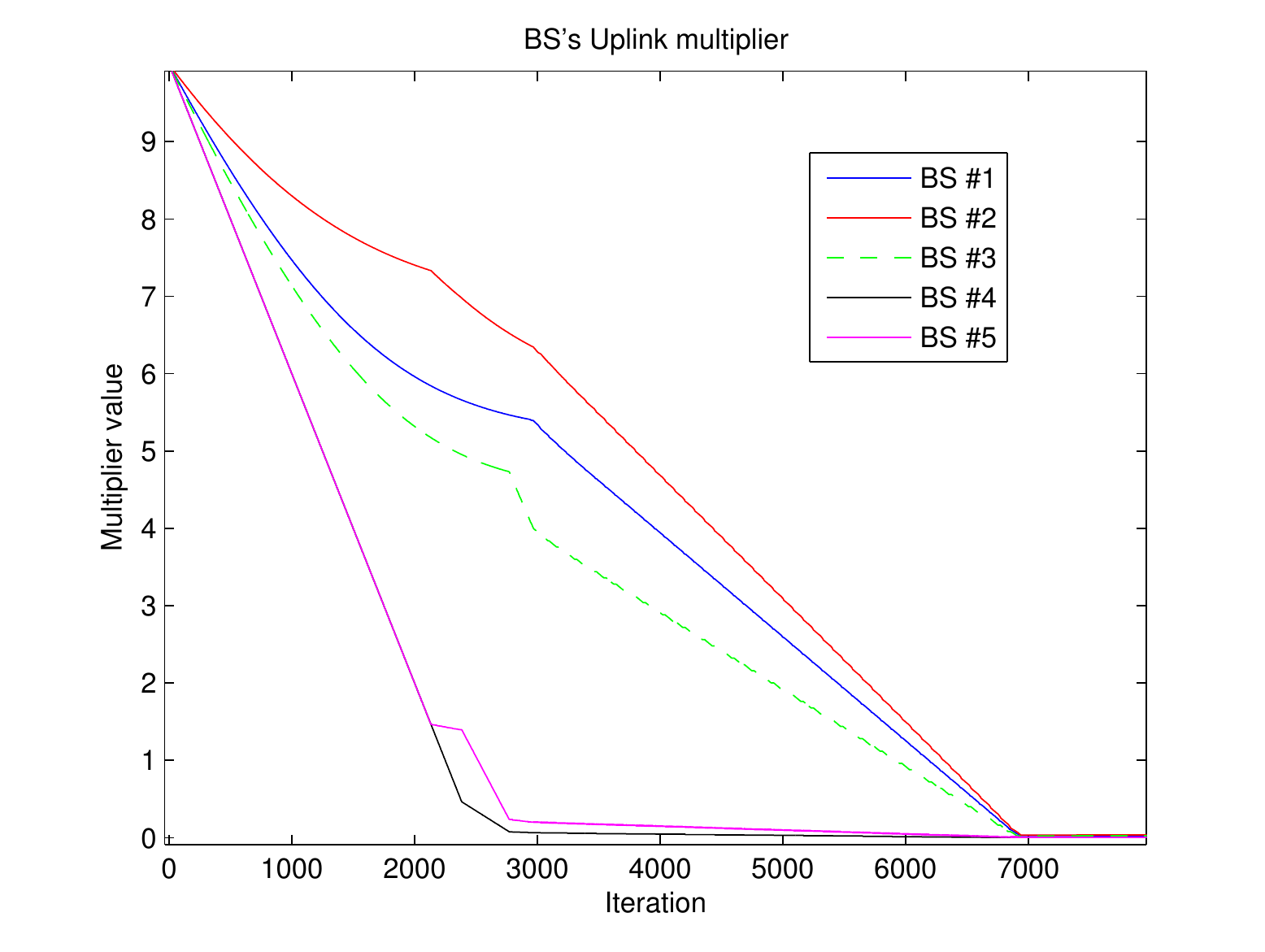}}
\caption{BSs' multipliers.}
\label{fig:test6_multipliers_second}
\end{figure}

\begin{figure}[!htb]
\centering     
\subfigure[Downlink allocations.]{\label{fig:test6_DLAlloc_second}\includegraphics[scale=0.46]{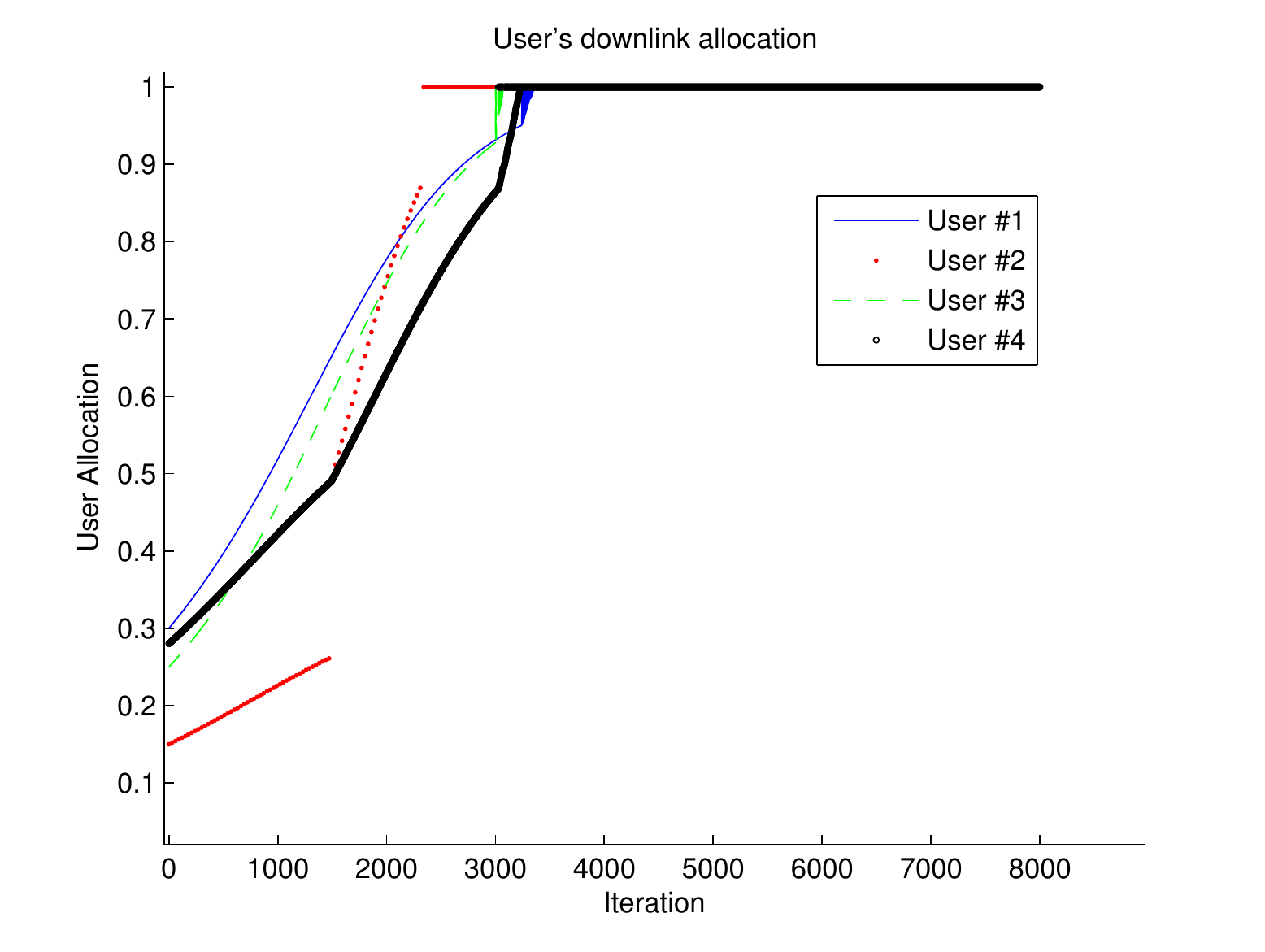}}
\subfigure[Uplink allocations.]{\label{fig:test6_ULAlloc_second}\includegraphics[scale=0.46]{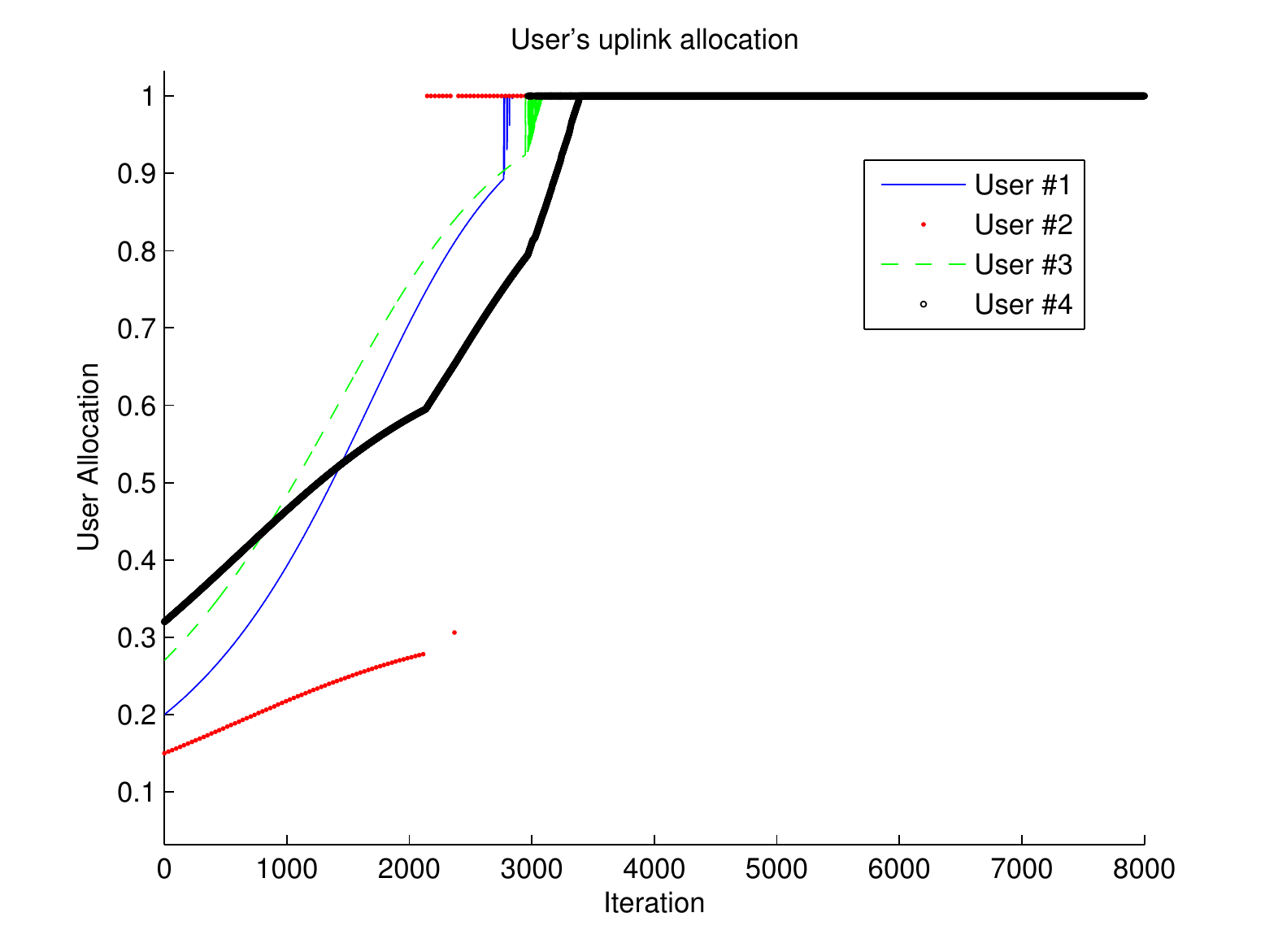}}
\caption{Users' allocations.}
\label{fig:test6_allocations_second}
\end{figure}

User$\#2$ leaves on more time the second base station but now, a decoupled association scheme for uplink and downlink arises. Note in~\eqref{eq:test6_allocations_second} that User$\#2$ decides to associate \textsc{bs}$\#4$ in downlink, while \textsc{bs}$\#5$  is preferred in uplink because that represents a better overall performance. Observe that any of the base stations remains completely idle.

\begin{equation}\label{eq:test6_allocations_second}
\text{Allocation}_{\textsc{dl}} = \begin{pmatrix}
    0 & 0 & 1 & 0 & 0\\
    0 & 0 & 0 & 1 & 0 \\
    1 & 0 & 0 & 0 & 0\\
    0 & 1 & 0 & 0 & 0
  \end{pmatrix};\quad
\text{Allocation}_{\textsc{ul}} = \begin{pmatrix}
    0 & 0 & 1 & 0 & 0\\
    0 & 0 & 0 & 0 & 1 \\
    1 & 0 & 0 & 0 & 0\\
    0 & 1 & 0 & 0 & 0
  \end{pmatrix}.
\end{equation}

\subsection{About the convergence of the algorithm}\label{sec:speed}
Since the algorithm under test is based on a \emph{gradient projection method}, it is worth mentioning some aspects about the speed of convergence. As we have already mentioned previously, the algorithm seems to be fast enough to execute in conventional mobile devices. Nevertheless, its convergence time obviously depends upon the multipliers' initial value as well as upon the step size of the gradient method, $\gamma$. The former aspect would require further research that we will consider as future work. Regarding the latter, we may choose the size of the step in an adaptive way. The main idea is to use a big step size at the beginning of the algorithm so as to improve the convergence speed to a value near the optimal point. However, have in mind that a large step size might cause the gradient method to become unstable. That is the reason why a smaller step size should be chosen as the iterative process advances in order to ensure a small steady state error. Plenty of examples dealing with this topic can found in the literature of the subject, especially, that related to  the \emph{Least-Mean-Square algorithms}~\cite{lms} which are widely used in adaptative filtering.

\subsection{Algorithm mode of operation}
Untill now, we assumed that the algorithm is being executed continuously. Another option would be to recalculate the allocations and multiplier values after a certain amount of time or when certain conditions are met. If the network has achieved the optimal association and allocation values, for instance, because no further changes happen after a given number of iterations we shall stop the execution of the algorithm. If a new user joins or leaves the system, we need to recompute again both optimal association and allocation values for each affected user, that is, for each user who was associated to the same base station that the new user is joining or leaving in a gicen link (uplink or downlink). Every affected user then, can be seen as a new user entering the system (since their optimal allocation and allocation decision may change) and they may trigger the resource allocation process in their respectice serving base stations in both uplink and downlink.

\section{Integration with the network simulator}
After the validation of the distributed algorithms, they have been integrated into the simulation tool presented in Chapter~\ref{chap:dude}. In this section, we conduct a comparison between the three main options studied throughout this work, that is, the initial approach based on a fixed association and allocation scheme, the second one, which implied allowing a dynamic allocation of resources while mantaining a fixed assocatiation strategy. Finally, we bring into comparison the third approach which relies on a less restrictive non-fixed policy for both association and allocation processes.

Table~\ref{parameters_final} shows the test-bed parameters that have been used in order to conduct the assessment. Note that we are carrying out this comparative evaluation under the throughput maximisation policy ($\alpha - fairness < 1$) since it might be the case of greater interest for a real network deployment and network operators.

\begin{savenotes}
  \begin{table}[!htb]
    \caption{Deployment parameters}
    \centering
    \label{parameters_final}
    \begin{tabular}{cl} \\ \hline
    
    \multicolumn{2}{|c|}{\cellcolor[HTML]{329A9D}{\color[HTML]{000000} Shared network parameters}}   \\ \hline
      \multicolumn{1}{|c|}{\cellcolor[HTML]{EFEFEF}{\color[HTML]{343434} Area of interest}} & \multicolumn{1}{l|}{$1000$ m $\times$ $1000$ m} \\ \hline
      \multicolumn{1}{|c|}{\cellcolor[HTML]{EFEFEF}{\color[HTML]{343434} }}
                                                                                            &
                                                                                              \multicolumn{1}{l|}{$\lambda_\text{Macrocells}
                                                                                              =
                                                                                              3$}
      \\ 
\multicolumn{1}{|c|}{\cellcolor[HTML]{EFEFEF}{\color[HTML]{343434} }}                                                                                                 & \multicolumn{1}{l|}{$\lambda_\text{Femtocells}$ = $\lambda_\text{Macrocells}\cdot$  ratio\footnote{$ratio = \frac{\lambda_F}{\lambda_M}$. Ratio of the number of femtocells to the number of macrocells. In our case, this ratio is equal to 3.}} \\
\multicolumn{1}{|c|}{\multirow{-3}{*}{\cellcolor[HTML]{EFEFEF}{\color[HTML]{343434} \begin{tabular}[c]{@{}c@{}}Network deployment\\ (PPP intensities)\end{tabular}}}} & \multicolumn{1}{l|}{$\#_ {Users} = 50$}                       \\ \hline
\multicolumn{1}{|c|}{\cellcolor[HTML]{EFEFEF}{\color[HTML]{343434} }}                                                                                                 & \multicolumn{1}{l|}{MBS = 46 dBm}                                   \\
\multicolumn{1}{|c|}{\cellcolor[HTML]{EFEFEF}{\color[HTML]{343434} }}                                                                                                 & \multicolumn{1}{l|}{FBS = 20 dBm}                                   \\
\multicolumn{1}{|c|}{\multirow{-3}{*}{\cellcolor[HTML]{EFEFEF}{\color[HTML]{343434} Transmit power}}}  & \multicolumn{1}{l|}{Device = 20 dBm} \\ \hline
\multicolumn{1}{|c|}{\cellcolor[HTML]{EFEFEF}{\color[HTML]{343434} Path-loss exponent}}                                                                               & \multicolumn{1}{l|}{$4$}                                      \\ \hline
\multicolumn{1}{|c|}{\cellcolor[HTML]{EFEFEF}{\color[HTML]{343434} Noise level}}                                                                                      & \multicolumn{1}{l|}{$- 106$ dBm}                                      \\ \hline
\multicolumn{2}{|c|}{\cellcolor[HTML]{329A9D}{\color[HTML]{000000} Chapter 2 Algorithm's parameters}}  \\ \hline
\multicolumn{1}{|c|}{\cellcolor[HTML]{EFEFEF}{\color[HTML]{343434} Fairness parameter}}                                                                                      & \multicolumn{1}{l|}{$\alpha-fairness = 0.5$} \\ \hline
\multicolumn{1}{|c|}{\cellcolor[HTML]{EFEFEF}{\color[HTML]{343434} Asymmetry  weight}}                                                                                      & \multicolumn{1}{l|}{$A = 2$}                                      \\ \hline
\multicolumn{2}{|c|}{\cellcolor[HTML]{329A9D}{\color[HTML]{000000} Chapter 3 Algorithm's parameters}}  \\ \hline
\multicolumn{1}{|c|}{\cellcolor[HTML]{EFEFEF}{\color[HTML]{343434} Fairness parameter}}                                                                                      & \multicolumn{1}{l|}{$\alpha-fairness = 0.5$} \\ \hline
\multicolumn{1}{|c|}{\cellcolor[HTML]{EFEFEF}{\color[HTML]{343434} Per-user asymmetry}}                                                                                      & \multicolumn{1}{l|}{$\epsilon_u = 2$}     \\ \hline
\multicolumn{1}{|c|}{\cellcolor[HTML]{EFEFEF}{\color[HTML]{343434} Gradient step}}                                                                                      & \multicolumn{1}{l|}{$\gamma = 0.004$} \\ \hline
\multicolumn{1}{|c|}{\cellcolor[HTML]{EFEFEF}{\color[HTML]{343434} Number of iterations}}                                                                                      & \multicolumn{1}{l|}{$8000$} \\ \hline

\end{tabular}
\end{table}
\end{savenotes}

Let us start by assessing the gains of each approach by presenting several measurements. Some of these are: aggregate throughput for both links, base stations' workload, per-user aymmetry, as well as, the different coverage maps resulting from each solution under the same network and user deployment. Figure~\ref{fig:Final_assessment_original_association} illustrates the coverage maps consequence of applying the original pathloss and received power criteria for uplink and downlink, respectively. The points represent the base stations while crosses indicate the location of the users. Additionally, light blue lines divide up the space according to the distance policy in the uplink map.

\begin{figure}[!htb]
\centering     
\subfigure[Downlink coverage.]{\label{fig:Final_DL_association_original}\includegraphics[scale=0.46]{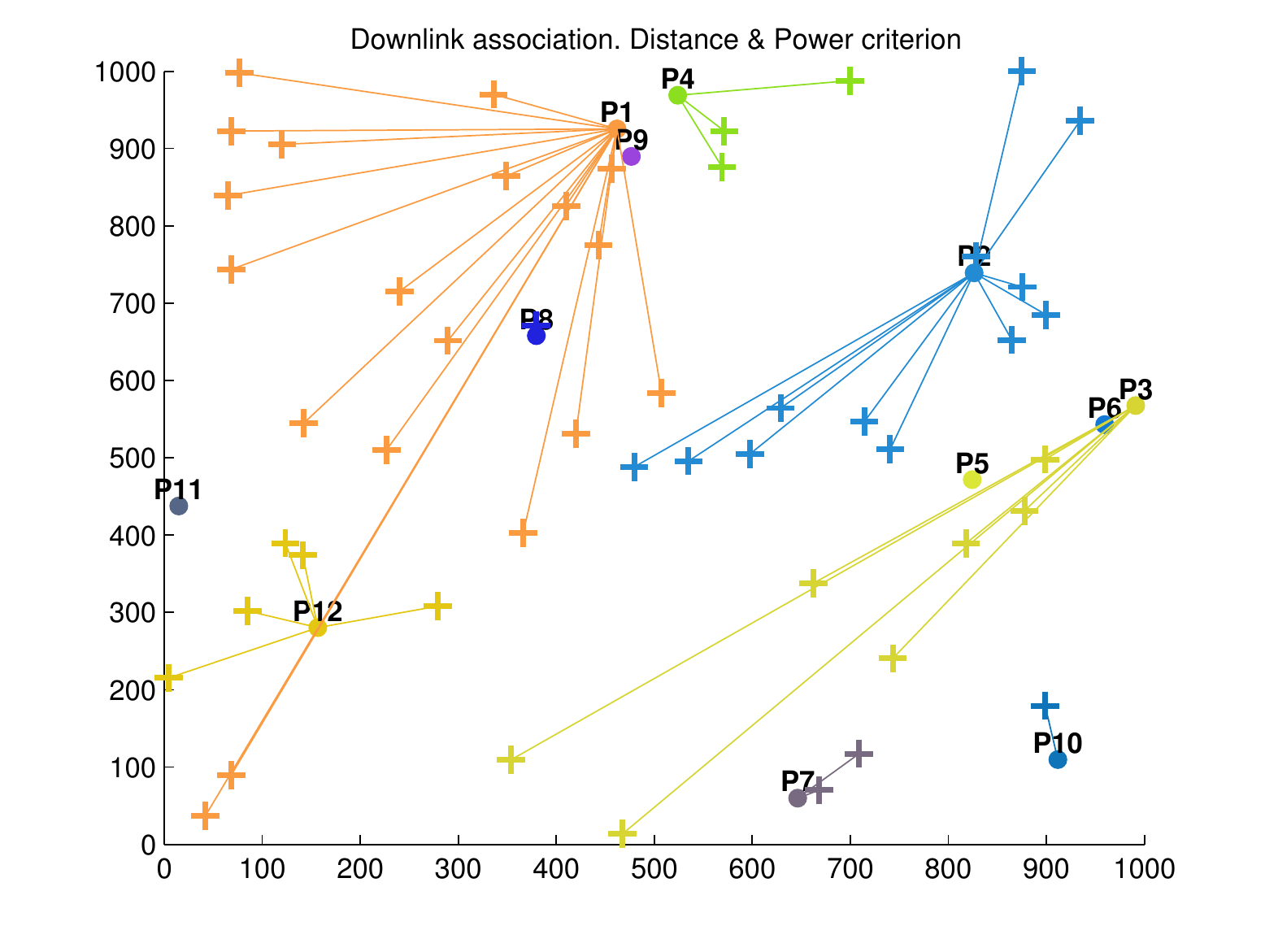}}
\subfigure[Uplink coverage.]{\label{fig:Final_UL_association_original}\includegraphics[scale=0.46]{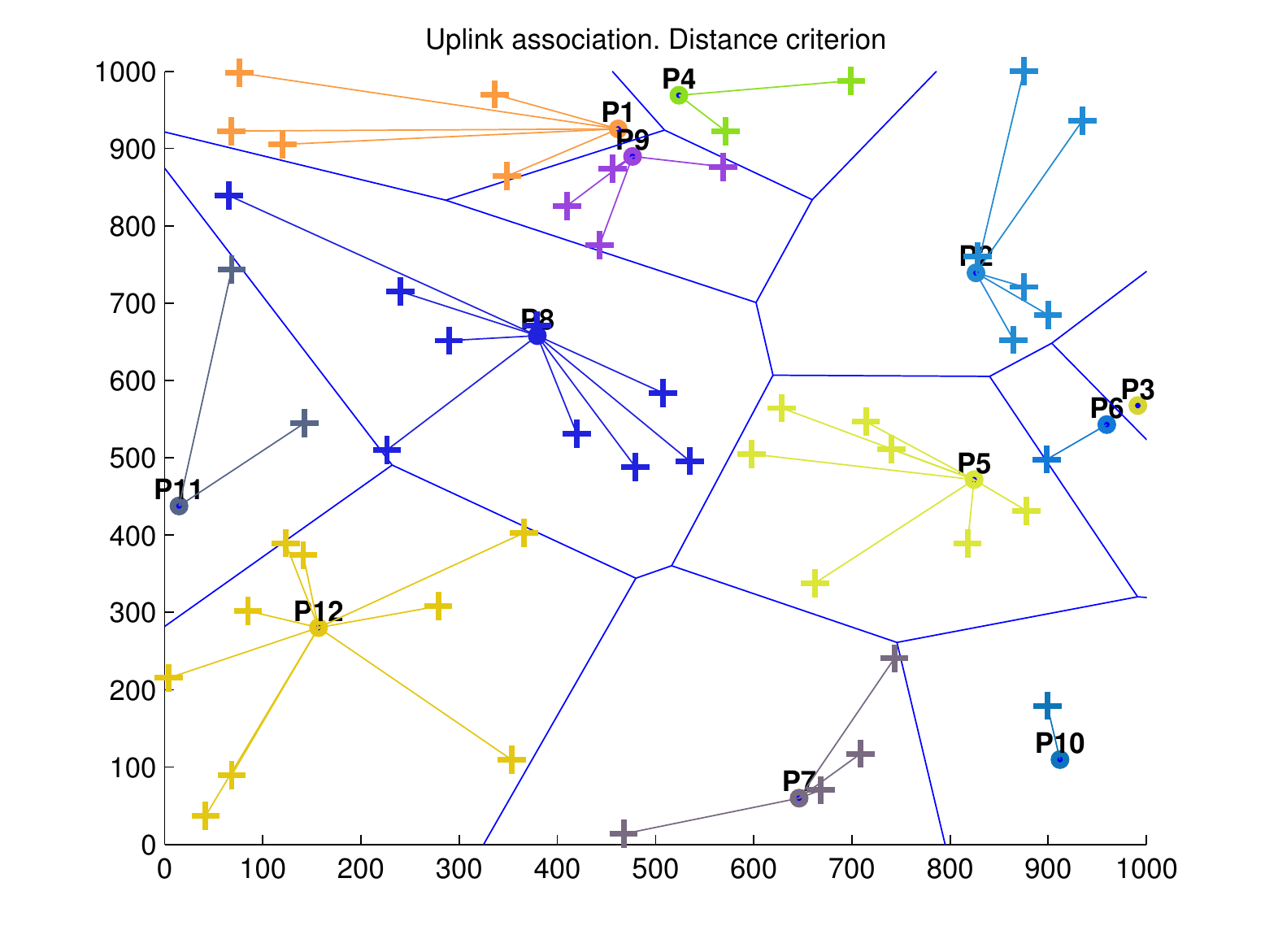}}
\caption{Downlink and uplink coverage maps.}
\label{fig:Final_assessment_original_association}
\end{figure}

Figure~\ref{fig:Final_assessment_joint_association} shows the effect of using the implemented gradient algorithm in user association. It should be noted that the original solution makes some base stations to be highly loaded while others remain almost idle (see \textsc{bs p}1 in Fig.~\ref{fig:Final_DL_association_original}). We can also observe this fact in the uplink. We should notice that the users appear to share base stations in a  more homogeneous and balanced way with the last algorithm, as numerically shown below. It is rather important to stress that we are not providing the association maps for the algorithm studied in Chapter~\ref{chap:num}, since they are the same as for the original approach. Recall that it allowed changing the resource allocation parameters for each user but the association decision was fixed and equal to that of the original \textsc{dud}e scheme.

\begin{figure}[!htb]
\centering     
\subfigure[Downlink coverage.]{\label{fig:Final_DL_association_joint}\includegraphics[scale=0.46]{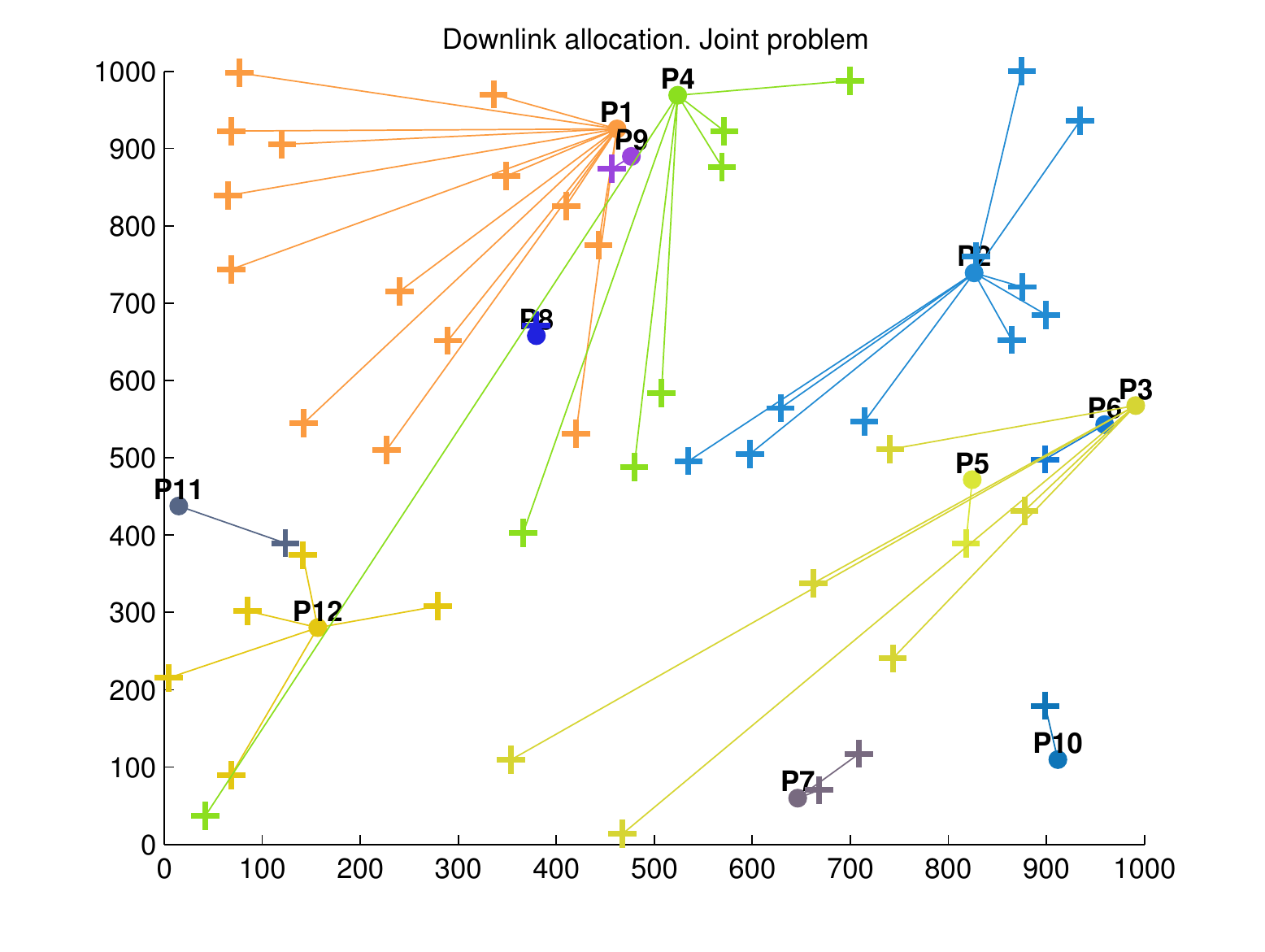}}
\subfigure[Uplink coverage.]{\label{fig:Final_UL_association_joint}\includegraphics[scale=0.46]{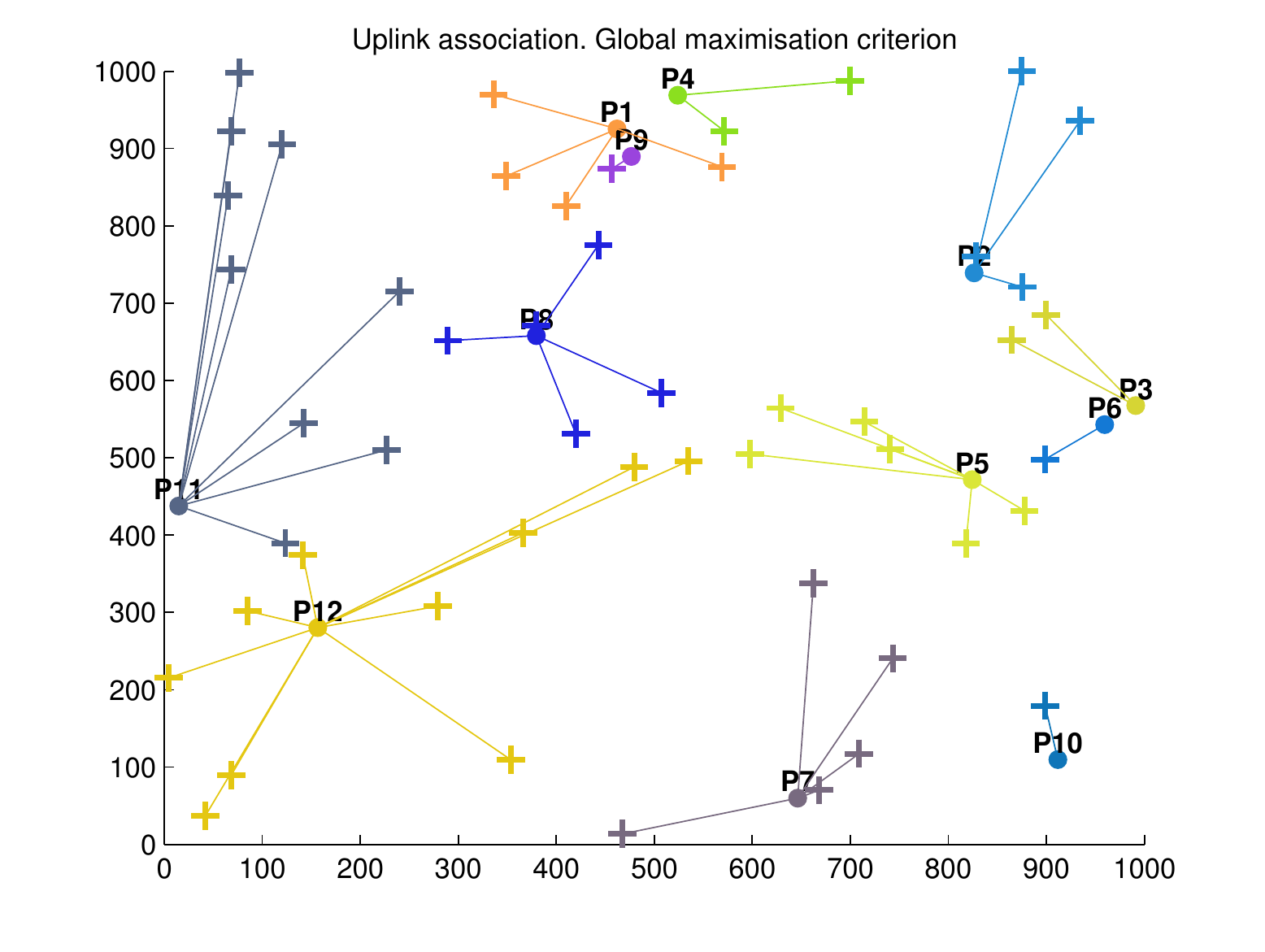}}
\caption{Downlink and uplink coverage maps.}
\label{fig:Final_assessment_joint_association}
\end{figure}

Prior to the discussion of the above-mentioned measurements, we include the plots showing the convergence of a subset of the base stations' multipliers when using the gradient-based algorithm (see. Figure~\ref{fig:Final_assessment_joint_multipliers}).

\begin{figure}[!htb]
\centering     
\subfigure[Downlink \textsc{bs}s' multipliers.]{\label{fig:Final_DL_Mult}\includegraphics[scale=0.46]{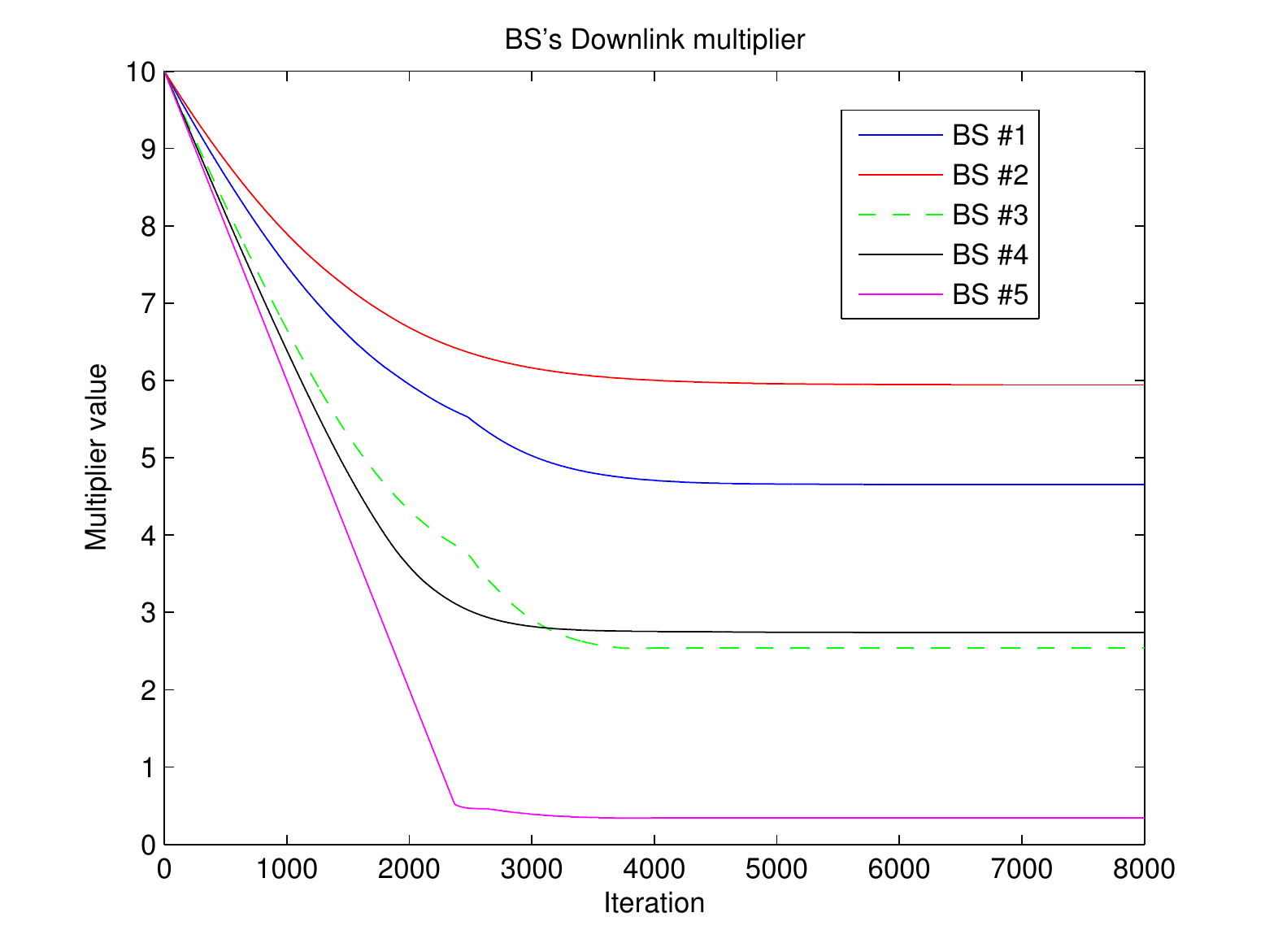}}
\subfigure[Uplink \textsc{bs}s' multipliers.]{\label{fig:Final_UL_Mult}\includegraphics[scale=0.46]{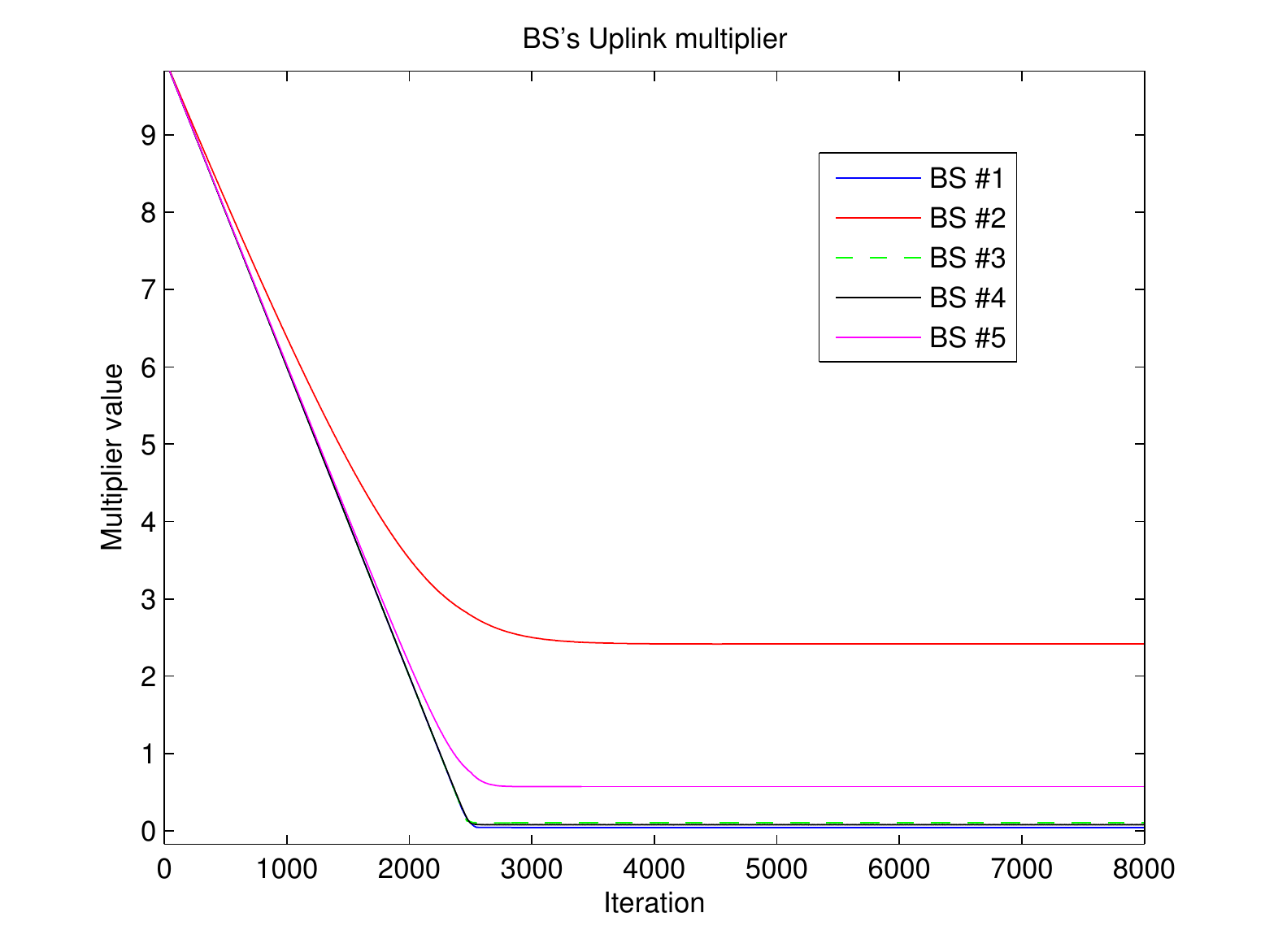}}
\caption{\textsc{bs}s' multipliers.}
\label{fig:Final_assessment_joint_multipliers}
\end{figure}

Regarding the performance indicators, we focus on three main measurementes. Namely, the aggregate capacity that the network is able to provide, the amount of work a given base station is expected to handle and the rate asymmetry a user may expect depending on the chosen approach. Figure~\ref{fig:Final_aggregate} compares the aggregate spectral efficiency of each one of the three schemes. The first two bars on the left show the overall performance of the network while using the original uplink-downlink decoupling policy. The second pair, suggest a slight improvement in both links as we are able to tweak the allocations for each user, making use of the algorithm studied in Chapter~\ref{chap:num}. Finally, the third pair of bars represents the performance of the network under the gradient-based global optimisation algorithm. It is worth noting the great gain that this scheme yields for both uplink and downlink.

\begin{figure}[!htb]
\centering     
\subfigure[Throughput aggregate.]{\label{fig:Final_aggregate}\includegraphics[scale=0.46]{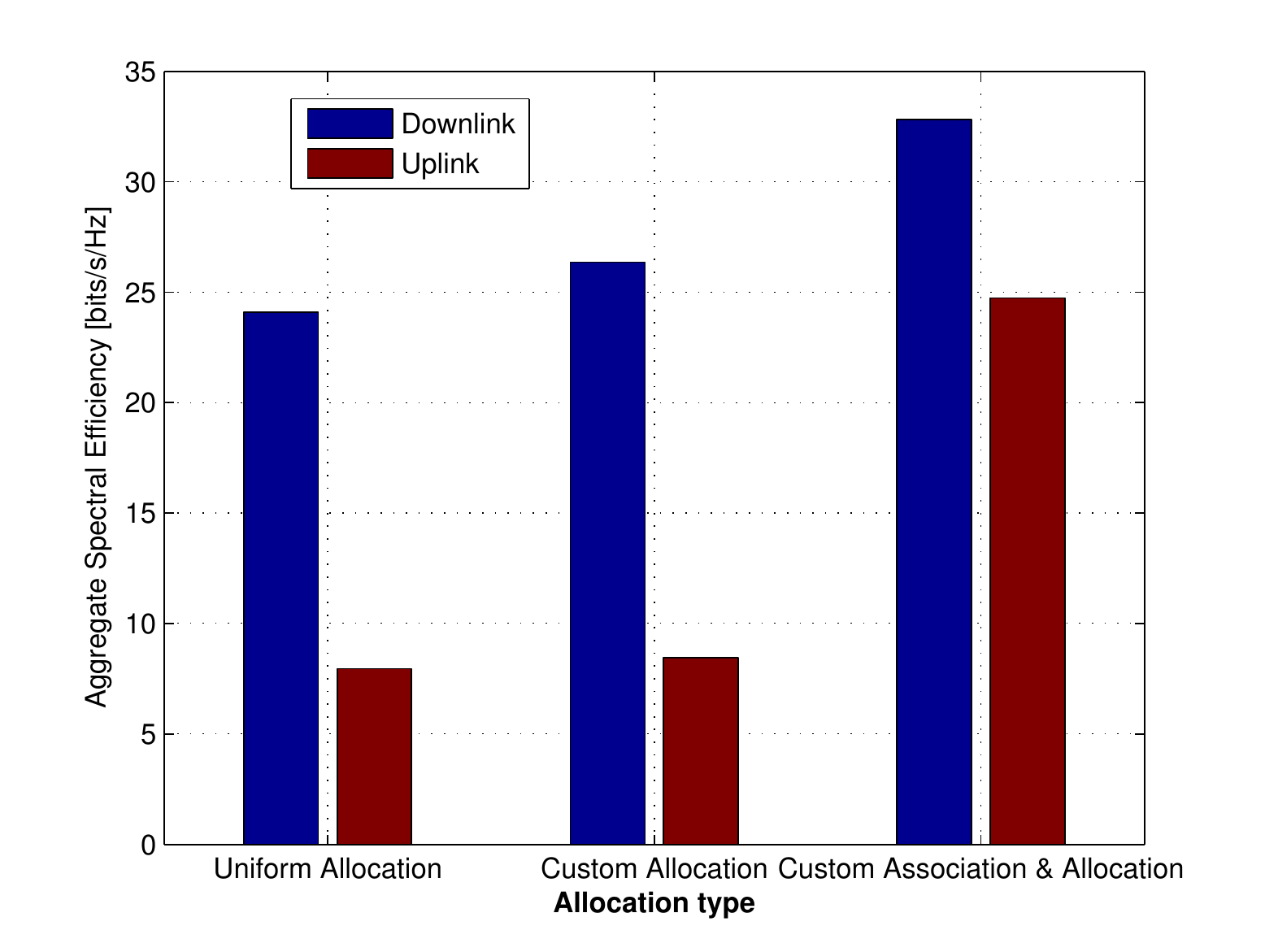}}
\subfigure[Mean user's rate asymmetry.]{\label{fig:Final_mean_asymmetry}\includegraphics[scale=0.46]{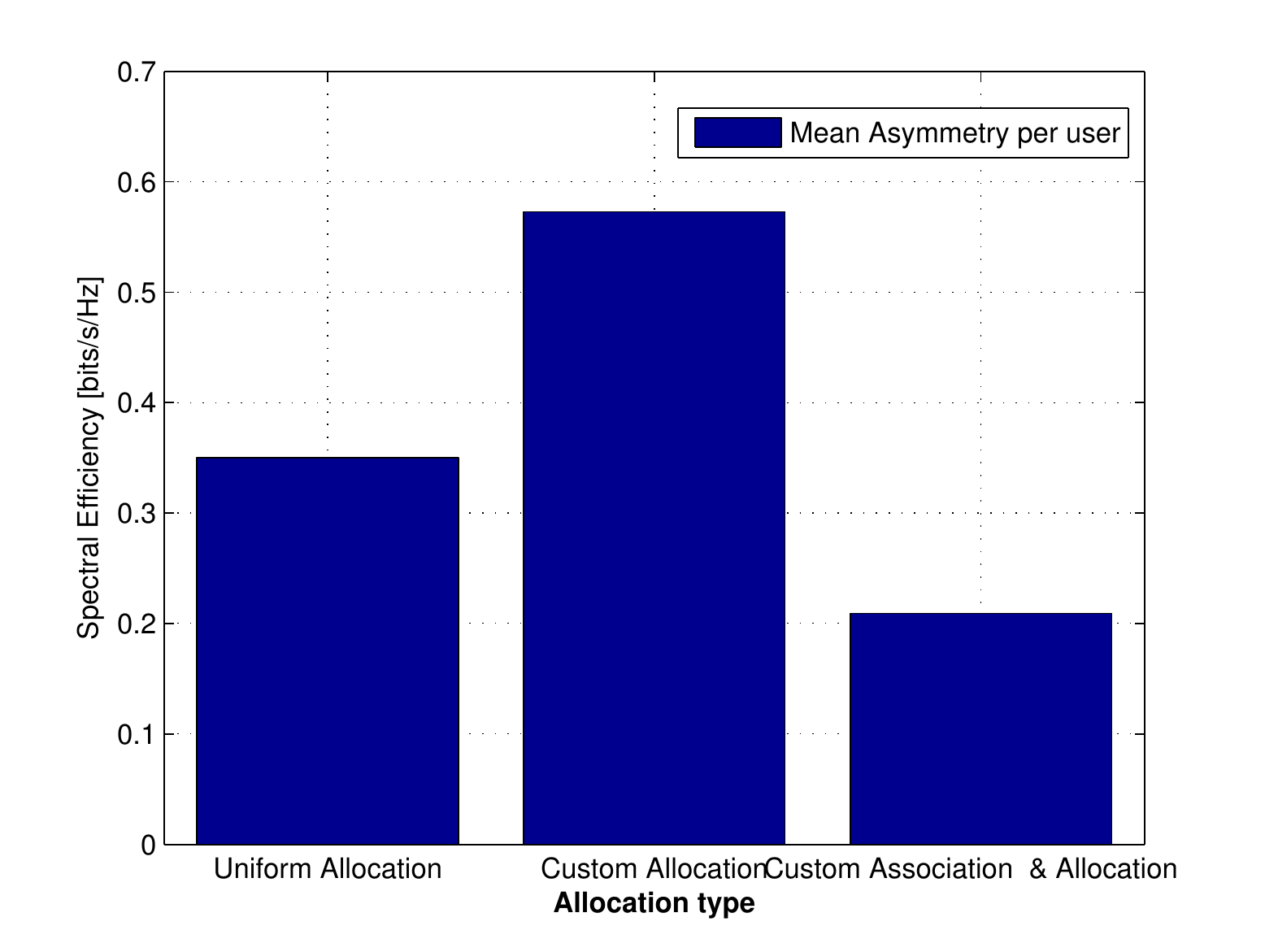}}
\caption{Performance of the algorithms.}
\label{fig:Final_assessment_joint_measurements1}
\end{figure}

In addition, we observe in Figure~\ref{fig:Final_mean_asymmetry}, how the third option also leads to an improvement in rate asymmetry. Conversely, the second algorithm is not able to achieve a throughput maximisation without affecting the rate symmetry, which seems reasonable since it is not allowed to alter the users' association policies. Finally, we investigate the consequences of each scheme concerning the workload of the base stations. Again, note that we only provide two different values: one for both the original decoupling scheme and Chapter~\ref{chap:num}'s algorithm and another one for the gradient-based association and allocation scheme. This is due to the fact that the first two share the same allocation decisions.

\begin{figure}[!htb]
  \begin{center}
    \includegraphics[scale=0.65]{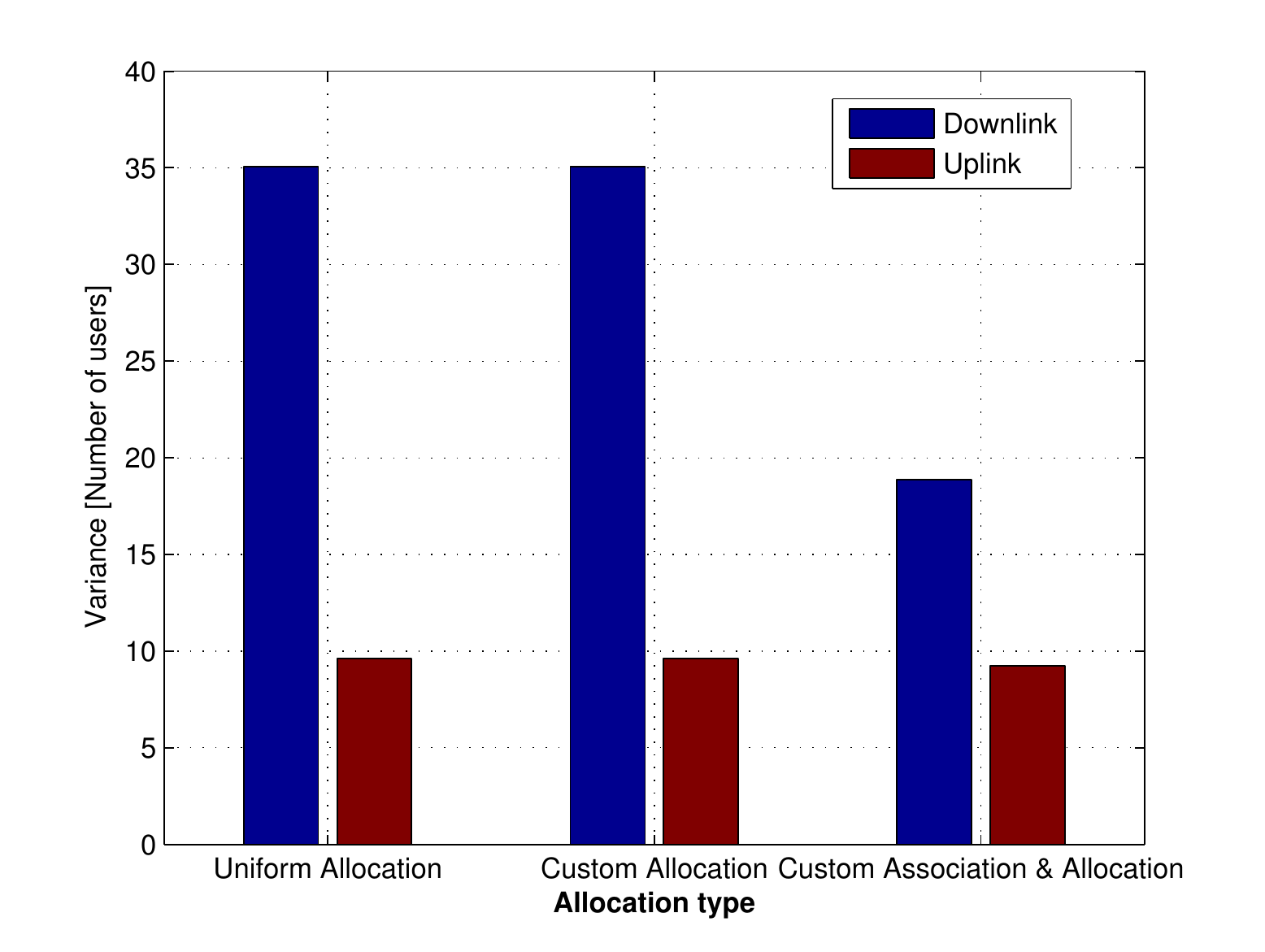}
    \caption{User load variance.}
    \label{fig:Final_user_variance}
  \end{center}
\end{figure}

Figure~\ref{fig:Final_user_variance} plots the variance of the number of users a base station is serving in a given link. It is worth emphasising the variance reduction in both uplink and downlink when employing the decentralised algorithm for utility maximisation. This means that the load a given base station has to deal with is more similiar to that of its neighbors and, therefore, the situations where few base stations serve most of the users and the majority remains almost idle are more unlikely to appear.
\chapter{Conclusions}

\section{Contributions}

We studied network utility from different perspectives. Firstly, we have assessed the performance gain obtained when using uplink-downlink decoupling on a 2-tier heterogeneous network deployment. A random spatial distribution approach has been used to compute both the position of base stations and user devices (\textsc{ppp}). These models, based on stochastic geometry, have shown their accuracy to model real-world network deployments and they represent the trending alternative to the traditional hexagonal grid deployments. With this, we try to provide a tool through which a network topology may be designed and tested. The pool of parameters that each of the algorithms devised in this work have, should help the operators design the best network deployment for each particular case.

Then, we study the network maximisation problem (\textsc{num}) within the context of single station association (\textsc{ssa}) policy, paying special attention to the symmetrical link balance. We separate the user's association decision from the resource allocation process and then devise an explicit solution when the association decision is fixed. A performance comparison between DUDe scheme and this solution is performed over different regimes of general $\alpha$-utility fairness function. We show that the implemented solution reduces the gap between uplink and downlink aggregates on the network and that different fairness and asymmetry parameters combinations may lead to diverse network behaviours. Furthermore, we establish the non-convexity of network utility maximisation problem under single station association policy when the association decision is not fixed and propose a new approach.

Finally, since the joint user association and resource allocation problem is not readily tractable following the \textsc{ssa} approach, we address this problem under Multi-Station association (\textsc{msa}) policy. Therefore, we relax the initial approach and allow each user to associate to more than one base station per link at the same time so as to retain convexity of the problem. Surprisingly, despite the fact that this approach is naturally more complex than the previous ones, it leads to a fast and scalable \emph{decentralised} solution via full dual decomposition of the global optimisation problem. We derive the simpler subproblems that each user and base station shall solve in order to reach the global optimum point and identify the computational complexity and message passing needs. The decomposition process revealed that the calculations which should be carried out by users and base stations are very simple and most of the message passing might be omitted and replaced by direct sensing of the physical channel conditions. In addition, the gradient projection method has been proven to converge in a reasonable amount of time, which make it suitable for real network deployments. Besides, we provide the numerical results which support the validity and the performance of the decentralised algorithm. We present how the algorithm behaves in different scenarios, highlighting the main features of the implemented solution. Regarding the distributed algorithm, a lot of flexibility is possible thanks to the parameters which are used to control its behaviour. A given user may activate/deactivate its own multipliers depending on its needs or even modify $\epsilon_u$ depending on the \textsc{q}o\textsc{s} that it is allowed to achieve at each moment.

To conclude, we conduct a comparison between the thee main options studied throughout this work. After the simulations, we note the great gain that the decentralised algorithm represents in terms of rate aggregate, base station offloading and per-user uplink-downlink rate symmetry.

\section{Future work}

Natural extensions to this work may include $i)$ extending the network simulator, $ii)$ studying the joint cell association and resource allocation problem under \textsc{ssa} policy making use of signomial geometric programming (\textsc{sgp}) and $iii)$ adding mobility support to the gradient-based decentralised algorithm. 

With reference to the first matter, the provided model can be easily extended to implement new physical layer technologies such as \textsc{mimo}, cell biasing, power control, etc. At the network level, device to device communications, scheduling and complex cooperation techniques between base stations can be included with minimal effort. Future work may also include point processes which model a minimum separation between points, i.e, \emph{Hard core point processes} (\textsc{hcpp}s). In that case, no two points of the process coexist with a separating distance less than a predefined hard core parameter. \emph{Poisson cluster processes} (\textsc{pcp}s), built from a parent \textsc{ppp} can also be useful to model the clustering behaviour observed on real cellular networks. The same discussion applies to user devices.

In view of the last point, increasing the base stations density and allowing heterogeneity with the presence of macro and small cells, poses new challenges for the mobility concerns which should be tackled in the near future in order to support continuous connectivity.

\appendix
\chapter{Appendix}

\section{The circle of Apollonius}\label{app:before}

\begin{proof}

\begin{itemize}
\item Let $\mathbb{X} = (x,y)$ be a point on a two dimensional grid which is equidistant to two given points ($P, Q$).
\item Let $P,\, Q$ be the two points under study. We want to compute the dominance area for each one of them. 
\item Let $\displaystyle d_w(a,\, b) = \frac{|a-b|}{W_b}$ be the definition of the weighted distance between two given points.
\item Let $W_p$ and $W_q$ be the weigth factors for each one of the points.

\end{itemize}

\begin{equation}
d\,(\mathbb{X},\,P) = d\,(\mathbb{X},\,Q)
\end{equation}

\begin{equation}
\frac{|\mathbb{X} - P|}{W_p} = \frac{|\mathbb{X} - Q|}{W_q} \longrightarrow \frac{|\mathbb{X} - P|}{|\mathbb{X} - Q|} = \frac{W_p}{W_q} = \lambda
\end{equation}

\begin{equation}
\frac{|(x,y) - (P_x, P_y)|}{|(x,y) - (Q_x, Q_y)|} = \lambda \longrightarrow \frac{\sqrt{(X-P_x)^2 + (y-P_y)^2}}{\sqrt{(x-Q_x)^2+(y-Q_y)^2}} = \lambda
\end{equation}

\begin{equation}
(x-P_x)^2 + (y-P_y)^2 = \lambda^2 [(x-Q_x)^2 + (y-Q_y)^2]
\end{equation}

\begin{equation}
\begin{aligned}
&x^2 - 2x P_x + P_x^2 - \lambda^2 x^2 + 2x Q_x \lambda^2 - Q_x^2\lambda^2  \\
&+ y^2 - 2y P_y + P_y^2 - \lambda^2 y^2 + 2y Q_y \lambda^2 - \lambda^2 Q_y^2 = 0
\end{aligned}
\end{equation}

Now, reworking the terms depending on $x$ yields

\begin{equation}
(1-\lambda^2)x^2 - 2x P_x + P_x^2 + 2x Q_x\lambda^2 - Q_x^2\lambda^2
\end{equation}

\begin{equation}
x^2 - 2x \frac{P_x}{(1-\lambda^2)} + \frac{P_x^2}{(1-\lambda^2)} + 2x \frac{Q_x\lambda^2}{(1-\lambda^2)} - \frac{Q_x^2\lambda^2}{(1-\lambda^2)}
\end{equation}

\begin{equation}
x^2 - 2x \frac{P_x - Q_x\lambda^2}{(1-\lambda^2)} + \frac{P_x^2 - Q_x^2\lambda^2}{(1-\lambda^2)}
\end{equation}

The above expression can be rewritten as

\begin{equation}
\bigl(x - \frac{P_x - Q_x\lambda^2}{(1-\lambda^2)}\bigr)^2 + \frac{P_x^2 - Q_x^2\lambda^2}{(1-\lambda^2)} - \frac{P_x^2 - 2P_xQ_x\lambda^2 + Q_x^2\lambda^4}{(1-\lambda^2)^2}
\end{equation}

\begin{equation}\label{eq:indep1}
\bigl(x - \frac{P_x - Q_x\lambda^2}{(1-\lambda^2)}\bigr)^2 + \frac{2P_xQ_x\lambda^2 - \lambda^2 P_x^2 - Q_x^2\lambda^2}{(1-\lambda^2)^2}
\end{equation}

We repeat the same procedure with the terms which depend on $y$

\begin{equation}
(1-\lambda^2)y^2 - 2yP_y + 2yQ_y\lambda^2 + P_y^2 - \lambda^2 Q_y^2
\end{equation}

\begin{equation}
y^2 - 2y \frac{P_y - Q_y\lambda^2}{(1-\lambda^2)} + \frac{P_y^2 - \lambda^2 Q_y^2}{(1-\lambda^2)}
\end{equation}

\begin{equation}
\bigl(y - \frac{P_y - Q_y\lambda^2}{(1-\lambda^2)}\bigr)^2 + \frac{P_y^2 - \lambda^2 Q_y^2}{(1-\lambda^2)} - \frac{P_y^2 - 2P_yQ_y\lambda^2 + Q_y^2\lambda^4}{(1-\lambda^2)^2}
\end{equation}

\begin{equation}\label{eq:indep2}
\bigl(y - \frac{P_y - Q_y\lambda^2}{(1-\lambda^2)}\bigr)^2 + \frac{2P_yQ_y\lambda^2 - P_y^2\lambda^2 - Q_y^2\lambda^2}{(1-\lambda^2)^2}
\end{equation}

In addition, combining the independent terms in \eqref{eq:indep1} and \eqref{eq:indep2} we obtain

\begin{equation}
\frac{-\lambda^2[P_x^2 - 2P_xQ_x + Q_x^2 + P_y^2 - 2P_yQ_y + Q_y^2]}{(1-\lambda^2)^2}
\end{equation}

\begin{equation}
\frac{-\lambda^2[(P_x - Q_x)^2 + (P_y - Q_y)^2]}{(1-\lambda^2)^2}
\end{equation}

Finally, the final equation that all the points must satisfy is

\begin{equation}
\bigl(x - \frac{P_x - Q_x\lambda^2}{(1-\lambda^2)}\bigr)^2 + \bigl(y - \frac{P_y - Q_y\lambda^2}{(1-\lambda^2)}\bigr)^2 = \frac{\lambda^2[(P_x - Q_x)^2 + (P_y - Q_y)^2]}{(1-\lambda^2)^2}.
\end{equation}

\end{proof}

\section{Convexity of $f_\alpha^{\textsc{ssa}}$}\label{app:a}

\begin{proof}
  The downlink/uplink rate for users $u$, $R_u(\mathbf{y}, \mathbf{z}) =
  \sum_b r_{ub} y_{ub} z_{ub}$ is linear in $\mathbf{y}$ for fixed
  $\mathbf{z}$. The $\alpha$-fair utility function is concave, so the
  composition $U_\alpha \bigl( R_\alpha(\mathbf{y}, \mathbf{z}) \bigr)$ is
  concave in $\mathbf{y}$. The difference $R_u(\mathbf{y}, \mathbf{z}) -
  R_u(\mathbf{y}^\prime, \mathbf{z}^\prime)$ is linear in $(\mathbf{y},
  \mathbf{y}^\prime)$ for fixed vectors $\mathbf{z}, \mathbf{z}^\prime$, and
  the absolute value $| \cdot |$ is a convex function. Hence, $-|
  R_u(\mathbf{y}, \mathbf{z}) - R_u(\mathbf{y}^\prime, \mathbf{z}^\prime)|$ is
  a concave function of $(\mathbf{y}, \mathbf{y}^\prime)$.
\end{proof}

\backmatter
\cleardoublepage
\phantomsection

\bibliographystyle{IEEEtran}
\bibliography{IEEEabrv,tail/thesis}
\addcontentsline{toc}{chapter}{Bibliography}

\end{document}